\documentclass{article}
\usepackage{amssymb}
\usepackage{amsmath}
\usepackage{amsthm}
\usepackage{fullpage}
\usepackage{graphicx}
\usepackage{setspace}
\usepackage{color}
\usepackage{natbib}

\usepackage{bbding}
\usepackage{tikz}
\usepackage{pgf}
\usepackage{pgfplots}

\usetikzlibrary{shapes}
\usetikzlibrary{arrows,decorations.pathmorphing,backgrounds,fit,positioning,
shapes.symbols,chains}

\definecolor{MyDarkBlue}{rgb}{0.1,0,0.40}
\definecolor{MyLightBlue}{rgb}{0, 229, 238}
\definecolor{MyMaize}{rgb}{255, 255, 0}
\definecolor{MyLemon}{rgb}{250, 250, 205}
\definecolor{MyDarkRed}{rgb}{178, 48, 96}

\tikzstyle{format} = [draw, thin, fill=blue!20]
\tikzstyle{medium} = [ellipse, draw, thin, fill=green!20, minimum height=2.5em]
\tikzstyle{rand} = [circle, draw=red, very thick, minimum height=2.5em]
\tikzstyle{cond} = [draw, thin]
\tikzstyle{mybox} = [draw=MyDarkBlue, very thick,
    rectangle, rounded corners, inner sep=3pt, inner ysep=10pt]
\tikzstyle{qbox} = [draw=MyDarkBlue, very thick, fill=blue!5,
    rectangle, rounded corners, inner sep=15pt, inner ysep=35pt]

\tikzstyle{sbox} = [draw=MyDarkBlue, very thick, fill=blue!5,
    rectangle, rounded corners, inner sep=1pt, inner ysep=12pt]

\tikzstyle{fancytitle} =[fill=MyDarkBlue, text=white]
\tikzstyle{taylor} = [rectangle, draw=MyDarkBlue, thick,
    minimum height=2.5em, inner sep=2pt, inner ysep=1pt]
\tikzstyle{plain} = [draw=none, fill=none]

\newenvironment{frcseries}{\fontfamily{frc}\selectfont}{}

\newcommand{\pn}{\mathbb{P}_n}

\newcommand{\pHat}{\hat{\mathbb{P}}_n^{(b)}}

\newcommand{\rtn}{\sqrt{n}}

\newcommand{\calD}{\mathcal{D}}

\newcommand{\bHat}{\hat{\beta}}
\newcommand{\bBoot}{\hat{\beta}^{(b)}}
\newcommand{\T}{\intercal}
\newcommand{\wh}[1]{\widehat#1}

\newcommand{\Cov}{\ensuremath{\mathrm{Cov}}}

\newtheorem{thm}{Theorem}[section]

\newtheorem{lem}[thm]{Lemma}
\newtheorem{cor}[thm]{Corollary}
\theoremstyle{definition}

\newtheorem*{slntrmrk*}{Remark}

\title{Statistical Inference in Dynamic Treatment Regimes}
\author{Eric B. Laber, Daniel J. Lizotte, William and  Susan A. Murphy}
\date{\today}

\begin{document}
 \begin{center}
 \textbf{Dynamic treatment regimes: technical challenges and
   applications}  \\
 \textbf{Eric B. Laber, Daniel J. Lizotte, Min Qian, William E. Pelham, \\
 and Susan A. Murphy\footnote{Eric B. Laber is in the Department of Statistics at
 North Carolina State University, 2311 Stinson Dr., Raleigh, NC, 27695 (E-mail
 \textit{laber@stat.ncsu.edu}).
 He acknowledges support from NIH grant P01 CA142538. Daniel J. Lizotte
 is in the Department of Computer Science at the University of
 Waterloo, Ontario, N2L~G31. He acknowledges support from the Natural
 Sciences and Engineering Research Council of Canada.
 Min Qian is in the Department of Biostatistics at Columbia University, New York City, NY, 10032.
 Susan A. Murphy is in the Departments of
 Statistics and Psychiatry at the University of Michigan, Ann Arbor, MI, 48109.
 She acknowledges support from NIMH grant R01-MH-080015 and NIDA grant
 P50-DA-010075.
  }}
\end{center}
\begin{abstract}
  Dynamic treatment regimes are of growing interest across the
  clinical sciences as these regimes provide one way to operationalize
  and thus inform sequential personalized clinical decision making.  A
  dynamic treatment regime is a sequence of decision rules, with a
  decision rule per stage of clinical intervention; each decision rule
  maps up-to-date patient information to a recommended treatment.  We
  briefly review a variety of approaches for using data to construct
  the decision rules. We then review an interesting challenge, that of
  nonregularity that often arises in this area.  By nonregularity, we
  mean the parameters indexing the optimal dynamic treatment regime
  are nonsmooth functionals of the underlying generative distribution.
  A consequence is that no regular or asymptotically unbiased
  estimator of these parameters exists.
  Nonregularity arises in inference for parameters in the optimal
  dynamic treatment regime; we illustrate the effect of nonregularity
  on asymptotic bias and via sensitivity of asymptotic, limiting,
  distributions to local perturbations.  We propose and evaluate a
  locally consistent Adaptive Confidence Interval (ACI) for the
  parameters of the optimal dynamic treatment regime.
  We use data from the Adaptive Interventions for Children with ADHD
  study as an illustrative example.  We conclude by highlighting and
  discussing emerging theoretical problems in this area.
\end{abstract}
\newpage
\section{Introduction}
Dynamic treatment regimes, also called treatment policies, adaptive
interventions or adaptive treatment strategies, were created to
inform the development of health-related interventions composed of
sequences of individualized treatment decisions.
 These regimes formalize
sequential individualized treatment decisions via a sequence of
decision rules that map dynamically evolving patient information to a
recommended treatment.  An optimal dynamic treatment regime (DTR) optimizes the expectation of a
desired cumulative outcome over a population of interest. 


The estimation of optimal DTRs presents a number of interesting
technical challenges and exciting open problems, one of which is
inference for nonregular parameters.  In particular, if an estimated
optimal DTR is to inform clinical decisions or guide future research,
it is essential to have reliable measures of uncertainty for the
estimated regime.  However, many of the most commonly used approaches
to estimating an optimal DTR involve estimation and inference for
parameters that are nonsmooth functionals of the underlying generative
distribution.
Consequently, estimators of these quantities
are necessarily nonregular and asymptotically biased
\citep{van1991differentiable, robins2004optimal, hirano}; standard
asymptotic approximations to the sampling distributions of these
estimators cannot be used directly to form reliable confidence
intervals or to carry out hypothesis testing.  The primary purpose
of this paper is to present the bias and other inferential problems related to this nonregularity
and offer potential
solutions for these problems in the context of DTR research.

In general the  data available for constructing  an optimal
DTR comes in the form of $n$ independent identically distributed
trajectories, one for each subject, of the form $(X_{1}, A_{1}, Y_{1},
\ldots, X_{T}, A_{T}, Y_{T})$ where: $X_t$ denotes interim subject
information collected during the course of the $t$th treatment; $A_t$
denotes the treatment received at time $t$; and $Y_t$ denotes an
outcome measured at the end of the $t$th treatment stage.
These trajectories may be collected in either a randomized ($A_t$ are
assigned with a known probability) or observational (the distribution
of $A_t$ is not known) study.  Traditionally most of the available
data for use in constructing DTRs has been observational and as a
result, causal inference issues dominate the discussion of statistical
methods,\ \cite{robins1986, hernan2000, murphyZThree, robins2004optimal,
  hernan2006, moodie, robinsetal2008, Orellana10, schulte}.
 However a growing number of
experimental studies, called Sequential, Multiple, Assignment
Randomized Trials (SMART) are being conducted
\citep[][]{lavori2000design, murphy2005experimental,
  inbalOne, lei2012smart}. These studies generally involve
two to three treatment stages ($T=2$ or $3$) and $A_t$ is randomized
at each stage.  See \cite{methCenterURL} for a partial list
of such studies.  To maintain the focus on the bias and other
inferential problems related to the nonregularity, we consider methods
for use with data collected in a sequential multiple assignment
randomized trial.

The Adaptive Pharmacological and Behavioral Treatments for Children
with ADHD Trial \citep[W. Pelham (PI);][]{inbalTwo, lei2012smart}
exemplifies the most common SMART; we use this study for illustration.
In the first stage of treatment, children are uniformly randomly
assigned to either a low dose of methylphenidate (a psychostimulant
drug) or a low intensity of behavioral modification therapy.
Beginning at 2 months and monthly thereafter (for the remainder of the
8 month study), each child is assessed for nonresponse; nonresponse
occurred if two different teacher ratings concerning the child's
school behavior fell below a prespecified criterion.  If nonresponse
occurs the child is re-randomized uniformly between two tactics:
intensify current treatment or augment the current treatment with the
other treatment (for example, augment methylphenidate with behavioral
modification therapy).  As long as the child did not meet the
criterion for nonresponse the child remained on current treatment.
See Figure \ref{fig:pelhamDiag} for a schematic of this trial.

\begin{figure}
\begin{center}
  \begin{tikzpicture}[thick, auto, scale=1, every node/.append style={transform shape}]

    \path[->] node[rand] (r1) {\color{red}{R}};

    \path[->] node[mybox, above right=1.25cm and .01cm of r1] (txt1A) 
        {%
          \scriptsize{Low Intensity BMOD}
        }
        (r1) edge node {} (txt1A.south west);
    \node[fancytitle, right=10pt] at (txt1A.north west) {
      \scriptsize{Treatment A}
      };

    \path[->] node[mybox, below right=1.25cm and .01cm of r1] (txt1B) 
        {%
          \scriptsize{Low Intensity MEDS}
       }
        (r1) edge node {} (txt1B.north west);
    \node[fancytitle, right=10pt] at (txt1B.north west) {
      \scriptsize{Treatment B}
      };
    
    \path[->] node[mybox, right=.75cm of txt1A] (ad1A) 
        {%
          \scriptsize{Response?}
        }
        (txt1A) edge node {} (ad1A);

    \path[->] node[mybox, right=.75cm of txt1B] (ad1B) 
        {%
          \scriptsize{Response?}
        }
        (txt1B) edge node {} (ad1B);

   \path[->] node[rand, below right=.80cm of ad1A] (r2A) {\color{red}{R}}
       (ad1A) edge node {No} (r2A);

   \path[->] node[rand, below right=.80cm of ad1B] (r2B) {\color{red}{R}}
       (ad1B) edge node {No} (r2B);

   \path[->] node[mybox, above right=1.00cm and 1.90cm of ad1A] (txt2A) 
        {%
          \scriptsize{Low Intensity BMOD}
        }
        (ad1A) edge node {Yes} (txt2A.south west);
    \node[fancytitle, right=10pt] at (txt2A.north west) {
      \scriptsize{Continue}
      };

  \path[->] node[mybox, above right=1.15cm and 0.6cm of r2A] (txt2AA) 
        {%
          \scriptsize{Augment with MEDS}
        }
        (r2A) edge node {} (txt2AA.south west);
    \node[fancytitle, right=10pt] at (txt2AA.north west) {
      \scriptsize{Treatment AA}
      };

  \path[->] node[mybox, above right=-.25cm and 0.60cm of r2A] (txt2AB) 
        {%
          \scriptsize{Intensify BMOD}
        }
        (r2A) edge node {} (txt2AB.south west);
    \node[fancytitle, right=10pt] at (txt2AB.north west) {
      \scriptsize{Treatment AB}
      };

  \path[->] node[mybox, above right=1.00cm and 1.90cm of ad1B] (txt2B) 
        {%
          \scriptsize{Low Intensity MEDS}
        }
        (ad1B) edge node {Yes} (txt2B.south west);
    \node[fancytitle, right=10pt] at (txt2B.north west) {
      \scriptsize{Continue}
      };

  \path[->] node[mybox, above right=1.15cm and 0.6cm of r2B] (txt2BA) 
        {%
          \scriptsize{Augment with BMOD}
        }
        (r2B) edge node {} (txt2BA.south west);
    \node[fancytitle, right=10pt] at (txt2BA.north west) {
      \scriptsize{Treatment BA}
      };

  \path[->] node[mybox, above right=-.25cm and 0.60cm of r2B] (txt2BB) 
        {%
          \scriptsize{Intensify MEDS}
        }
        (r2B) edge node {} (txt2BB.south west);
   \node[fancytitle, right=10pt] at (txt2BB.north west) {
     \scriptsize{Treatment BB}
     };
\end{tikzpicture}
\end{center}
\caption{Schematic describing the Adaptive
  Pharmacological and Behavioral Treatments for Children with ADHD
  SMART [W. Pelham (PI)].}\label{fig:pelhamDiag}
\end{figure}
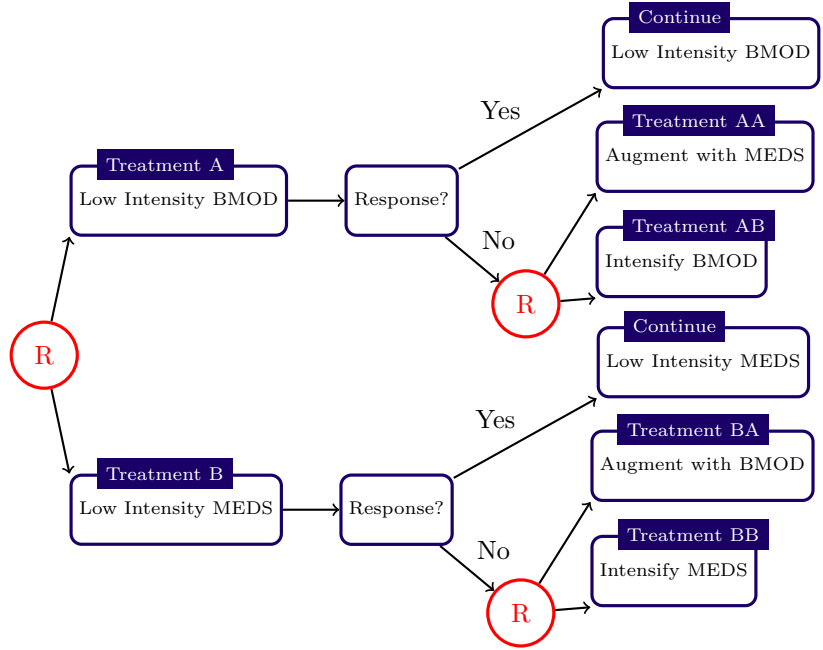


In Section 2 we briefly review different methods for constructing
optimal DTRs and provide greater detail for one such method,
$Q$-learning.  In Section 3 we discuss the problem of asymptotic bias
and show, using local alternatives, that bias-correcting shrinkage methods may perform infinitely
worse  than uncorrected methods.  In Section 4
we discuss interval estimation and propose a locally consistent
confidence interval for parameters indexing the optimal DTR.  In
Section 5 we examine the finite sample performance of the proposed
confidence interval using simulated data.  In section 6 we perform an
analysis of data from a clinical trial involving school-aged children
with ADHD.  We use this trial to illustrate open problems in model
selection and high-dimensional modeling for DTRs that arise even in
relatively simple settings.  Section 7 provides a general discussion
of some open problems relating to estimation and inference of DTRs.

\section{Review of Methods for Constructing Dynamic Treatment Regimes}

Throughout we consider the setting in which there are two stages of binary treatment;
this simple setting is sufficient for us to  illustrate the salient theoretical
challenges.  Furthermore many SMARTs including the ADHD study
described above involve two stages of binary treatment.
Recall that on each subject we
observe a time-ordered trajectory $(X_1, A_1, X_2, A_2, X_3)$.
The treatment $A_1$ is randomly assigned with probability possibly
depending on $X_1$ and $A_2$ is randomly assigned with probability
possibly depending on $(X_1,A_1,X_2)$. In the ADHD study both $A_1$
and $A_2$ are randomized with probability $1/2$ between the binary
alternatives.  $X_1$ denotes baseline (pre-randomization) subject
information; $A_1$ denotes in the initial treatment, coded to take
values in $\lbrace 0, 1\rbrace$; $X_2$ denotes subject information
collected during the course of the first treatment but prior to the
second treatment;$A_2$ denotes the second treatment, coded to take
values in $\lbrace 0,1\rbrace$; $X_3$ denotes subject information
collected during the course of the second treatment.  The outcomes,
$Y_1$ and $Y_2$ are summaries; $Y_1=y_1(X_1, A_1, X_2)$ and $Y_2 =
y_2(X_1, A_1, X_2, A_1, X_3)$ where $y_1$ and $y_2$ are known
functions.  Here we assume that both $Y_1$ and $Y_2$ are continuous
variables that are coded so that higher values are better.  Define
$Y\triangleq Y_1 + Y_2$ to be the total cumulative outcome.

In the ADHD study $X_1$ contains more than $25$ variables, some discrete
and some continuous, and $X_t,$ $t=2,3$ contains more than $40$
measurements collected each month; thus, over the course of the eight
month study the protocol dictated the collection of more then 360
measurements per subject.  In general $X_t,$ $t=1,2,3$ will contain a
large number of repeated measurements.  The current state-of-the-art
is that these measurements are summarized into low-dimensional
summaries motivated by clinical judgment, exploratory analyses and
convenience; this is certainly the case in the ADHD example.  An
important open problem is the development of formal feature extraction
and construction techniques for DTRs.  Here we assume that these
features are known.
Let $H_t,\, t=1,2$ denote a real-valued feature vector summarizing
information available to the decision maker at time $t$.  Thus, $H_1$
is a summary of information contained in $X_1$ and $H_2$ is summary of
information contained in $(X_0^{\T}, A_1, X_2^{\T})$.  In the ADHD
example, $H_1$ contains baseline ADHD severity, an indicator of
oppositional defiant disorder, and an indicator of prior exposure to
ADHD medication; $H_2$ contains $H_1$, as well as, an indicator of
adherence to initial treatment, and month of non-response to initial
treatment.

In this two stage setting, a DTR is a pair of decision rules $\pi = (\pi_1,
\pi_2)$, where $\pi_t:\mathrm{dom}(H_t) \rightarrow \mathrm{dom}(A_t)$
so that a patient presenting at time $t$ with $H_t = h_t$ is assigned
treatment $\pi_t(h_t)$.  The value of a DTR $\pi$, denoted
$\mathbb{E}^{\pi}Y$, is the expected outcome under the restriction
that $A_t = \pi_t(H_t)$.  The optimal DTR, say
$\pi^{\mathrm{opt}},$ satisfies $\mathbb{E}^{\pi^{\mathrm{opt}}}Y =
\sup_{\pi}\mathbb{E}^{\pi}Y$.

Methods for estimating optimal DTRs from data can be broadly
classified as either indirect or direct estimation
methods~\citep{barto2004}.  Indirect estimation methods use
approximate dynamic programming with parametric, semiparametric or
nonparametric methods to first estimate a series of outcome models and
then from these models infer the optimal DTR.  $Q$-learning
\citep{murphyZfive, bibhasBook, minChapter, bibhasChapter},
$A$-learning \citep{murphyZThree, robins2004optimal},
regret-regression \citep{henderson2009regret} are popular indirect
methods in the statistical literature.  We provide a detailed discussion 
of $Q$-learning
below.

Direct estimation methods, also known as policy search methods,
maximize an estimator of the expected cumulative outcome over DTRs in
a pre-specified class.  Recent statistical work in this area includes
marginal structural mean models \citep{robinsetal2008, Orellana10},
augmented value maximization \citep{baqun, baqun2}, and outcome
weighted learning \citep{yingqi, yingqi2}.

One potential advantage of indirect methods is that the requisite
outcome models can be built using standard statistical models
(generalized regression models, time series models, etc.) which can be
checked for goodness of fit.  This is particularly attractive when
scientific theory, expert opinion can be used in forming the outcome
model. A potential drawback is that the optimal DTR is indirectly
inferred from the outcome models rather than being estimated directly.
In contrast, most direct estimation methods do not or minimally utilize
outcome models and thereby are robust to model misspecification.
However, direct estimation methods generally produce
estimators of the parameters (in an DTR) with higher variance than
indirect estimation methods.  This fact has been recognized for some
time in the computer science literature with efforts there focused on
using outcome models in combination with direct methods so as to
reduce variance \citep[][]{sutton1999, konda2003}.
%
Indeed there is a vast literature concerning both indirect and direct
methods for constructing optimal policies, (i.e., dynamic treatment
regimes) in the field of reinforcement learning with many good
introductory books \citep[][]{sutton, si2004handbook, busoniu2010,
  csaba2010, wiering2012}.  However the focus of this work is on
algorithms for estimation; inference, e.g., confidence intervals or
test statistics, that can be used in discussing the level of
confidence concerning the constructed DTR with clinical scientists,
are, to our knowledge, absent.
%
%
%


To illustrate and discuss inferential challenges, we consider
estimators constructed using $Q$-learning.
Q-Learning is attractive to statistical practitioners because
Q-Learning can be viewed as a multi-stage extension of regression
\citep[][]{inbalTwo}, thus enabling much of the intuition developed in
that area to be (somewhat) easily translated to the area of DTRs.
Q-Learning is an indirect method of constructing a DTR from data; in
the appendix \ref{ap:outcomewt}, we illustrate review a direct method,
outcome-weighted learning, and illustrate that the use of this method
poses the same inferential challenges as $Q$-Learning. The problems we
identify with Q-Learning apply to many of the aforementioned
estimators.
%
%

Define the $Q$-functions \citep[][]{sutton, murphyZfive}
as
\begin{eqnarray}
  Q_2(h_2, a_2) &\triangleq& \mathbb{E}(Y|H_2=h_2, A_2=a_2), \nonumber \\
  Q_1(h_1, a_1) &\triangleq& \mathbb{E}\left(
    \max_{a_2}Q_2(H_2, a_2)\big|H_1=h_1, A_1=a_1
\right),\label{q1Def}
\end{eqnarray}
so that $Q_2(h_2, a_2)$ measures the quality of assigning treatment
$a_2$ to a patient presenting with $h_2$  at the second stage, and
$Q_1(h_1, a_1)$ measures the quality of assigning treatment $a_1$ to
a patient presenting with $h_1$ at baseline assuming optimal
treatment selection at the second stage.  If the $Q$-functions
are known, then the optimal DTR is given by
the dynamic programming solution, $\pi_{t}^{\mathrm{dp}}(h_t) =
\arg\max_{a_t}Q_t(h_t, a_t)$ \citep[][]{bellman}.

Note that $\pi_t^{\mathrm{dp}}(h_t) = 1_{Q_t(h_t, 1) - Q_t(h_t, 0) \ge
  0}$ (recall that $a_t\in\{0,1\}$).  $Q$-learning provides estimators
of the $Q$-contrasts, $Q_t(h_t, 1) - Q_t(h_t, 0)$.
Owing to the max-operator in (\ref{q1Def}),
$Q_1$ is a nonsmooth functional of the underlying generative
distribution, hence the estimand 
is also nonsmooth.
We next illustrate how this nonsmoothness impacts the sampling
distributions of DTR estimators using $Q$-learning.

\subsection{$Q$-Learning}
$Q$-learning estimates the optimal DTR by postulating regression
models for the $Q$-functions and then taking the plug-in dynamic
programming solution.  Consider linear models for the $Q$-functions of
the form $Q_t(h_t, a_t;\beta_t) = h_{t,0}^{\T}\beta_{t,0} + a_t
h_{t,1}^{\T}\beta_{t,1}$ where $h_{t,0}$ and $h_{t,1}$ are known
feature vectors constructed from $h_t$ and $\beta_t =
(\beta_{t,0}^{\T}, \beta_{t,1}^{\T})^{\T}$; these feature vectors
might contain splines or other nonlinear basis expansions.  Recall
that an open problem in DTR research is the development of a
principled feature construction method.  The above linear model
highlights a crucial difference between usual goal of constructing
features for prediction and constructing features for decision making.
To see this note that from the linear model for the $Q$-function, only
the features $h_{t,1}$ will be used by the decision rule
$\pi_t^{\mathrm{dp}}$.  Thus high quality features for decision making
(as opposed to prediction) should interact with the treatment $a_t$
sufficiently strongly so that the $\pi_t^{\mathrm{dp}}(h_{t})$ varies
by $h_{t,1}$.  At this time research focused on discovering features
for decision making has been in the one-step setting
\citep[see][]{Gunter, foster2011, dusseldorp2013qualitative, holly}; 
the multistage setting is essentially open.




The
parameters indexing the $Q$-functions are estimated using least squares.
Let $\pn$ denote empirical expectation, for example
$\pn f(Z) = n^{-1}\sum_{i=1}^{n}f(Z_i)$ where $\{Z_i\}_{i=1}^n$ is a random sample.
One version
of the $Q$-learning  algorithm is as follows.
\begin{enumerate}
  \item Stage 2 regression: $\wh{\beta}_{2} =
\arg\min_{\beta_{2}} \pn \left(
  Y_2 - Q_2(H_2, A_2;\beta_2)
\right)^2$.
 \item Predicted second stage outcome: $\widetilde{Y} =
   Y_1 + \max_{a_2}Q_2(H_2, a_2;\wh{\beta}_{2})$.
 \item Stage 1 regression: $\wh{\beta}_{1} = \arg\min_{\beta_1}\pn
   \left(
     \widetilde{Y} - Q_1(H_1, A_1;\beta_1)
     \right)^2.$
\end{enumerate}
The $Q$-learning estimator of the optimal DTR is thus
$\wh{\pi}_{t}(h_t) = \arg\max_{a_t}Q_t(h_t, a_t;\wh{\beta}_{t})$.
The second stage coefficients $\wh{\beta}_{2}$ are ordinary least
squares estimators and are thus regular and asymptotically normal
under mild conditions (see Section 4).  However, the first stage
coefficients depend on the maximized second stage $Q$-function;
because the max operator is nonsmooth, the estimated coefficients
$\wh{\beta}_{1}$ are in turn a nonsmooth function of the data.

For notational simplicity from here until Section 6,
 $Y_1 \equiv 0$ so that $Y = Y_2$, and thus
we will omit any subscripts on $Y$.
Define the following population analogs of the estimators used
in $Q$-learning:
\begin{eqnarray*}
  \beta_2^* &\triangleq& \arg\min_{\beta_2}P\left(
    Y - Q_2(H_2, A_2;\beta_2)
\right)^2, \\
 \widetilde{Y}^* &\triangleq & \max_{a_2}Q_2(H_2, a_2;\beta_2^*)= H_{2,0}^{\T}\beta_{2,0}^* +
\left[H_{2,1}^{\T}\beta_{2,1}^*\right]_+,  \\
\beta_1^* &\triangleq& \arg\min_{\beta_1}P\left(
\widetilde{Y}^* - Q_1(H_1, A_1;\beta_1)
\right)^2,
\end{eqnarray*}
where $P$ denotes expectation with respect to the distribution of
$(X_0, A_1, X_1, Y_1, A_2, X_2, Y_2)$ and the second line follows from the fact that $a_2\in\{0,1\}$.  In addition, define $B_t \triangleq
(H_{t,0}^{\T}, A_t H_{t,1}^T)^{\T}$, $\Sigma_{t,\infty} \triangleq
PB_tB_t^{\T}$ for $t=1,2$, and $\widehat{\Sigma}_{t} \triangleq \pn
B_t B_t^{\T}$.  We assume $\widehat{\Sigma}_{t}$ is invertible.  Then
$\widehat{\beta}_{1} = \widehat{\Sigma}_{1}^{-1}\pn B_1
\widetilde{Y}$, $\beta_1^* =
\Sigma_{1,\infty}^{-1}PB_1\widetilde{Y}^*$ so that
$\rtn(\widehat{\beta}_{1} - \beta_1^*) = \wh{\Sigma}_{1}^{-1}\rtn\pn
B_1(\widetilde{Y} - B_1^{\T}\beta_1^*)$.  It is useful to decompose
$\wh{\Sigma}_{1}^{-1}\rtn\pn B_1(\widetilde{Y} - B_1^{\T}\beta_1^*)$
as
\begin{equation}\label{coeffDecomp} \mathbb{S}_{n} +
\wh{\Sigma}_{1}^{-1}\pn B_1 \mathbb{U}_{n},
\end{equation} where
\begin{eqnarray*} \mathbb{S}_n &=& \hat{\Sigma}_{1}^{-1}\rtn \pn
B_1\bigg[ \left(H_{2,0}^{\T}\beta_{2,0}^* +
\left[H_{2,1}^{\T}\beta_{2,1}^*\right]_+ - B_1^{\T}\beta_1^* \right) +
H_{2,0}^{\T}\left(\bHat_{2,0} - \beta_{2,0}^*\right)\bigg], \\
\mathbb{U}_n &=& \rtn\left(\big[H_{2,1}^{\T} \bHat_{2,1}\big]_+ -
\left[H_{2,1}^{\T}\beta_{2,1}^*\right]_+\right).
\end{eqnarray*}
The term $\mathbb{S}_{n}$ is smooth and asymptotically normal but
$\mathbb{U}_{n}$ is nonsmooth in $\widehat{\beta}_{2,1}$.
To understand the implications of this
nonsmoothness, fix $H_{2,1} = h_{2,1}$.  If
$h_{2,1}^{\T}\beta_{2,1}^*\ne 0$, then
$\mathbb{U}_{n}\big|_{H_{2,1}=h_{2,1}}$ is asymptotically normal with
mean zero.  However, if $h_{2,1}^{\T}\beta_{2,1}^* = 0$ then
$\mathbb{U}_{n}\big|_{H_{2,1}=h_{2,1}} =
\left[h_{2,1}^{\T}\rtn(\wh{\beta}_{2,1} -\beta_{2,1}^*)\right]_{+}$
which converges to the positive part of a mean zero normal random
variable.  Thus, the limiting distribution of $\rtn(\wh{\beta}_{1} -
\beta_1^*)$ depends abruptly on the value of $\beta_{2,1}^*$ and the
distribution of $H_{2,1}$.  This abruptness signals nonregular inference.

If $H_{2,1}$ is composed only of continuous variables then some
sceptism is natural because $P[H_{2,1}^{\T}\beta_{2,1}^* = 0]=0$.
However in most clinical trials, the effect of treatment can be
expected to be small ($H_{2,1}^{\T}\beta_{2,1}^*$ is the effect of
stage 2 treatment) relative to the noise level, thus even though
$H_{2,1}^{\T}\beta_{2,1}^*$ may not be $0$, it's estimator can be
expected to be near $0$ with high probability.  And as we shall see
that the limiting distribution of $\rtn(\wh{\beta}_{1} - \beta_1^*)$
depends abruptly on the value of $\beta_{2,1}^*$ and the distribution
of $H_{2,1}$ indicates that the small sample behavior of
$\rtn(\wh{\beta}_{1} - \beta_1^*)$ is poorly approximated by
fixed-parameter asymptotic results that assume
$P[H_{2,1}^{\T}\beta_{2,1}^* = 0]=0$ (see discussion of bias in
Section 3 and evaluation of confidence intervals in Section 5).
Moving-parameter (e.g., local ) asymptotic results provide a better
reflection of small sample behavior and are provided in the Sections 3
and 4.

\section{Asymptotic bias}
In the study of nonregular estimators, much attention has been given
to asymptotic bias, characterized here as bias that is $O(1/\rtn)$.
Since asymptotic bias may be indicative of bias in small samples,
incorrect Type I error levels in hypothesis testing, and poor coverage
rates of confidence intervals \citep[e.g.,][]{blume1968, casella1981,
  bickel1981, robins2004optimal, marchand2004,
  chakraborty2009inference, moodieT}, there is great interest in
characterizing and reducing asymptotic bias.  Here we: (i)
characterize the asymptotic bias of the first stage $Q$-learning
estimator; (ii) show that the asymptotic bias can be reduced by using
a shrinkage estimator; and (iii) argue that shrinking too aggressively
can lead to arbitrarily bad performance in finite samples.

We use $\mathbb{E}$ to denote expectation  over $P$ (the distribution
of the observed data).  Let $c\in\mathbb{R}^{\dim(\beta_1^*)}$ be
fixed.  For any $\rtn$-consistent estimator $\tilde{\beta}_1$ of
$\beta_1^*$ with $\rtn(\tilde{\beta}_1-\beta_1^*)$ converging in
distribution to $\mathbb{M}$, define the $c$-directional asymptotic
bias of $\tilde{\beta}_1$ as
\begin{equation*}
\mathrm{Bias}(\tilde{\beta}_{1}, c) \triangleq
\mathbb{E}c^{\T}\mathbb{M}.
\end{equation*}
Define
\begin{eqnarray*}
g_{2}(B_2, Y;\beta_2^*) &\triangleq& B_2(Y - B_2^{\T}\beta_2^*), \\
g_{1}(B_1, H_2;\beta_1^*, \beta_2^*) &\triangleq&
B_1\left(H_{2,0}^{\T}\beta_{2,0}^*
 + \left[H_{2,1}^{\T}\beta_{2,1}^*\right]_{+} -
 B_{1}^{\T}\beta_1^*\right).
\end{eqnarray*}
Throughout we assume:
\begin{itemize}
\item[(A1)] The histories $H_2$, features $B_1$, and outcomes $Y$,
satisfy the moment inequalities\\  $P||H_{2}||^2\,||B_1||^2 < \infty$
 and $PY^2||B_2||^2 < \infty$.
\item[(A2)] The matrices $\Sigma_{t,\infty}$ and $
\mathrm{Cov}\,(g_1, g_2)$ are strictly positive definite.
\end{itemize}\noindent
Assumptions (A1)-(A2) are quite mild, requiring only full rank design
matrices and some moment conditions.
Using standard methods it can be shown that $\mathbb{V}_{n} \triangleq
\rtn(\bHat_{2} - \beta_{2}^*)$ is asymptotically normal with mean zero
and variance-covariance $\Omega =
(PB_2B_2^{\T})^{-1}PB_2B_2^{\T}(Y-B_2^{\T}\beta_2^*)^2(PB_2B_2^{\T})^{-1}$.
Let $\Sigma_{21,21}$
denote the submatrix of $\Omega$ corresponding the limiting
asymptotic covariance of $\rtn(\bHat_{2,1} - \beta_{2,1}^*)$ and
$\hat{\Sigma}_{21,21}$ the corresponding plug-in estimator.
The following result is proved
in Appendix~\ref{ap:proofs}.
\begin{thm}\label{QBias}
  Assume (A1) and (A2) and let $c\in\mathbb{R}^{\dim(\beta_1^*)}$ be fixed.  Then:
  \begin{equation*}
    \mathrm{Bias}(\wh{\beta}_{1}, c) =
    \frac{c^{\T}\Sigma_{1,\infty}^{-1}P\left[B_1\sqrt{
        H_{2,1}^{\T}\Sigma_{21,21} H_{2,1}}1_{H_{2,1}^{\T}\beta_{2,1}^* = 0}\right]}{\sqrt{2\pi}}.
  \end{equation*}
\end{thm}\noindent
The asymptotic bias of $Q$-learning is nonzero when the
second stage treatment effect, $H_{2,1}^{\T}\beta_{2,1}^*$, satisfies
$P(H_{2,1}^{\T}\beta_{2,1}^* = 0) > 0$.

A common strategy for reducing asymptotic bias in $Q$-learning is
to shrink the predicted outcome $\widetilde{Y}$.
 \cite{moodieT} proposed a hard-thresholding approach;
 \cite{chakraborty2009inference}  proposed a soft-thresholding
 estimator;  and more recently \cite{song} proposed a penalized
 version of $Q$-learning.
We use the soft-thresholding estimator
 proposed by \cite{chakraborty2009inference} as an illustrative example.
\cite{chakraborty2009inference} illustrate, using simulation studies,
that soft-thresholding reduces bias in small samples.
Define
\begin{equation}\label{bibhasEst}
  \widetilde{Y}^{\sigma} \triangleq \wh{\beta}_{2,0}^{\T}H_{2,0} + \left[
H_{2,1}^{\T}\wh{\beta}_{2,1}
\right]_{+}\left(
1 - \frac{\sigma H_{2,1}^{\T}\wh{\Sigma}_{21,21}H_{2,1}}{n(\wh{\beta}_{2,1}^{\T}H_{2,1})^2}
\right)_{+},
\end{equation}
where $\sigma$ is nonnegative constant.  For positive values of
$\sigma$, the soft-thresholding estimator shrinks the nonsmooth
part of the predicted outcome towards zero.  The first stage
soft-thresholding estimators are given by
\begin{equation*}
\wh{\beta}_1^{\sigma} \triangleq \arg\min_{\beta_1}\pn\left(
\widetilde{Y}^{\sigma} - Q_1(H_1, A_1;\beta_1)
\right)^2.
\end{equation*}
The following result
is proved in Appendix~\ref{ap:proofs}.
\begin{thm}
Assume (A1) and (A2) and let $c\in\mathbb{R}^{p_1}$ be fixed.  Then:
\begin{enumerate}
  \item $\big|\mathrm{Bias}(\wh{\beta}_1^{\sigma}, c)\big| \le
    \big|\mathrm{Bias}(\wh{\beta}_{1},c)\big|$ for any $\sigma \ge 0$.
  \item If $\mathrm{Bias}(\wh{\beta}_{1}, c) \ne 0$ then for $\sigma > 0$
    \begin{equation*}
      \frac{\mathrm{Bias}(\wh{\beta}_{1}^{\sigma},c)}
      {\mathrm{Bias}(\wh{\beta}_{1},c)} =
        \exp\lbrace-\sigma/2\rbrace - \sigma\int_{\sqrt{\sigma}}^{\infty}
        \frac{1}{x}\exp\lbrace-x^2/2\rbrace \mathrm{d}x.
    \end{equation*}
\end{enumerate}
\end{thm}\noindent
\cite{chakraborty2009inference} recommend $\sigma=3$ which corresponds
to an approximate empirical Bayes estimator; plugging $\sigma=3$ into
the above expression shows an approximate 13-fold reduction in asymptotic
bias.
The soft-thresholding estimator has smaller asymptotic bias than
$Q$-learning and the preceding result seems to suggest that larger
values of $\sigma$ are preferred; indeed if $\sigma \rightarrow
\infty$ the asymptotic bias of the soft-thresholding estimator
converges to zero.   These results are point-wise in the
parameter space for $(\beta_1,\beta_2)$; that is for any fixed true
parameter value of $(\beta_1,\beta_2)$ the asymptotic bias converges
to zero.

While  it appears that these methods reduce asymptotic bias it is known that the methods cannot  completely
remove the asymptotic bias without driving the mean squared
error to infinity  \cite[see, for example,][]{sethuramanDoss,
  brownLiu, chen2004}.   Furthermore, even considering just the bias, if we evaluate the bias in a uniform (across the parameter space) manner
the situation looks quite different.
  In fact, from this
viewpoint, we see that soft-thresholding may actually incur
significantly more bias in finite samples than $Q$-learning,
especially for large values of $\sigma$.  Intuitively reducing bias at
one point in the parameter space leads to increased bias at other
points.   We illustrate the bias both from a theoretical viewpoint as well as providing a toy example that highlights the bias.

Local or moving-parameter asymptotics play an important role in the
theoretical study of nonsmooth estimators, such as $\wh{\beta}_{1}$.  Local asymptotics provide a way to understand and study
the behavior of a nonsmooth estimator in a more uniform manner across the parameter space, in particular by using generative models that
are arbitrarily `close' to the problematic nonsmooth points in the parameter space.
%
%
 Consider the following local asymptotic
framework.
\begin{itemize}
\item[(A3)] For any $s \in \mathbb{R}^{\dim(\beta_{2,1}^*)}$, there exists
  a sequence of local alternatives $P_{n}$ converging to $P$ in the
  sense that:
\begin{equation*}
\int\left[
\rtn\left(dP_{n}^{1/2} - dP^{1/2}\right) - \frac{1}{2}v_sdP^{1/2}
\right]^2 \rightarrow 0,
\end{equation*}
for some real-valued measurable function $v_s$ for which
\begin{itemize}
\item if $\beta_{2,n}^* \triangleq
\arg\min_{\beta}P_{n}(Y - Q_2(H_2, A_2;\beta)^2$, then
$\beta_{2,1,n}^{*} \triangleq \beta_{2,1}^* + s/\rtn + o(1/\rtn)$ and
\item $P_{n}||H_{2}||^2\,||B_1||^2$,  $P_{n}Y_2^2||B_2||^2$ are bounded sequences.
\end{itemize}
\end{itemize}
See the Appendix for the relationship between $v_s$ and $s$.
Define $\tilde{Y}_{n}^* = H_{2,0}^{\T}\beta_{2,0,n}^{*} +
\left[H_{2,1}^{\T}\beta_{2,1,n}^{*}\right]_+$ and $\beta_{1,n}^{*}
\triangleq \arg\min_{\beta}P_{n}(\tilde{Y}_{n}^* - Q_{1}(H_1,
A_1;\beta))^2$.  For any estimator $\widetilde{\beta}_1$ of
$\beta_1^*$ for which $\rtn(\widetilde{\beta}_1-\beta_{1,n}^*)$
converges in distribution under $P_n$ to a random vector indexed by $s$, say $\mathbb{M}(s)$,
define the $c$-directional asymptotic bias under
$P_n$ as
\begin{equation*}
  \mathrm{Bias}(\wh{\beta}_{1}, c, s) \triangleq
  \mathbb{E}c^{\T}\mathbb{M}(s).
\end{equation*}
The following result is proved in Appendix~\ref{ap:proofs}.
\begin{thm}
Assume (A1)-(A3) and let $c\in\mathbb{R}^{\dim(\beta_1^*)}$ be fixed.  Further
assume
that $P1_{H_{2,1}^{\T}\beta_{2,1}^* =0} > 0$. Then:
\begin{enumerate}
  \item $\sup_{s\in\mathbb{R}^{\dim(\beta_{2,1}^*)}}\big|\mathrm{Bias}(\wh{\beta}_{1}, c,
    s)| \le
    \frac{||c^{\T}\Sigma_{1,\infty}^{-1}||P\left[||B_1||\sqrt{H_{2,1}^{\T}\Sigma_{21,21}
        H_{2,1}}1_{H_{2,1}^{\T}\beta_{2,1}^*=0}\right]}{\sqrt{2\pi}} + o(1)$.
  \item $\sup_{s\in\mathbb{R}^{\dim(\beta_{2,1}^*)}}\big|\mathrm{Bias}(
    \wh{\beta}_{1}^{\sigma}, c, s)\big| \rightarrow \infty$ as $\sigma\rightarrow
    \infty$.
\end{enumerate}
\end{thm}\noindent
The preceding suggests that thresholding too aggressively may lead to
large bias in finite samples; results of this type
are anticipated by \cite{sethuramanDoss, brownLiu, hirano2012}.

Next we consider a toy example which more clearly illuminates the effect of thresholding on  bias.
 Consider data
$\lbrace (A_i, Y_i)\rbrace_{i=1}^{n}$ from a two-arm
randomized study where: $A \in\lbrace 0, 1\rbrace$ denotes
a randomly assigned binary treatment; and $Y\in \mathbb{R}$
denotes the outcome coded so that higher values are better.
Assume subjects are randomized with equal probability so
that $P(A=1)=1/2$.  Define $\mu_{a}^*\triangleq \mathbb{E}(Y|A=a)$,
and $\theta^* \triangleq \max(\mu_0^*, \mu_1^*)$ so that $\theta^*$
denotes mean outcome if all subjects are assigned treatment
$\arg\max_{a}\mu_a^*$.  Let $\widehat{\mu}_{a} \triangleq
\pn Y1_{A=a}/\pn 1_{A=a}$, then the plug-in estimator
of $\theta^*$ is
\begin{equation*}
\widehat{\theta} = \max(\widehat{\mu}_{0}, \widehat{\mu}_{1}) = \frac{\widehat{\mu}_{0} +
\widehat{\mu}_{1}}{2} + \frac{|\widehat{\mu}_{0} - \widehat{\mu}_{1}|}{2},
\end{equation*}
which is the sum of a smooth term, $(\widehat{\mu}_{0} +
\widehat{\mu}_{1})/2$, and a non-smooth term $|\widehat{\mu}_0 -
\widehat{\mu}_{1}|/2$.  In this example, the problematic area of the parameter space is $\Theta_{\mathrm{Bad}} =
\lbrace (\mu_1, \mu_2)\in\mathbb{R}^2\,: \, \mu_1 = \mu_2\rbrace$;
under mild regularity conditions it can be seen that if
$\theta^*\notin \Theta_{\mathrm{Bad}}$, then $\rtn(\widehat{\theta} -
\theta^*)$ converges in distribution to mean zero normal random
variable, whereas if $\theta^*\in\Theta_{\mathrm{Bad}}$, then
$\rtn(\widehat{\theta} - \theta^*)$ converges in distribution to $(Z_0
+ Z_1)/2 + |Z_0 - Z_1|/2$ where $Z_0, Z_1$ are independent mean zero
normal random variables.  Thus, when
$\theta^*\in\Theta_{\mathrm{Bad}}$, since $\mathbb{E}|Z_0 -Z_1| \ge 0$ with
equality only when both $Z_0$ and $Z_1$ are degenerate,
$\widehat{\theta}$ has positive asymptotic bias.

One approach to reducing the asymptotic bias of $\widehat{\theta}$
is by thresholding the nonsmooth term in $\widehat{\theta}$.
Assume that $\mathrm{Var}(Y|A=a) = 1$ for $a=0,1$.
For $\sigma > 0$, define
\begin{equation}\label{maxMeansThresh}
  \widehat{\theta}^{\sigma} \triangleq
\frac{\widehat{\mu}_{0} + \widehat{\mu}_{1}}{2} + \frac{|\widehat{\mu}_{0}
- \widehat{\mu}_{1}|}{2}\left(1-\frac{4\sigma}
{n(\widehat{\mu}_{0}-\widehat{\mu}_{1})^2}\right)_{+},
\end{equation}
so that (\ref{maxMeansThresh}) is analogous to (\ref{bibhasEst}).  In
fact, (\ref{maxMeansThresh}) is a special case of (\ref{bibhasEst}) and is the resulting estimator of the mean response at the first stage
when there are no stage 2 covariates (except for the treatment indicator).  Thus, analogous arguments to those in the
preceding section show that, for $\theta^*\in\Theta_{\mathrm{Bad}}$,
$\widehat{\theta}^{\sigma}$ has smaller asymptotic bias than
$\widehat{\theta}$, and that this asymptotic bias decreases as
$\sigma$ increases.  Similarly, a local asymptotic analysis suggests
that aggressive shrinkage may lead to large bias in finite samples.

We now illustrate the small sample behavior of
$\widehat{\theta}^{\sigma}$ using simulated data.  We assume $Y|A=a
\sim \mathrm{Normal}(\mu_a, 1)$ and that treatment assignment is
perfectly balanced.  We use 1000 Monte Carlo replications to estimate
bias for each parameter setting.  The leftmost plot in Figure
\ref{biasPlots} shows the bias as a function of the treatment effect
$\mu_1^* -\mu_0^*$ and tuning parameter $\sigma$ for $n=10$.  Note
that when $n=10$ a standard normal 90\% confidence interval for
$\mu_1^*-\mu_0^*$ has a width of about two.  Thus, the $y$-axis has
been scaled to roughly correspond to a 90\% confidence interval centered around the problematic point $0$.
From the plot it is clear that {\em if} $\mu_1^*-\mu_0^*=0$, larger
values of $\sigma$ correspond to lower bias; however, as anticipated
from the local asymptotic analysis, large values of $\sigma$ cause the
bias to increase dramatically as $\mu_1^*-\mu_0^*$ moves away from
zero but stays within the confidence interval.  As the data do not
contain sufficient information to differentiate between different
parameter values within the confidence interval, an adaptive shrinkage
strategy based on the estimated treatment difference
$\widehat{\mu}_{1}-\widehat{\mu}_{0}$ is not possible.  The middle
plot in Figure \ref{biasPlots} shows the same bias plot for $n=100$
displayed with the same $y$-axis as the $n=10$ case; the very small
yellow-red cross-section above the region around $\sigma=0$ is
anticipated by the fixed asymptotic analysis which states for if
$\mu_1^*-\mu_0^*\ne0$ the bias decreases as the sample size increases.
However, the rightmost plot in Figure \ref{biasPlots} shows the bias
for $n=100$ after rescaling the $y$-axis to reflect power (i.e., now
the range of the $y$-axis corresponds to the length of a standard
normal 90\% confidence interval for $\mu_1^*-\mu_0^*$ when $n=100$);
the figure is {\em essentially identical} to the leftmost ($n=10$)
plot.  The similarity of these plots after rescaling exemplifies the
insights gained from a local asymptotics approach which allows notions
of `closeness' to persist as the sample size increases.

\begin{figure}[ht]
  \includegraphics[width=0.32\linewidth]{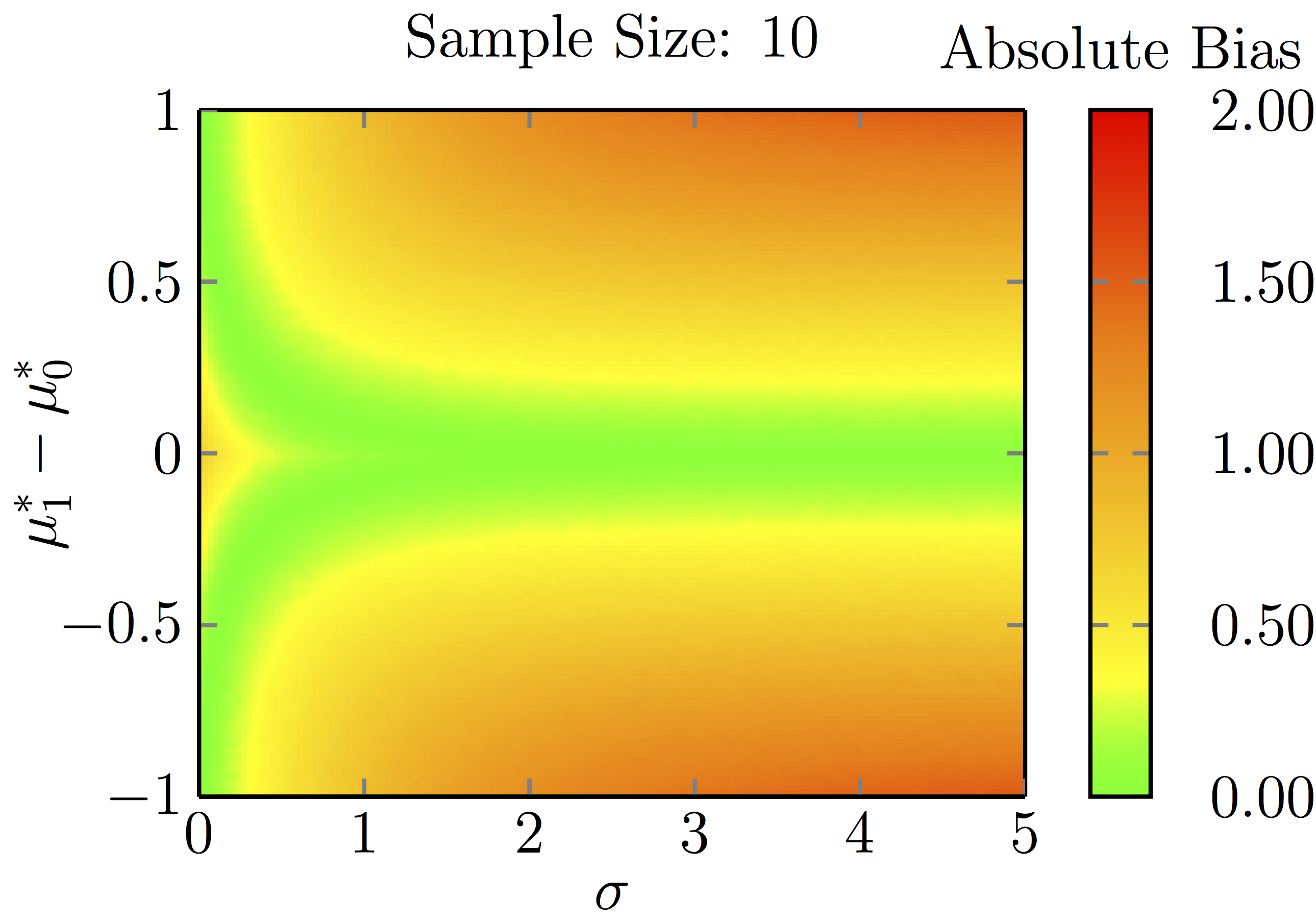}
  \includegraphics[width=0.32\linewidth]{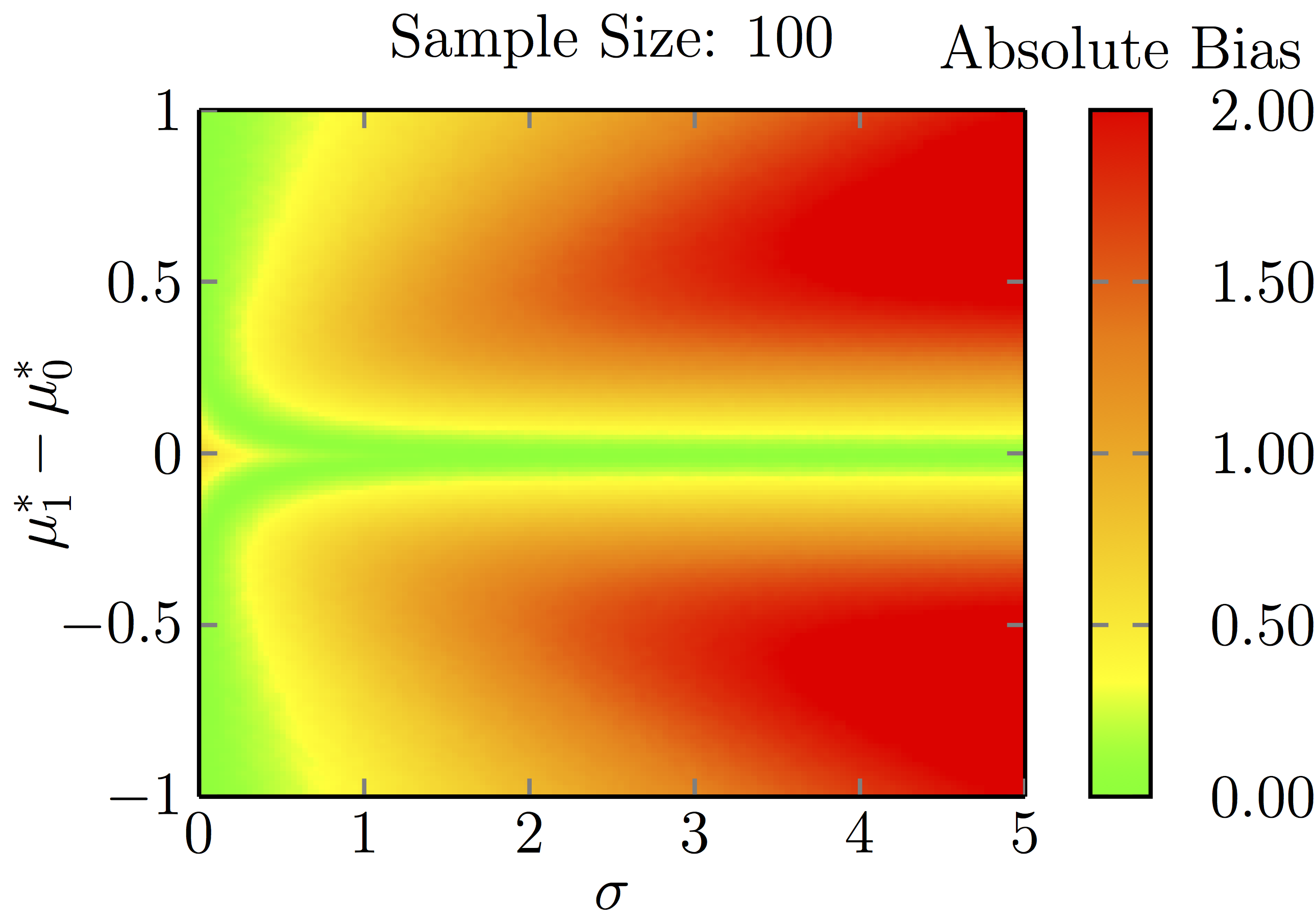}
  \includegraphics[width=0.32\linewidth]{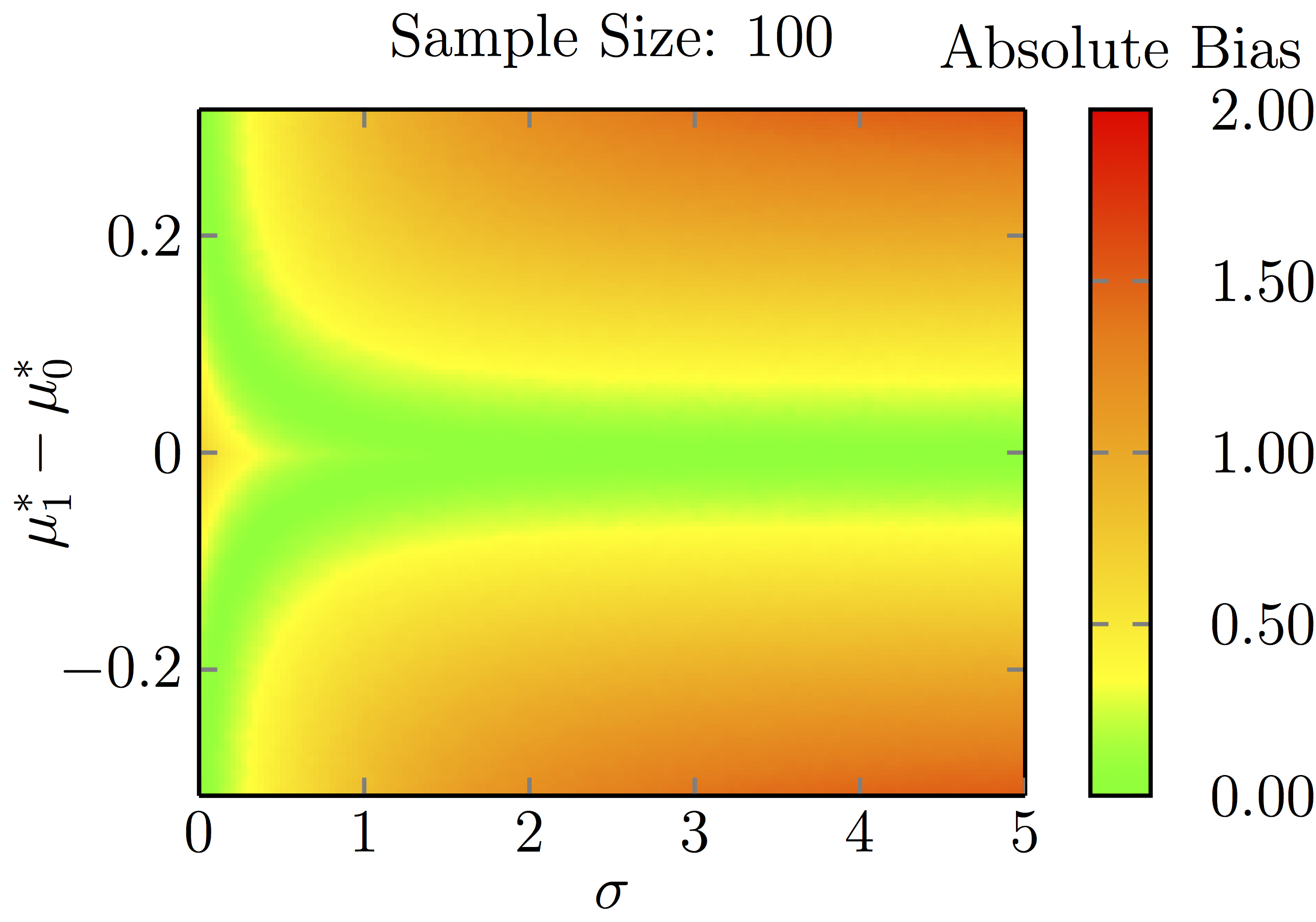}
\caption{ \textbf{Left:} Bias, in units of
  $1/\sqrt{n}$, as a function of effect size $\mu_1^* - \mu_0^*$ and
  tuning parameter $\sigma$ for $n=10$.  \textbf{Center:} Bias, in
  units of $1/\sqrt{n}$, as a function of effect size
  $\mu_1^*-\mu_0^*$ and tuning parameter $\sigma$ for $n=100$;
  \textbf{Right:} Same as center plot after rescaling $y$-axis.}
\label{biasPlots}
\end{figure}

%

\section{Confidence intervals}
If estimated optimal DTRs are to be used to inform clinical decision
making or future research it is essential that they be accompanied by
reliable measures of uncertainty.  Constructing valid confidence intervals
from nonregular estimators is  difficult because it
is impossible to uniformly consistently estimate the sampling
distribution of a nonregular estimator \citep[][]{van1991differentiable,
andrews2000inconsistency, leeb2003finite, hirano2012}.  Estimators that
reduce asymptotic bias, for example thresholding
\citep{chakraborty2009inference} and singular penalization
\citep{song, goldberg2012adaptive}, were originally suggested
as methods for constructing high-quality confidence intervals
for parameters in $Q$-learning.  However, these methods involve
additional nonsmooth operations of the data and it can be shown
that the confidence intervals proposed with these estimators are inconsistent
under local alternatives.
Furthermore, asymptotic bias
only reflects the mean of the sampling distribution whereas confidence
intervals require estimation of the tails of the sampling distribution.
Thus, in general reducing asymptotic bias is not sufficient for valid
inference.

On the other hand, confidence intervals that deliver the desired level
of confidence can be used to conduct inference even in the presence of
bias on the order $1/\sqrt{n}$.
In this section we: (i) review an
adjusted projection interval proposed by Robins [2004]; and (ii)
propose a new procedure that is adaptive and locally consistent.
Additional discussion and potential extensions of the methods proposed
here are provided in Section 7.


\subsection{An adjusted projection interval}
Recall that $h_{2,1}^\T\beta_{2,1}^*$ is the second stage treatment
effect (see Section 2.1) for feature vector $h_{2,1}$; small sample
inferential problems occur when this second stage treatment effect is
small with positive probability (e.g., small sample bias, poor
coverage properties of standard CIs).  Robins [2004] using ideas
similar to those of Berger and Boos [1994] proposed an adjusted
projected confidence interval.  In the context of Q-Learning this idea
is as follows.  For any $\beta_{2,1}$ define
$\widetilde{Y}(\beta_{2,1}) \triangleq \max_{a_2} Q_2(H_2, a_2;
(\widehat{\beta}_{2,0}^{\T}, \beta_{2,1}^{\T}))$ and
$\widetilde{Y}^*(\beta_{2,1}) \triangleq \max_{a_2}Q_2(H_2, a_2;
(\beta_{2,0}^{*\T}, \beta_{2,1}^{\T}))$; subsequently define
$\widehat{\beta}_{1}(\beta_{2,1}) \triangleq \arg\min_{\beta_1} \pn
(\widetilde{Y}_{1}(\beta_{2,1}) - Q_{1}(H_1, A_1; \beta_1))^2$ and
$\beta_{1}^*(\beta_{2,1}) \triangleq \arg\min_{\beta_1}
P(\widetilde{Y}^* - Q_1(H_1, A_1;\beta_1))^2$.  Note that $\beta_{1}^*
= \beta_{1}^*(\beta_{2,1}^*)$.  For $\beta_{2,1}$ fixed, it follows
from standard arguments that $\rtn(\widehat{\beta}_{1}(\beta_{2,1}) -
\beta_{1}^*(\beta_{2,1}))$ is regular, asymptotically normal with mean
zero.  Let $C(\beta_{2,1})$ denote the asymptotic variance-covariance
matrix of $\rtn(\widehat{\beta}_{1}(\beta_{2,1}) -
\beta_{1}^*(\beta_{2,1}))$ and let $\widehat{C}(\beta_{2,1})$ denote a
consistent estimator of $C(\beta_{2,1})$.  A Wald-type asymptotic
$(1-\alpha)\times 100\%$ confidence region for
$\beta_1^*(\beta_{2,1})$ is therefore
\begin{equation*}
\mathbb{I}_{n,\alpha}(\beta_{2,1}) \triangleq \left\lbrace \beta_1\in\mathbb{R}^{\dim(\beta_1^*)}\,: \,
n \left(\widehat{\beta}_{1}(\beta_{2,1}) -
  \beta_1\right)^{\T}\widehat{C}^{-1}(\beta_{2,1})
\left(\widehat{\beta}_{1}(\beta_{2,1}) - \beta_1\right) \le \chi_{1-\alpha, \dim(\beta_1^*)}^2
\right\rbrace,
\end{equation*}
where $\chi_{\alpha, d}^{2}$ is the $(1-\alpha)\times 100$ percentile of a
$\chi^2$-distribution with $d$ degrees of freedom.   In particular,
$\mathbb{I}_{n, \alpha}(\beta_{2,1}^*)$ is a valid asymptotic
$(1-\alpha)\times 100\%$ confidence interval for
$\beta_{1}^*(\beta_{2,1}^*) = \beta_{1}^*$.  Of course,
$\beta_{2,1}^*$ is unknown, but $\hat{\beta}_{2,1}$ is a regular
asymptotically normal estimator of $\beta_{2,1}^*$ and thus standard
methods for constructing confidence sets, e.g., the bootstrap or
Taylor series arguments, can be used to construct a valid
$(1-\eta)\times 100\%$ for $\beta_{2,1}^*$, say $\zeta_{n, \eta}$.
Then, the union
\begin{equation}
  \bigcup_{\beta_{2,1}\in\zeta_{n,\eta}}\mathbb{I}_{n,\alpha}(\beta_{2,1}),\label{projInt}
\end{equation}
is a valid $(1-\alpha-\eta)\times 100\%$ confidence region
for $\beta_{1}^*$.  To see this, note that
\begin{multline*}
P\left(
  \beta_{1}^* \notin  \bigcup_{\beta_{2,1}\in\zeta_{n,\eta}}\mathbb{I}_{n,\alpha}(\beta_{2,1})
\right)  = P\left(
\beta_{1}^* \notin
\bigcup_{\beta_{2,1}\in\zeta_{n,\eta}}\mathbb{I}_{n,\alpha}(\beta_{2,1}),\,
\beta_{2,1}^* \notin \zeta_{n,\eta}
\right) \\ +
P\left(
\beta_{1}^* \notin
\bigcup_{\beta_{2,1}\in\zeta_{n,\eta}}\mathbb{I}_{n,\alpha}(\beta_{2,1}),\,
\beta_{2,1}^* \in \zeta_{n,\eta}
\right),
\end{multline*}
which is bounded above by $P\left( \beta_{2,1}^* \notin \zeta_{n,\eta}
\right) + P\left( \beta_{1}^*(\beta_{2,1}^*) \notin
  \mathbb{I}_{n,\alpha}(\beta_{2,1}^*) \right) \le \eta + \alpha +
o_{P}(1). $ This confidence interval is appealing for its simplicity
but may be conservative especially when $H_{2,1}^{\T}\beta_{2,1}^*$ is
bounded away from zero with high probability.  One approach to reduce
conservatism is to first test $H_0:\beta_{2,1}^* \equiv 0$, if the
test rejects then $\mathbb{I}_{n, \alpha}(\widehat{\beta}_{2,1})$ is
used, if the test fails to reject then the projection interval
(\ref{projInt}) is used (Robins 2004).  This pretesting approach is
adaptive at the population level but may be conservative when the
distribution of $H_{2,1}^{\T}\beta_{2,1}^*$ has mass both near to and
far from zero.  A potentially less conservative approach is to
partition the observed sample into two groups according to the
(estimated) magnitude of $H_{2,1}^{\T}\beta_{2,1}^*$ and apply a
conservative procedure only to observations for which
$H_{2,1}^{\T}\beta_{2,1}^*$ is small.  We now discuss such a
procedure.

\subsection{Adaptive confidence intervals}
In this section we construct a regular, i.e., locally consistent,
confidence interval for linear combinations of the first stage
coefficients.   Note that confidence intervals for the second stage
coefficients can be obtained using standard methods for least squares
estimators.
Let $\widehat{\Sigma}_{1} \triangleq \pn
B_1B_1^{\T}$ so that $\wh{\beta}_1 =
\wh{\Sigma}_{1}^{-1}\pn B_1\widetilde{Y}$ and
$\beta_1^* = \Sigma_{1,\infty}^{-1}P B_1 \widetilde{Y}^*$.
Recall that it is not possible in general to construct a uniformly
convergent estimator of the limiting distribution of
$\rtn(\wh{\beta}_{1} - \beta_1^*)$~\citep[][]{van1991differentiable,
  hirano}.  For a given constant $c\in\mathbb{R}^{\dim(\beta_1^*)}$,
our approach is to bound $c^{\T}\rtn(\wh{\beta}_{1} - \beta_1^*)$
between two regular, uniformly convergent, upper and lower bounds.
Because these bounds are smooth, we can bootstrap them to form a
confidence set for $c^{\T}\beta_1^*$.  This strategy is similar to the
work of \cite{laber2011adaptive} on classification but differs in that
here the functional of interest is a fixed (rather than
data-dependent) parameter and the functional is more complicated.  We
present the two-stage binary-treatment case here; extensions to the
case of an arbitrary number of treatments and stages of treatment can
be found in a technical report \citep[][]{laber2010statistical}.

Recall that for any $c\in\mathbb{R}^{\dim(\beta_1^*)}$
$c^{\T}\rtn(\wh{\beta}_{1} - \beta_{1}^*) =
c^{\T}\wh{\Sigma}_{1}^{-1}\pn B_1(\widetilde{Y} - B_1^{\T}\beta_1^*)$
can be decomposed as $c^{\T}\mathbb{S}_{n} +
c^{\T}\wh{\Sigma}_{1}^{-1}\pn B_1 \mathbb{U}_{n},$ where the term
$\mathbb{S}_{n}$ is smooth and asymptotically normal but
$\mathbb{U}_{n}$ is nonsmooth.  Also recall that
$\mathbb{U}_n = \rtn\left(\big[H_{2,1}^{\T} \bHat_{2,1}\big]_+ -
\left[H_{2,1}^{\T}\beta_{2,1}^*\right]_+\right)$.
Our goal is to form smooth upper and
lower bounds on $c^{\T}\rtn(\wh{\beta}_{1} - \beta_{1}^*)$. To limit
conservatism, these bounds are based on the nonsmooth term
$\mathbb{U}_{n}$ and only involve subjects with small second stage
treatment effects, i.e., those subjects with histories $h_{2,1}$ with
$h_{2,1}^{\T}\beta_{2,1}^*\approx 0$.  We partition the observed data
into two groups: (Group 1) subjects for whom
$h_{2,1}^{\T}\beta_{2,1}^*$ cannot be distinguished from zero; and
(Group 2) subjects for whom $h_{2,1}^{\T}\beta_{2,1}^*$ is unlikely to
be near zero.  This partitioning is based on a \lq\lq pretest"
\citep[see][]{olshen1973conditional, andrews2001testing,
  andrewsSoares, cheng, andrews2009incorrect}.  The pretest is based on
$\wh{T}(h_{2,1})$ which is a test statistic that diverges to
$+\infty$ when $h_{2,1}^{\T}\beta_{2,1}^*$ is nonzero but is bounded in
probability when $h_{2,1}^{\T}\beta_{2,1}^*=0$.  The pretest assigns a
subject with $H_{2,1}=h_{2,1}$ to Group 1 if $\wh{T}(h_{2,1}) \le
\lambda_{n}$ and Group 2 otherwise;  $\lambda_{n}$ is a tuning
parameter.
In what follows we
assume $\wh{T}(h_{2,1}) =
n(h_{2,1}^{\T}\wh{\beta}_{2,1})^2/h_{2,1}^{\T}
\widehat{\Sigma}_{21,21}h_{2,1}$ where $\wh{\Sigma}_{21, 21}$ is the
submatrix of $(\pn B_2B_2^{\T})^{-1}
\pn B_2B_2^{\T}(Y-B_2^{\T}\wh{\beta}_{2,1})^2(\pn B_2
B_2^{\T})^{-1}$ corresponding to the plug-in estimator of the
asymptotic variance of $\mathbb{V}_{n}\triangleq \rtn(\wh{\beta}_{2,1}
- \beta_{2,1}^*)$.

The upper bound on $c^{\T}\rtn(\wh{\beta}_{1} - \beta_1^*)$ is given by
\begin{multline}\label{upperBoundPopn}
\mathcal{U}(c) \triangleq c^{\T}\mathbb{S}_n +
c^{\T}\hat{\Sigma}_{1}^{-1}\pn B_{1}\mathbb{U}_n1_{\hat{T}(H_{2,1}) > \lambda_n}  \\
+ \sup_{\gamma\in\mathbb{R}^{\dim(\beta_{2,1}^*)}}
c^{\T}\hat{\Sigma}_{1}^{-1}\pn
B_{1}
\left(
\left[H_{2,1}^{\T}(\mathbb{V}_n + \gamma)\right]_+ -
\left[H_{2,1}^{\T}\gamma\right]_+
\right)1_{\hat{T}(H_{2,1}) \le \lambda_n}.
\end{multline}
A lower bound, say $\mathcal{L}(c)$, is obtained by replacing $\sup$ with
$\inf$ in the above display.
The intuition behind this upper bound is as
follows. Notice that the second term  in (\ref{coeffDecomp}), namely $c^{\T}\hat{\Sigma}_{1}^{-1}\pn
B_{1}\mathbb{U}_n$,  is equal to
$c^{\T}\hat{\Sigma}_{1}^{-1}\pn B_{1}\mathbb{U}_n1_{\hat{T}(H_{2,1}) > \lambda_{n}} +
c^{\T}\hat{\Sigma}_{1}^{-1}\pn B_{1}\mathbb{U}_{n}1_{\hat{T}(H_{2,1})
  \le \lambda_n}$.
Rewrite the $\mathbb{U}_{n}$ in  $c^{\T}\hat{\Sigma}_{1}^{-1}\pn B_{1}\mathbb{U}_{n}1_{\hat{T}(H_{2,1}) \le \lambda_n}$   as $\left[H_{2,1}^{\T}(\mathbb{V}_n + \rtn\beta_{2,1}^*)\right]_+ -
\left[H_{2,1}^{\T}\rtn\beta_{2,1}^*\right]_+$.  Thus
$c^{\T}\hat{\Sigma}_{1}^{-1}\pn
B_{1}\mathbb{U}_n$,  is equal to
\begin{multline}\label{explainDecompPopn}
c^{\T}\hat{\Sigma}_{1}^{-1}\pn B_{1}\mathbb{U}_n1_{\hat{T}(H_{2,1}) > \lambda_{n}}\\
+  c^{\T}\hat{\Sigma}_{1}^{-1}\pn
B_{1}
\left(
\left[H_{2,1}^{\T}(\mathbb{V}_n + \rtn\beta_{2,1}^*)\right]_+ -
\left[H_{2,1}^{\T}\rtn\beta_{2,1}^*\right]_+
\right)
1_{\hat{T}(H_{2,1}) \le \lambda_n}.
\end{multline}
The quantity, $\left[H_{2,1}^{\T}\rtn\beta_{2,1}^*\right]_+$
characterizes the degree of nonregularity of $\rtn(\bHat_1 -
\beta_1^*)$ (see Theorem 4.2 below).  Replacing $\rtn\beta_{2,1}^*$
with $\gamma$ and taking the supremum over all $\gamma \in
\mathbb{R}^{\dim(\beta_{2,1}^*)}$ is one way of making the second term
in (\ref{explainDecompPopn}) insensitive to local perturbations of
$\beta_{2,1}^*$.

To use the bounds to construct a $(1-\alpha)\times 100\%$ confidence
interval for $c^{\T}\beta_1^*$, first note that $c^{\T}\wh{\beta}_{1}
- \mathcal{U}(c)/\rtn \le c^{\T}\beta_1^* \le c^{\T}\wh{\beta}_{1}
-\mathcal{L}(c)/\rtn$.  We approximate the distribution of the bounds
using the nonparametric bootstrap.  Let $\wh{u}$ denote the
$(1-\alpha/2)\times 100$ percentile of the bootstrap distribution of
$\mathcal{U}(c)$, and let $\wh{l}$ denote the $(\alpha/2)\times 100$
percentile of the bootstrap distribution of $\mathcal{L}(c)$.  Then,
$[c^{\T}\wh{\beta}_{1}-\wh{u}\rtn, c^{\T}\wh{\beta}_{1} -
\wh{l}/\rtn]$ is the proposed confidence interval for
$c^{\T}\beta_1^*$.  We term this confidence interval an adaptive
confidence interval (ACI) for reasons that will become clear shortly.

\begin{slntrmrk*}
  In the ACI $\lambda_{n}$ is a potentially important tuning
  parameter.  In Section 5 we demonstrate that the double bootstrap is
  an effective strategy for constructing a data-driven choice of
  $\lambda_{n}$.
\end{slntrmrk*}

\subsubsection{Theoretical results}
In this section we describe the limiting behavior of the bounds
$\mathcal{L}(c)$ and $\mathcal{U}(c)$ and relate them to the limiting
distribution of $c^{\T}\rtn(\wh{\beta}_{1} - \beta_1^*)$. We  assume:
\begin{itemize}
  \item[(A4)] With probability one the sequence $\lambda_{n}$ tends to infinity
     with $n$  and satisfies $\lambda_n/n \rightarrow 0$.
\end{itemize}
\begin{thm}[Validity of population bounds]
\label{thm:stage1}
Assume (A1)-(A2)  and (A4) and fix $c\in\mathbb{R}^{\dim(\beta_1^*)}$.
\begin{enumerate}
\item $c^{\T}\rtn(\bHat_1 -\beta_1^*) \leadsto
c^{\T}\mathbb{S}_{\infty} + c^{\T}\Sigma_{1,\infty}^{-1}
P\left(B_{1}H_{2,1}^{\T}\mathbb{V}_{\infty}1_{H_{2,1}^{\T}\beta_{2,1}^* > 0}\right) +
c^{\T}\Sigma_{1,\infty}^{-1}PB_{1}\left[H_{2,1}^{\T}\mathbb{V}_{\infty}\right]_+
1_{H_{2,1}^{\T}\beta_{2,1}^*=0}$.
\item If for each $n$, the underlying generative distribution is $P_{n}$, which
  satisfies (A3), then the limiting distribution of
  $c^{\T}\rtn(\bHat_1 - \beta_{1,n}^* )$ is equal to
\begin{multline}
c^{\T}\mathbb{S_\infty} + c^{\T}\Sigma_{1,\infty}^{-1}P\left(B_{1}H_{2,1}^{\T}\mathbb{V}_{\infty}1_{H_{2,1}^{\T}\beta_{2,1}^* > 0}\right)\label{nonreg}
 \\
+ c^{\T}\Sigma_{1,\infty}^{-1}P\left[
B_{1}\left(\left[
H_{2,1}^{\T}(\mathbb{V}_{\infty} + s)
\right]_+ - \left[H_{2,1}^{\T}s\right]_+
\right)1_{H_{2,1}^{\T}\beta_{2,1}^* = 0}\right].
\end{multline}
\item The limiting distribution of $\mathcal{U}(c)$ under both $P$ and
  under $P_{n}$ is
equal to
\begin{multline}
c^{\T}\mathbb{S_\infty} + c^{\T}\Sigma_{1,\infty}^{-1}P\left(B_{1}H_{2,1}^{\T}\mathbb{V}_{\infty}1_{H_{2,1}^{\T}\beta_{2,1}^* > 0}\right)  \\
+ \sup_{\gamma\in\mathbb{R}^{\dim(\beta_{2,1}^*)}}c^{\T}\Sigma_{1,
\infty}^{-1}P\left[
B_{1}\left(\left[
H_{2,1}^{\T}(\mathbb{V}_{\infty} + \gamma)
\right]_+ - \left[H_{2,1}^{\T}\gamma\right]_+
\right)1_{H_{2,1}^{\T}\beta_{2,1}^* = 0}\right],\label{regularized}
\end{multline}
\end{enumerate}
where  $\left(\mathbb{S}_{\infty}^\T,\ \mathbb{V}_{\infty}^\T\right)$ is
asymptotically multivariate normal with mean zero.
\end{thm} \noindent See the Appendix for a proof and the formula for
the $\Cov(\mathbb{S}_{\infty},\mathbb{V}_{\infty})$.  Notice that
limiting distributions of $c^{\T}\rtn(\bHat_1 - \beta_1^*)$ and
$\mathcal{U}(c)$ (or equivalently $\mathcal{L}(c)$) are equal in the
case $H_{2,1}^{\T}\beta_{2,1}^* \ne 0$ with probability one.  That is,
when there is a large treatment effect for almost all patients then
the upper (or lower) bound is tight.  However, when there is a
non-null subset of patients for whom there is no treatment effect,
then the limiting distribution of the upper bound is stochastically
larger than the limiting distribution of $c^{\T}\rtn(\bHat_1 -
\beta_1^*)$.  Thus, the ACI adapts to the setting in which all
patients experience a treatment effect.

Because the distribution of (\ref{nonreg}) depends on the local
alternative, $s$, $\hat\beta_1$ is a nonregular estimator
\citep[][]{van1996weak}.  One might hope to construct an
estimator of the distribution of (\ref{nonreg}) and use this estimator
to approximate the distribution of $c^{\T}\rtn(\bHat_1 - \beta_1^*)$.
However, a consistent estimator of the distribution of (\ref{nonreg})
does not exist because $P_{n}$ is contiguous with respect to
$P$ (by assumption A3).  To see this, let $F_{s}(u)$ be the
distribution of (\ref{nonreg}) evaluated at a point, $u$.  If a
consistent estimator, say $\hat F_n(u)$, existed, that is $\hat
F_n(u)$ converges in probability to $F_{s}(u)$ under $P_{n}$,
then the contiguity implies that $\hat F_n(u)$ converges in
probability to $F_{s}(u)$ under $P$.  This is a contradiction (at
best $\hat F_n(u)$ converges in probability to $F_0(u)$ under $P$).
Because we cannot consistently estimate $s$ and we do not know
the value of $s$, the tightest estimable upper bound on
(\ref{nonreg}) is given by (\ref{regularized}).  As we shall next see,
we are able to consistently estimate the distribution of
(\ref{regularized}).

In order to form confidence sets, the bootstrap distributions of
$\mathcal{U}(c)$ and $\mathcal{L}(c)$ are used.  The next result
regards the consistency of these bootstrap distributions.  Let $\pHat$
denote the bootstrap empirical measure, that is, $\pHat \triangleq
n^{-1}\sum_{i=1}^n M_{n,i}\delta_{\mathcal{T}_i}$ for $(M_{n,1},
M_{n,2},\ldots, M_{n,n}) \sim \mathrm{Multinomial}(n,\, (1/n, 1/n,
\ldots, 1/n))$.  We use the superscript $(b)$ to denote that a
functional has been replaced by its bootstrap analog, so that if
$\omega \triangleq f(\pn)$ then $w^{(b)} \triangleq f(\pHat)$.  Denote
the space of bounded Lipschitz-1 functions on $\mathbb{R}^2$ by
$BL_1(\mathbb{R}^2)$.  Furthermore, let $\mathbb{E}_M$ and $P_M$
denote the expectation and probability with respect to the bootstrap
weights. The following results are proved in the Appendix.
\begin{thm}
\label{thm:ACI}
Assume (A1)-(A2), (A4) and fix $c\in\mathbb{R}^{\dim(\beta_1^*)}$.
Then $(\mathcal{U}(c), \mathcal{L}(c))$ and
$(\mathcal{U}^{(b)}(c), \mathcal{L}^{(b)}(c))$ converge to the same limiting
distribution in probability.  That is,
\begin{equation*}
\sup_{v \in BL_1(\mathbb{R}^2)}\bigg|
\mathbb{E}v\left(\left(\mathcal{U}(c), \mathcal{L}(c)\right)\right)
- \mathbb{E}_Mv\left(\left(
\mathcal{U}^{(b)}(c), \mathcal{L}^{(b)}(c)
\right)\right)
\bigg|
\end{equation*}
converges in probability to zero.
\end{thm}\noindent
\begin{cor}
  Assume (A1)-(A2), (A4) and fix $c\in\mathbb{R}^{\dim(\beta_1^*)}$.  Let
  $\hat{u}$ denote the $(1-\alpha/2)\times 100$ percentile of
  $\mathcal{U}^{(b)}(c)$ and $\hat{l}$ denote the $(\alpha/2)\times
  100$ percentile of $\mathcal{L}^{(b)}(c)$.  Then
\begin{equation*}
P_M\left(
c^{\T}\bHat_1 - \hat{u}/\rtn \le c^{\T}\beta_1^* \le c^{\T}\bHat_1 - \hat{l}/\rtn
\right) \ge 1- \alpha + o_P(1).
\end{equation*}
Furthermore, if $P(H_{2,1}^{\T}\beta_{2,1}^* = 0) = 0$, then the above
inequality can be strengthened to equality.
\end{cor}\noindent
The preceding results show that the ACI can be use to construct valid
confidence intervals regardless of the underlying parameters or
generative model.  Moreover, in settings where there is a treatment
effect for almost every patient, the ACI delivers asymptotically exact
coverage.  See Section 5 for discussion of the choice of the tuning
parameter $\lambda_n$.

\section{Experiments}
In this section we examine the small sample performance of the
adaptive confidence interval (ACI) proposed in the Section 4.2 where
performance is measured in terms of coverage and average interval
width.  We consider both fixed and data-driven choices for the tuning
parameter $\lambda_{n}$.  For a fixed value we choose $\lambda_{n} =
\sqrt{\log\log\,n}$; additional simulations taken over a range of
$\lambda_n$ values are provided in the Appendix.  These
simulations show that the method is potentially sensitive to the
choice of $\lambda_n$.  Consequently, we also consider a data-driven
choice of $\lambda_n$, tuned using the double-bootstrap
\citep{davison1997}.  In particular, we consider a range of values of
$\lambda_{n}$ of the form $\lambda_{n} = \tau\sqrt{\log\log\,n}$ where
$\tau\in [m, M]$ where $0 < m < M < \infty$ are fixed constants. See
the Appendix for the specifics of the double bootstrap algorithm.  Note that the
theoretical properties of the ACI continue to hold with this adaptive
scheme for choosing $\lambda_n$ since $m\sqrt{\log\log\,n} \le
\tau\sqrt{\log\log\,n} \le M\sqrt{\log\log\,n}$ so that $\lambda_n$
satisfies (A4).

We compare the empirical performance of the ACI with $\lambda_n$ fixed
to equal $\sqrt{\log\log\,n}$ (FACI) and $\lambda_n$ chosen using the
double-bootstrap (DACI) with the following methods: the centered
percentile bootstrap (CPB); the centered percentile bootstrap of the
soft-thresholding (ST) method of \cite{chakraborty2009inference} as
described in Section 3; and the adaptive $m$-out-of-$n$ (MOFN)
bootstrap with data-driven tuning of \cite{mofn}. We also implemented
and tested the projection interval described in Section 4.1 with $\eta
= 0.01, \alpha = 0.04$; results are not shown in the tables as they were
too wide to be useful. The projection interval always covered at least
at the nominal level (and frequently much more -- in 6 of 18
experiments it covered 100\% of the time) but it was between 1.46 and
2.07 times wider than the DACI, which also achieves or exceeds nominal
coverage.  The hard-thresholding method of \cite{moodieT} and the
penalized approach of \cite{song} are similar in both theory and
performance to the soft-thresholding approach and thus are omitted
from our experiments.

Nine generative models
are used in these evaluations; each of these
generative models has two stages of treatment and two treatments
at each stage.   Generically, each of the models can be described as follows:
\begin{itemize}
\item $X_i \in \{-1,1\}$, $A_i \in \{-1,1\}$ for $i \in \{1,2\}$
\item $P(A_1= 1) = P(A_1 = -1) = 0.5$, $P(A_2= 1) = P(A_2 = -1) = 0.5$
\item $X_1 \sim \mathrm{Bernoulli}(0.5)$, $X_2|X_1,A_1 \sim \mathrm{Bernoulli}(\mathrm{expit}({\delta_1 X_1 + \delta_2 A_1}))$
\item $Y = \gamma_1 + \gamma_2 X_1 + \gamma_3 A_1 + \gamma_4 X_1 A_1 + \gamma_5 A_2 + \gamma_6 X_2 A_2 + \gamma_7 A_1 A_2 + \epsilon$,
$\epsilon \sim N(0,1)$
\end{itemize}
where $\mathrm{expit}(x) = \mathrm{e}^{x} / (1 + \mathrm{e}^x)$. This
class is parameterized by nine values
$\gamma_1,\gamma_2,...,\gamma_7,\delta_1,\delta_2$. The analysis model
uses  feature vectors  defined by:
\begin{displaymath}
  \begin{array}{ll}
  H_{2,0}  = (1, X_1, A_1,X_1 A_1, X_2)^\T, &
  H_{2,1}  = (1, X_2, A_1)^\T,\\
  H_{1,0}  = (1, X_1)^\T, &
  H_{1,1}  = (1, X_1)^\T.
  \end{array}
\end{displaymath}
Our analysis models are given by $Q_2(H_2, A_2; \beta_2) \triangleq
H_{2,0}^{\T}\beta_{2,0} + H_{2,1}^{\T}\beta_{2,1}A_2$ and $Q_1(H_1,
A_1; \beta_1) \triangleq H_{1,0}^{\T}\beta_{1,0} +
H_{1,1}^{\T}\beta_{1,1}A_1$.  Below the
analysis models are correctly specified (match the generative models).  This avoids conflating poor
performance of confidence intervals due to misspecification with poor
performance due to nonregularity.  We use a contrast encoding for
$A_1$ and $A_2$ to allow for a comparison with Chakraborty et al.\
(2009).

The form of this class of generative models is useful as it allows us
to influence the degree of nonregularity present in our example
problems through the choice of the $\gamma_i$ and $\delta_i$, and in
turn evaluate performance in these different scenarios. Recall that in
Q-learning, nonregularity occurs when more than one stage-two
treatment produces nearly the same optimal expected reward for a set
of patient histories that occur with positive probability. In the
model class above, this occurs if the model generates histories for
which $\gamma_5 A_2 + \gamma_6 X_2 A_2 + \gamma_7 A_1 A_2 \approx 0$,
i.e., if it generates histories for which $Q_2$ depends weakly or not
at all on $A_2$. By manipulating the values of $\gamma_i$ and
$\delta_i$, we can control i) the probability of generating a patient
history such that $\gamma_5 A_2 + \gamma_6 X_2 A_2 + \gamma_7 A_1 A_2
= 0$, and ii) a standardized effect size $E[ (\gamma_5 + \gamma_6
X_2 + \gamma_7 A_1) / \sqrt{\mathrm{Var}(\gamma_5 + \gamma_6 X_2 +
  \gamma_7 A_1)}]$. Each of these quantities, denoted by $p$ and
$\phi$, respectively, can be thought of as measures of
nonregularity.

Table \ref{tab:modelparams_2act_2stage} provides the parameter
settings; the first six settings were considered by Chakraborty et
al.\ (2009), and are described by them as ``nonregular'',
``near-nonregular'', and ``regular''.  To these six, we have added
three additional examples labeled A, B, and C. Example A is an example
of a strongly regular setting. Example B is an example of a nonregular
setting where the nonregularity is strongly dependent on the stage 1
treatment. In example B, for histories with $A_1 = 1$, there is a
moderate effect of $A_2$ at the second stage. However, for histories
with $A_1 = -1$, there is no effect of $A_2$ at the second stage,
i.e., both treatments at the second stage are equally optimal. In
example C, for histories with $A_1 = 1$, there is a moderate effect of
$A_2$, and for histories with $A_1 = -1$, there is a small effect of
$A_2$. Thus example C is a `near-nonregular' setting that behaves
similarly to example B. In addition to these new examples, we give
extensions of all nine examples to a setting with three treatments at
the second stage; details are given in Appendix~\ref{ap:threetxt}.

\begin{table}[h!]
\begin{tabular}{c|c|c|c|c|c}
Example & $\gamma$ & $\delta$ & Type & \multicolumn{2}{c}{Regularity Measures} \\
\hline
1 & $(0, 0,  0,   0, 0,    0,   0)^\T$ & $(0.5, 0.5)^\T$ &
nonregular & $p = 1$ & $\phi = 0/0$ \\
2 & $(0, 0,   0,   0, 0.01, 0,   0)^\T$ & $(0.5, 0.5)^\T$  &
near-nonregular & $p = 0$ & $\phi = \infty$ \\
3 & $(0, 0, -0.5, 0, 0.5,  0,   0.5)^\T$ & $(0.5, 0.5)^\T$ &
nonregular & $p = 1/2$ & $\phi = 1.0$ \\
4 & $(0, 0, -0.5, 0, 0.5,  0,   0.49)^\T$ & $(0.5, 0.5)^\T$ &
near-nonregular & $p = 0$ & $\phi = 1.02$ \\
5 & $(0, 0, -0.5, 0, 1.0, 0.5, 0.5)^\T$ & $(1.0, 0.0)^\T$ &
nonregular & $p = 1/4$ & $\phi = 1.41$ \\
6 & $(0, 0, -0.5, 0, 0.25, 0.5, 0.5)^\T$ & $(0.1, 0.1)^\T$ &
regular & $p = 0$ & $\phi = 0.35$ \\
\hline
A & $(0, 0, -0.25, 0, 0.75, 0.5, 0.5) ^\T$ & $(0.1, 0.1)^\T$ & regular &
$p=0$ & $\phi = 1.035$ \\
B & $(0, 0, 0, 0, 0.25, 0, 0.25)^\T$ & $(0, 0)^\T$ & nonregular &
$p=1/2$ & $\phi = 1.00$ \\
C & $(0, 0, 0, 0, 0.25, 0, 0.24)^\T$ & $(0, 0)^\T$ & near-nonregular &
$p=0$ & $\phi = 1.03$ \\
\end{tabular}
\caption{Parameters indexing
the example models.}
\label{tab:modelparams_2act_2stage}
\end{table}

We first provide confidence intervals for
the coefficient of $A_1$ (the treatment variable), $\beta_{1,1,1}^*$
in settings in which there are two or three treatments at stage
2. (The three-treatment version of the ACI is given by \cite{laber2010statistical}.) Note that given the working models and generative models defined by the
parameter settings in Table \ref{tb:modelparams_2act_2stage}, we can
determine the exact value of any parameter $c^{\T}\beta_1^*$ of
interest to set the ground truth for our experiments.
Table \ref{tb:compare_2act_2stage} shows the estimated
coverage for the coefficient of $A_1$, $\beta_{1,1,1}^*$.  This
simulation uses a sample size of 150, a total of 1000 Monte Carlo
replications and 1000 bootstrap samples.  Target coverage is
$0.95$. The CPB fares poorly in terms of coverage, falling
significantly below nominal coverage on seven of nine examples.  The
ST method fails to cover for examples A, B and C.
Recall that the ST method has not been developed
for the setting in which there are more than two treatments at the
second stage.

\begin{table}
\begin{small}
\begin{center}
\begin{tabular}{cllllll|lll}
\parbox{5em}{\centering Two txts\\at stage 2} &
\parbox{2.9em}{\centering{Ex. 1\\NR}} &
\parbox{2.9em}{\centering{Ex. 2\\NNR}} &
\parbox{2.9em}{\centering{Ex. 3\\NR}} &
\parbox{2.9em}{\centering{Ex. 4\\NNR}} &
\parbox{2.9em}{\centering{Ex. 5\\NR}} &
\parbox{2.9em}{\centering{Ex. 6\\R}} &
\parbox{2.9em}{\centering{Ex. A\\R}} &
\parbox{2.9em}{\centering{Ex. B\\NR}}&
\parbox{2.9em}{\centering{Ex. C\\NNR}}
\vspace{0.15em} \\
\hline
          CPB&  0.934* &  0.935* &  0.930* &  0.933* &  0.938  &  0.928* &  0.939  &  0.925* &  0.928* \\
         FACI&  0.989  &  0.987  &  0.967  &  0.969  &  0.954  &  0.952  &  0.950  &  0.962  &  0.962  \\
         DACI&  0.968  &  0.971  &  0.958  &  0.961  &  0.949  &  0.943  &  0.949  &  0.953  &  0.953  \\
         MOFN&  0.965  &  0.966  &  0.957  &  0.958  &  0.952  &  0.945  &  0.949  &  0.954  &  0.959  \\
           ST&  0.948  &  0.945  &  0.938  &  0.942  &  0.952  &  0.943  &  0.919* &  0.759* &  0.762*
\end{tabular}\\
\vskip1em

All three of the FACI, DACI, and MOFN methods deliver nominal coverage
on all of the examples.  The FACI in particular is conservative on
examples one and two.  The average interval diameters are shown in
Table \ref{tb:compare_2act_2stage_loglogn_widthsW}; this is to be
expected given that it is based on upper and lower bounds. However, we
note that the DACI, whose $\lambda_n$ is tuned using the double
bootstrap, has a much smaller width than the FACI, particularly in the
three-treatment examples. It is the narrowest among the methods that
cover in all examples.

\begin{tabular}{cllllll|lll}
\parbox{5em}{\centering Three txts\\at stage 2} &
\parbox{2.9em}{\centering{Ex. 1\\NR}} &
\parbox{2.9em}{\centering{Ex. 2\\NNR}} &
\parbox{2.9em}{\centering{Ex. 3\\NR}} &
\parbox{2.9em}{\centering{Ex. 4\\NNR}} &
\parbox{2.9em}{\centering{Ex. 5\\NR}} &
\parbox{2.9em}{\centering{Ex. 6\\R}} &
\parbox{2.9em}{\centering{Ex. A\\R}} &
\parbox{2.9em}{\centering{Ex. B\\NR}}&
\parbox{2.9em}{\centering{Ex. C\\NNR}}
\vspace{0.15em} \\
\hline
         CPB &  0.933* &  0.938  &  0.915* &  0.921* &  0.931* &  0.907* &  0.940  &  0.885* &  0.895* \\
         FACI&  0.999  &  0.999  &  0.967  &  0.968  &  0.963  &  0.969  &  0.958  &  0.969  &  0.969* \\
         DACI&  0.987  &  0.987  &  0.952  &  0.955  &  0.957  &  0.945  &  0.953  &  0.940  &  0.945  \\
\end{tabular}
\end{center}
\caption{\label{tb:compare_2act_2stage}
     Monte
    Carlo estimates of coverage probabilities of confidence intervals
    for the main effect of treatment, $\beta_{1,1,1}^*$ at
    the $95\%$ nominal level. Estimates are constructed using 1000
    datasets of size 150 drawn from each model, and 1000 bootstraps
    drawn from each dataset. Estimates significantly below
    $0.95$ at the $0.05$ level are marked with $*$. There is no ST or MOFN method when there are three treatments at Stage 2. Examples are
    designated NR = nonregular, NNR = near-nonregular, R = regular.}
\end{small}

\begin{small}
\begin{center}
\begin{tabular}{cllllll|lll}
\parbox{5em}{\centering Two txts\\at stage 2} &
\parbox{2.9em}{\centering{Ex. 1\\NR}} &
\parbox{2.9em}{\centering{Ex. 2\\NNR}} &
\parbox{2.9em}{\centering{Ex. 3\\NR}} &
\parbox{2.9em}{\centering{Ex. 4\\NNR}} &
\parbox{2.9em}{\centering{Ex. 5\\NR}} &
\parbox{2.9em}{\centering{Ex. 6\\R}} &
\parbox{2.9em}{\centering{Ex. A\\R}} &
\parbox{2.9em}{\centering{Ex. B\\NR}}&
\parbox{2.9em}{\centering{Ex. C\\NNR}}
\vspace{0.15em} \\
\hline
         CPB &  0.385* &  0.385* &  0.430* &  0.430* &  0.457  &  0.436* &  0.451  &  0.428* &  0.428* \\
         FACI&  0.490  &  0.490  &  0.481  &  0.481  &  0.483  &  0.471  &  0.474  &  0.484  &  0.484  \\
         DACI&  0.442  &  0.441  &  0.470  &  0.470  &  0.482  &  0.469  &  0.474  &  0.473  &  0.473  \\
         MOFN&  0.443  &  0.443  &  0.474  &  0.474  &  0.489  &  0.486  &  0.482  &  0.488  &  0.488  \\
           ST&  0.339  &  0.339  &  0.426  &  0.427  &  0.469  &  0.436  &  0.480* &  0.426* &  0.424*
\end{tabular}\\
\vskip1em
\begin{tabular}{cllllll|lll}
\parbox{5em}{\centering Three txts\\at stage 2} &
\parbox{2.9em}{\centering{Ex. 1\\NR}} &
\parbox{2.9em}{\centering{Ex. 2\\NNR}} &
\parbox{2.9em}{\centering{Ex. 3\\NR}} &
\parbox{2.9em}{\centering{Ex. 4\\NNR}} &
\parbox{2.9em}{\centering{Ex. 5\\NR}} &
\parbox{2.9em}{\centering{Ex. 6\\R}} &
\parbox{2.9em}{\centering{Ex. A\\R}} &
\parbox{2.9em}{\centering{Ex. B\\NR}}&
\parbox{2.9em}{\centering{Ex. C\\NNR}}
\vspace{0.15em} \\
\hline
         CPB &  0.446* &  0.446  &  0.518* &  0.518* &  0.567* &  0.518* &  0.557  &  0.508* &  0.507* \\
         FACI&  0.700  &  0.700  &  0.652  &  0.652  &  0.637  &  0.632  &  0.617  &  0.661  &  0.662  \\
         DACI&  0.564  &  0.564  &  0.590  &  0.590  &  0.617  &  0.591  &  0.604  &  0.596  &  0.597  \\
\end{tabular}
\end{center}
\caption{\label{tb:compare_2act_2stage_loglogn_widthsW}Monte
    Carlo estimates of the mean width  of confidence intervals for the
    main effect of treatment $\beta_{1,1,1}^*$ at the $95\%$ nominal level. Estimates are
    constructed using 1000 datasets of size 150 drawn from each
    model, and 1000 bootstraps drawn from each dataset. Models have
    two treatments at each of two stages. Widths with corresponding
    coverage significantly below nominal are marked with $*$.
There is no ST or MOFN method when there are three treatments at Stage 2.   Examples are designated NR = nonregular, NNR = near-nonregular, R = regular.}
 \end{small}
\end{table}

\begin{table}
\begin{small}
\begin{center}
\begin{tabular}{cllllll|lll}
\parbox{5em}{\centering Two txts\\at stage 2} &
\parbox{2.9em}{\centering{Ex. 1\\NR}} &
\parbox{2.9em}{\centering{Ex. 2\\NNR}} &
\parbox{2.9em}{\centering{Ex. 3\\NR}} &
\parbox{2.9em}{\centering{Ex. 4\\NNR}} &
\parbox{2.9em}{\centering{Ex. 5\\NR}} &
\parbox{2.9em}{\centering{Ex. 6\\R}} &
\parbox{2.9em}{\centering{Ex. A\\R}} &
\parbox{2.9em}{\centering{Ex. B\\NR}}&
\parbox{2.9em}{\centering{Ex. C\\NNR}}
\vspace{0.15em} \\
\hline
         CPB &  0.892* &  0.908* &  0.924* &  0.925* &  0.940  &  0.930* &  0.936  &  0.925* &  0.931* \\
         FACI&  0.952  &  0.962  &  0.952  &  0.954  &  0.950  &  0.953  &  0.947  &  0.952  &  0.954  \\
         DACI&  0.940  &  0.946  &  0.946  &  0.948  &  0.947  &  0.945  &  0.951  &  0.952  &  0.947  \\
         MOFN&  0.944  &  0.947  &  0.948  &  0.948  &  0.952  &  0.942  &  0.951  &  0.950  &  0.950  \\
           ST&  0.935* &  0.930* &  0.889* &  0.878* &  0.891* &  0.620* &  0.687* &  0.686* &  0.663*
\end{tabular}\\
\end{center}
\caption{\label{tb:interceptC}
     Monte
    Carlo estimates of coverage probabilities of confidence intervals for the coefficient of the intercept, $\beta_{1,0,1}^*$ at
    the $95\%$ nominal level. Estimates are constructed using 1000
    datasets of size 150 drawn from each model, and 1000 bootstraps
    drawn from each dataset. Estimates significantly below
    $0.95$ at the $0.05$ level are marked with $*$. Examples are
    designated NR = nonregular, NNR = near-nonregular, R = regular.}
\end{small}
\vskip1em
\begin{small}
\begin{center}
\begin{tabular}{cllllll|lll}
\parbox{5em}{\centering Two txts\\at stage 2} &
\parbox{2.9em}{\centering{Ex. 1\\NR}} &
\parbox{2.9em}{\centering{Ex. 2\\NNR}} &
\parbox{2.9em}{\centering{Ex. 3\\NR}} &
\parbox{2.9em}{\centering{Ex. 4\\NNR}} &
\parbox{2.9em}{\centering{Ex. 5\\NR}} &
\parbox{2.9em}{\centering{Ex. 6\\R}} &
\parbox{2.9em}{\centering{Ex. A\\R}} &
\parbox{2.9em}{\centering{Ex. B\\NR}}&
\parbox{2.9em}{\centering{Ex. C\\NNR}}
\vspace{0.15em} \\
\hline
         CPB &  0.404* &  0.404* &  0.430* &  0.429* &  0.457  &  0.449* &  0.450  &  0.428* &  0.428* \\
         FACI&  0.506  &  0.506  &  0.481  &  0.481  &  0.483  &  0.490  &  0.474  &  0.490  &  0.490  \\
         DACI&  0.459  &  0.459  &  0.466  &  0.466  &  0.481  &  0.482  &  0.473  &  0.473  &  0.473  \\
         MOFN&  0.475  &  0.476  &  0.469  &  0.470  &  0.488  &  0.486  &  0.477  &  0.483  &  0.483  \\
           ST&  0.344* &  0.344* &  0.427* &  0.427* &  0.466* &  0.469* &  0.474* &  0.430* &  0.428*
\end{tabular}\\
\end{center}
\caption{\label{tb:interceptW}Monte
    Carlo estimates of the mean width  of confidence intervals for the coefficient of the intercept, $\beta_{1,0,1}^*$ at the $95\%$ nominal level. Estimates are
    constructed using 1000 datasets of size 150 drawn from each
    model, and 1000 bootstraps drawn from each dataset. Models have
    two treatments at each of two stages. Widths with corresponding
    coverage significantly below nominal are marked with $*$.
   Examples are designated NR = nonregular, NNR = near-nonregular, R = regular.}
 \end{small}
\end{table}

The coefficient of $A_1$ is perhaps most relevant from a clinical
perspective.  However, from a methodological point of view, other
contrasts can be illuminating.  Table \ref{tb:interceptC} shows the
estimated coverage for the intercept using the same generative models.
The coverage of the CPB and ST methods is quite poor; the CPB attains
nominal coverage on only two of the nine examples, and the ST never
achieves nominal coverage.  Particularly disturbing is that the
ST method falls more than 30\% below nominal levels.  In contrast, the
FACI and DACI deliver nominal coverage on all examples.  Table
\ref{tb:interceptW} shows the average interval widths; the DACI is the
narrowest among the covering methods.

\section{Analysis of the ADHD study}

In this section we illustrate the use of the ACI on data from the
Adaptive Pharmacological and Behavioral Treatments for Children with
ADHD Trial (Nahum-Shani et al. 2012a; Lei et al. 2012).  The ADHD data
we use here consists of $n=138$ trajectories which are a subset of the
original $N = 155$ observations.  This subset was formed by removing
the $N - n = 17$ subjects who were either never randomized to an
initial treatment (14 subjects), or had massive item missingness (3
subjects).  A description of each of the variables is provided in
Table \ref{adhdTable}.
\begin{table}[here!]
\begin{center}
\begin{tabular}{|ccp{12.5cm}|}  \hline
  $X_{1,1} \in [0,3]$  &:& Baseline symptoms. Teacher-reported mean ADHD
symptom  score.  Measured at the end of the school year preceding
the study. \\
$X_{1,2}\in \lbrace 0, 1 \rbrace$ & : & ODD diagnosis. Indicator of a diagnosis of
ODD (oppositional defiant disorder) at baseline, coded so that $0$
corresponds to no such diagnosis.   \\
$X_{1,3}\in \lbrace 0, 1 \rbrace$ & : & Prior med.\ exposure. Indicator that subject received ADHD medication in the prior year, coded so that $ 0$
corresponds to no ADHD medication.   \\
$A_1 \in \lbrace -1, 1\rbrace$ & : & 1st stage treatment. Coded so
that $-1$ corresponds to medication while $1$ corresponds
to behavioral modification therapy. \\



$1_{\mathrm{NonRsp}}$ & : & Indicator of non-response, i.e.\ that a
patient was re-randomized to a second-stage treatment during the
study. Non-response was determined on the basis of two measures the Impairment Rating Scale (IRS)
(Fabiano et al. 2006) and an individualized list of target behaviors
(ITB) (e.g., Pelham et al.\ 1992).  The criterion for nonresponse at
each month was an average performance of less than 75
on the ITB and  a rating of impairment in at least one domain
on the IRS. These were measured beginning in week 8 of the study, and
montly thereafter.\\
$Y_1 \triangleq  Y\cdot(1 - 1_{\mathrm{NonRsp}})$ & : & First stage outcome of
responders, i.e.\ those who were {\em not} re-randomized (see definition of
$Y$ and $\tilde Y$ below).  \\
$X_{2,1}\in\lbrace 0, 1\rbrace$ & : & Adherence. Indicator of subject's adherence to
their initial treatment.  Adherence is coded so that a value of
$ 0$ corresponds to low adherence (taking less than 100\% of
prescribed medication or attending less than 75\% of therapy sessions)
while a value of
$1$ corresponds to high adherence.  \\
$X_{2,2}\in\lbrace 2, 8\rbrace$ & : & Month of non-response. Month during school year of
observed non-response and re-randomization (not used for responders) Two subjects did not follow protocol
 and were re-randomized during month 8.\\


$A_2 \in \lbrace -1, 1\rbrace$ & : & 2nd stage treatment. Coded so
that $A_2 = -1$ corresponds to augmenting the initial
treatment with the treatment \textit{not} received initially, and $A_2 = 1$ corresponds to enhancing (increasing the dosage of) the
initial treatment. \\
$Y \in \lbrace 1, 2, \ldots, 5\rbrace$ & : & Teacher-reported
Teacher Impairment Rating Scale (TIRS5) item score 8 months (32 weeks) after
initial randomization to treatment (Fabiano et al.\ 2006).  The TIRS5 is coded so that
higher values correspond to better clinical outcomes. \\
$Y_2 \triangleq  Y\cdot 1_{NonRsp}$ & : & Second stage outcome.
Only used for non-responders, i.e.\ subjects who {\em were} re-randomized. \\ \hline
\end{tabular}
\end{center}
\caption{Features, treatments and the outcome for the ADHD study.}
\label{adhdTable}
\end{table}
Notice that the outcomes $Y_1$ and $Y_2$ satisfy $Y_1 + Y_2 \equiv Y$,
where $Y$ is the teacher reported TIRS5 score after 32 weeks, i.e.\ at the end of the last
month of the study (month 8).

The first step in using $Q$-learning is to estimate
a regression model for the second stage; this analysis only uses data
from subjects that were re-randomized during the 8 month study.  Of
the $n = 138$ subjects, $81$ were re-randomized prior to the end of the study.  The feature vectors at the second stage are $H_{2,0}
\triangleq (1, X_{1,1}, X_{1,2}, X_{1,2}, X_{1,3}, X_{2,1}, A_1)^{\T}$
and $H_{2,1} \triangleq (1, X_{2,1}, A_1)^{\T}$.  Thus, the
$Q$-function $Q_2(H_2, A_2; \beta_2) \triangleq
H_{2,0}^{\T}\beta_{2,0} + H_{2,1}^{\T}\beta_{2,1}A_2$ contains an
interaction term between the second stage action $A_2$ and a subject's
initial treatment $A_1$, an interaction between $A_2$ and adherence to
their initial medication $X_{2,1}$, a main effect for $A_2$, and main
effects for all the other terms.  Table
\ref{secondStageFitCoefficients} provides the second stage least
squares coefficients along with centered percentile bootstrap interval
estimates.  Examination of the residuals (not shown here) showed no
obvious signs of model misspecification.  In short, the linear model
described above seems to fit the data reasonably well.

\begin{table}[here!]
\begin{center}
\begin{tabular}{llcccc}
Term  &   & Coeff. & Estimate & Lower  (5\%) & Upper  (95\%) \\ \hline
1 & Intercept      & $\beta_{2,0,1}$ & 1.36 & 0.48 & 2.26 \\
$X_{1,1}$  & Baseline symptoms & $\beta_{2,0,2}$ & 0.94 & 0.48 & 1.39 \\
$X_{1,2}$ & ODD diagnosis     & $\beta_{2,0,3}$ & 0.92 & 0.46 & 1.41 \\
$X_{1,3}$ & Prior med. exposure  & $\beta_{2,0,4}$ &-0.27 &-0.77 & 0.21 \\
$X_{2,1}$ & Adherence         & $\beta_{2,0,5}$ & 0.17 &-0.28 & 0.66 \\
$X_{2,2}$ & Month of non-response & $\beta_{2,0,6}$ & 0.02 &-0.20 & 0.20 \\
$A_1$ & 1st stage  txt  & $\beta_{2,0,7}$ & 0.03 &-0.18 & 0.23 \\
$A_2$ & 2nd stage txt  & $\beta_{2,1,1}$ &-0.72 &-1.13 & -0.35 \\
$A_2:X_{2,1}$ & 2nd stage txt~: Adherence  & $\beta_{2,1,2}$ & 0.97 & 0.48 & 1.52 \\
$A_2:A_1$ & 2nd stage txt~: 1st stage txt  & $\beta_{2,1,3}$ & 0.05 &-0.17& 0.27 \\
\end{tabular}
\end{center}
\caption{Least squares coefficients and 90\% CPB interval estimates for
second stage regression. }
\label{secondStageFitCoefficients}
\end{table}
Recall that the dependent variable in the first stage regression model
is the predicted future outcome $\tilde{Y}_1 \triangleq Y_1 +
\max_{a_2\in\lbrace -1, 1\rbrace} Q_2 (H_2, a_2; \bHat_2)$.  Since the
predictors used in the first stage must predate the assignment of
first treatment, the available predictors in Table \ref{adhdTable} are
baseline ADHD symptoms $X_{1,1}$, diagnosis of ODD at baseline
$X_{1,2}$, indicator of a subject's prior exposure to ADHD medication
$X_{1,3}$, and first stage treatment $A_1$.  The feature vectors for
the second stage are $H_{1,0} \triangleq (1, X_{1,1}, X_{1,2},
X_{1,3})$ and $ H_{1,1} \triangleq (1, X_{1,3})$, so that the first
stage $Q$-function $Q_1(H_1, A_1;\beta_1) \triangleq
H_{1,0}^{\T}\beta_{1,0} + H_{1,1}^{\T}\beta_{1,1}A_1$ contains an
interaction term between the first stage action $A_1$ and a subject's
prior exposure to ADHD medication $X_{1,3}$, a main effect for $A_1$,
and main effects for all other covariates.  The first stage regression
coefficients are estimated using least squares $\bHat_1 \triangleq
\arg\min_{\beta_1} \pn (\tilde{Y}_1 - Q_1(H_1, A_1;\beta_1))^2$.
Table \ref{firstStageFitCoefficients} provides the least squares
coefficients along with interval estimates formed using the DACI.
Plots of the residuals for this model (not shown here) show no obvious
signs of model misspecification.  Again a linear model seems to
provide a reasonable approximation to the $Q$-function in the first
stage.

\begin{table}[here!]
\begin{center}
\begin{tabular}{llcccc}
Term &    & Coeff.   & Estimate & Lower  (5\%) & Upper  (95\%) \\ \hline
$1$ & Intercept           & $\beta_{1,0,1}$ & 2.61 & 2.13 & 3.05 \\
$X_{1,1}$ & Baseline symptoms   & $\beta_{1,0,2}$ & 0.72 & 0.47 & 1.00 \\
$X_{1,2}$ & ODD diagnosis       & $\beta_{1,0,3}$ & 0.75 & 0.37 & 1.08 \\
$X_{1,3}$ & Prior med. exposure & $\beta_{1,0,4}$ &-0.37 &-0.80 & 0.01 \\
$A_1$ & Initial txt         & $\beta_{1,1,1}$ & 0.17 &-0.02 & 0.36 \\
$A_1 : X_{1,3}$ & Initial txt~: Prior med. exposure
                    & $\beta_{1,1,2}$ &-0.32 &-0.59 &-0.07 \\
\end{tabular}
\end{center}
\caption{Least squares coefficients and 90\% DACI interval estimates for
first stage regression. }
\label{firstStageFitCoefficients}
\end{table}

To
 construct an estimate of the optimal DTR, recall that for
any $H_t = h_t,~t=1,2$ the estimated optimal DTR $\hat{\pi} = (\hat{\pi}_1,
\hat{\pi}_2)$ satisfies $\hat{\pi}_t(h_t) \in \arg\max_{a_t}Q(h_t,a_t;\bHat_t)$.
The coefficients in Table \ref{secondStageFitCoefficients} and the
form of the second stage $Q$-function reveal that the second stage
decision rule $\hat{\pi}_2$ is quite simple.  In particular,
 $\hat{\pi}_2$ prescribes
treatment enhancement to subjects with high adherence to their initial
medication and it prescribes treatment augmentation to subjects with
low adherence to their initial medication.   The first stage decision
rule $\hat{\pi}_1$ is equally simplistic.  The coefficients in Table
\ref{firstStageFitCoefficients} show that the first stage decision
rule,
$\hat{\pi}_1$
prescribes medication to subjects who have had prior exposure to
medication, and behavioral modification to subjects who have not had
any such prior exposure.

The prescriptions given by the estimated optimal DTR $\hat{\pi}$ are
excessively decisive.  That is, they recommend one and only one
treatment regardless of the amount of evidence in the data to support
that the recommended treatment is in fact optimal.  When there is
insufficient evidence to recommend a single treatment as best for a
given patient history, it is preferred to leave the choice of
treatment to the clinician.  This allows the clinician to recommend
treatment based on cost, local availability, patient individual
preference, and clinical experience.  One way to assess if there is
sufficient evidence to recommend a unique optimal treatment for a
patient is to construct a confidence interval for the predicted
difference in mean response across treatments.  In the case of binary
treatments, for a fixed patient history $H_t = h_t$, one would construct
a confidence interval for the difference $Q_t(h_t, 1;\beta_t^*) -
Q_t(h_1, -1; \beta_t^*) = c^{\T}\beta_{t}^*$ where
$c = (\mathbf{0}^{\T}, 2h_{t,1}^{\T})^{\T}$.  If this confidence interval
contains zero then one would conclude that there is insufficient evidence
at the nominal level for a unique best treatment.

In this example,  the subject features that interact with
treatment are categorical. Consequently, we can construct confidence
intervals for the predicted difference in mean response across
treatments for every possible subject history.  These confidence
intervals are given in table (\ref{evidenceAssessmentInd}).
The
$90\%$ confidence intervals suggest that there is insufficient evidence at
the first stage to recommend a unique best treatment for each subject
history.  Rather, we would prefer not to make a strong recommendation
at stage one, and leave treatment choice solely at the discretion of
the clinician.  Conversely, in the second stage, the $90\%$ confidence
intervals suggest that there is evidence to recommend a
unique best treatment when a subject had low adherence---knowledge
that is important for evidence-based clinical decision making.
\begin{table}[here!]
\begin{center}
\begin{tabular}{cp{2.75cm}p{1.5cm}ccc}
Stage & History & Contrast for $\beta_{t,1}$ & Lower (5\%) & Upper (95\%) & Conclusion \\
\hline
1     & Had prior med. & (2 2) & -0.88 &  0.28 & Insufficient evidence\\
1     & No prior med.  & (2 0) & -0.04 &  0.72 & Insufficient
evidence\\
\hline
2     & High adherence and BMOD & (2 2  2) & -0.17 & 1.39 & Insufficient evidence \\
2     & Low adherence and BMOD  & (2 0  2) & -2.21 &-0.57 & Sufficient evidence \\
2     & High adherence and MEDS & (2 2 -2) & -0.37 & 1.26 & Insufficient evidence \\
2     & Low adherence and MEDS  & (2 0 -2) & -2.51 &-0.60 & Sufficient evidence \\
\end{tabular}
\end{center}
\caption{Confidence intervals for the predicted difference in mean response
  across treatments for each possible patient history. Intervals are
  at the 90\% leve.
  Confidence intervals that  contain zero indicate insufficient evidence for recommending a unique best
  treatment for patients with the given history. \label{evidenceAssessmentInd}}
\end{table}

\section{Summary, open problems, and the future of DTRs}
Nonregularity often arises in estimators of optimal DTRs.  We
discussed how nonregularity leads to asymptotic bias and complicates
inference.  Asymptotic bias can be reduced by applying shrinkage
methods; however, tuning these methods is an open problem, and
over-shrinkage can be infinitely worse than no shrinkage
at all.  We proposed the ACI, a locally consistent method for
constructing confidence intervals for first stage parameters in
$Q$-learning.  The ACI uses analytic bounds on
$c^{\T}\rtn(\widehat{\beta}_1-\beta_1^*)$.  However, a potentially
less conservative strategy would be to form bounds on the
$(\alpha/2)\times 100$ and $(1-\alpha/2)\times 100$ {\em percentiles}
of the sampling distribution of $c^{\T}\rtn(\widehat{\beta}_1
-\beta_1^*)$.  For example, one could define $\mathcal{B}(c,\gamma) =
c^{\T}\mathbb{S}_n + c^{\T}\hat{\Sigma}_{1}^{-1}\pn
B_{1}\mathbb{U}_n1_{\hat{T}(H_{2,1}) > \lambda_n} +
c^{\T}\hat{\Sigma}_{1}^{-1}\pn B_{1} \left(
  \left[H_{2,1}^{\T}(\mathbb{V}_n + \gamma)\right]_+ -
  \left[H_{2,1}^{\T}\gamma\right]_+ \right)1_{\hat{T}(H_{2,1}) \le
  \lambda_n}$.  Then, for any fixed $\gamma$ and level $\eta$ one
could use the bootstrap to estimate the $\eta\times 100$ percentile of
$\mathcal{B}(c,\gamma)$, say, $\widehat{q}_{\eta}^{(b)}(\gamma)$.  The
final confidence interval would be $\big(c^{\T}\widehat{\beta}_1 -
\sup_{\gamma \in
  \mathbb{R}^{\dim(\beta_{2,1}^*)}}\widehat{q}_{1-\alpha/2}^{(b)}(\gamma),
\allowbreak c^{\T}\widehat{\beta}_1 -
\inf_{\gamma\in\mathbb{R}^{\dim(\beta_{2,1}^*}}
q_{\alpha/2}^{(b)}(\gamma) \big)$.  See \cite[][]{andrews, cheng} and
references therein for bounding probabilities rather than statistics.
It would be interesting to compare this approach with the ACI.

In our development we assumed that the features $H_t$ were known {\em
  a priori.} However, in many practical examples, including the one we
considered here, $H_t$ is a heuristic low-dimensional representation
of hundreds or even thousands of sparsely observed and irregularly
spaced measurements.  By design, information is accumulating over
time, if one uses linear models nested inside the sequence of
treatments received, then the model size will grow exponentially in
the number of treatment stages.  Principled, i.e., data-driven,
methods for feature construction and extraction are needed.  On
approach would be to extend dimensionality-reduction methods from
machine learning (e.g., isomap, ICA, etc.) or functional data analysis
(e.g., functional principle components) to DTRs.

DTRs have the potential to produce better patient outcomes while
simultaneously reducing cost and patient burden.  Furthermore,
estimated optimal DTRs can provide important scientific insight by
revealing interactions between treatments and patient history and
delayed treatment effects.  However, technological advances are
continually improving the efficiency with which data can be collected,
stored, and accessed. DTR methodologies must adapt with these changes.
Here we discuss two emerging areas where current DTR methodology is
insufficient.  Both areas present unique
estimation, inference, and computational challenges.

\textit{Infinite horizon problems.}  In settings where number of
treatment stages is large (e.g., hundreds or thousands) it may be
appropriate to approximate the decision problems as having an infinite
number of time points.  An important area where such decision problems
arise is mobile-health (mHealth) where interventions are delivered
using smartphones or other mobile devices \citep[see, for
example,][]{kelly2012intelligent}.  Mobile devices present
unprecedented opportunity for collecting patient information and
delivering interventions \textit{in situ}, and thereby potentially
narrowing the so-called research-practice
gap~\citep{bickman2012technology}.  However, the breadth of
opportunities presented by mHealth are matched by their technical
challenges.  As the number of decision points grows large it becomes
infeasible to have separate models for the $Q$-function at each
decision point, in this case additional structure, for example, that
the generative model can be characterized as a stationary Markov
Decision Process \citep[MDP,][]{putterman1994markov}, is useful.
Existing methods for estimating an optimal DTR in the MDP setup
\citep{sutton} are highly algorithmic and their statistical properties
are largely unknown.  There are tremendous opportunities for
translating these algorithms into a statistical framework and
characterizing their statistical properties, e.g., convergence rates
and limiting distribution theory.

\textit{Spatial decision processes.}   In some applications, for
example, adaptive wildlife management,  separate treatments must be
administered across a series of spatial locations at each time point.
The treatment assignment at one spatial location may affect the outcomes at
neighboring locations.  Furthermore, the total number of treatments
than can be administered across all the spatial locations is often
limited by budget or other resource constraints. Thus,
it is not feasible to estimate a separate DTR at each spatial location but
rather a single large DTR recommending treatments for all spatial
locations simultaneously is needed.  That is, a DTR in this setting
is a sequence of functions mapping up-to-date information at all
spatial locations to a treatment recommendation at every spatial
location.   $Q$-learning, as described, cannot be applied as the
dimension of the model grows exponentially in the number of spatial
locations.  Suppose, for example, that there are $S$ spatial locations, $K$
treatment options available at each location, and a $p$-dimensional
feature vector at each spatial location; a linear model with a
main effect of feature, a main effect for treatment, and an
interaction between treatments and features would contain
$p\times K^S$ terms.  Furthermore, even if the $Q$-functions were
known exactly, simply computing the argmax over all $K^S$ possibilities
is computationally intractable for moderate values of $S$ and $K$.

\bibliographystyle{plainnat}
\bibliography{ejs}

\begin{thebibliography}{73}
\providecommand{\natexlab}[1]{#1}
\providecommand{\url}[1]{\texttt{#1}}
\expandafter\ifx\csname urlstyle\endcsname\relax
  \providecommand{\doi}[1]{doi: #1}\else
  \providecommand{\doi}{doi: \begingroup \urlstyle{rm}\Url}\fi

\bibitem[Andrews and Soares(2007)]{andrewsSoares}
Donald~W. Andrews and Gustavo Soares.
\newblock {Inference for Parameters Defined by Moment Inequalities Using
  Generalized Moment Selection}.
\newblock \emph{SSRN eLibrary}, 2007.

\bibitem[Andrews(2000)]{andrews2000inconsistency}
Donald~WK Andrews.
\newblock Inconsistency of the bootstrap when a parameter is on the boundary of
  the parameter space.
\newblock \emph{Econometrica}, 68\penalty0 (2):\penalty0 399--405, 2000.

\bibitem[Andrews(2001{\natexlab{a}})]{andrews}
Donald~W.K. Andrews.
\newblock Testing when a parameter is on the boundary of the maintained
  hypothesis.
\newblock \emph{Econometrica}, 69:\penalty0 683--734, 2001{\natexlab{a}}.

\bibitem[Andrews(2001{\natexlab{b}})]{andrews2001testing}
D.W.K. Andrews.
\newblock Testing when a parameter is on the boundary of the maintained
  hypothesis.
\newblock \emph{Econometrica}, 69\penalty0 (3):\penalty0 683--734,
  2001{\natexlab{b}}.

\bibitem[Andrews and Guggenberger(2009)]{andrews2009incorrect}
D.W.K. Andrews and P.~Guggenberger.
\newblock Incorrect asymptotic size of subsampling procedures based on
  post-consistent model selection estimators.
\newblock \emph{Journal of Econometrics}, 152\penalty0 (1):\penalty0 19--27,
  2009.

\bibitem[Anthony and Bartlett(1999)]{ab-nnltf-99}
M.~Anthony and P.L. Bartlett.
\newblock \emph{Neural Network Learning: Theoretical Foundations}.
\newblock Cambridge University Press, 1999.

\bibitem[Barto and Dieterich(2004)]{barto2004}
AG~Barto and T~Dieterich.
\newblock Reinforcement learning and its relation to supervised learning.
\newblock \emph{Handbook of Learning and Approximate Dynamic Programming},
  pages 45--63, 2004.

\bibitem[Bellman(1957)]{bellman}
R.E. Bellman.
\newblock \emph{Dynamic Programming}.
\newblock Princeton University Press, 1957.

\bibitem[Bickel(1981)]{bickel1981}
PJ~Bickel.
\newblock Minimax estimation of the mean of a normal distribution when the
  parameter space is restricted.
\newblock \emph{The Annals of Statistics}, 9\penalty0 (6):\penalty0 1301--1309,
  1981.

\bibitem[Bickel and Freedman(1981)]{bickel1981some}
P.J. Bickel and D.A. Freedman.
\newblock Some asymptotic theory for the bootstrap.
\newblock \emph{The Annals of Statistics}, pages 1196--1217, 1981.

\bibitem[Bickman et~al.(2012)Bickman, Kelley, and Athay]{bickman2012technology}
Leonard Bickman, Susan~Douglas Kelley, and Michele Athay.
\newblock The technology of measurement feedback systems.
\newblock \emph{Couple and Family Psychology: Research and Practice},
  1\penalty0 (4):\penalty0 274--284, 2012.

\bibitem[Blumenthal and Cohen(1968)]{blume1968}
Saul Blumenthal and Arthur Cohen.
\newblock Estimation of the larger of two normal means.
\newblock \emph{Journal of the American Statistical Association}, pages
  861--876, 1968.

\bibitem[Busoniu et~al.(2010)Busoniu, Babuska, De~Schutter, and
  Ernst]{busoniu2010}
Lucian Busoniu, Robert Babuska, Bart De~Schutter, and Damien Ernst.
\newblock \emph{Reinforcement learning and dynamic programming using function
  approximators}.
\newblock CRC Press, 2010.

\bibitem[Casella and Strawderman(1981)]{casella1981}
George Casella and William~E Strawderman.
\newblock Estimating a bounded normal mean.
\newblock \emph{The Annals of Statistics}, pages 870--878, 1981.

\bibitem[Chakraborty et~al.(2009)Chakraborty, Murphy, and
  Strecher]{chakraborty2009inference}
B.~Chakraborty, S.~Murphy, and V.~Strecher.
\newblock {Inference for non-regular parameters in optimal dynamic treatment
  regimes}.
\newblock \emph{Statistical Methods in Medical Research}, 19\penalty0 (3),
  2009.

\bibitem[Chakraborty et~al.(2013)Chakraborty, Laber, and Zhao]{mofn}
B.~Chakraborty, E.B. Laber, and Y.~Zhao.
\newblock Inference for optimal dynamic treatment regimes using an adaptive
  m-out-of-n bootstrap scheme.
\newblock \emph{Biometrics}, TBA\penalty0 (TBA):\penalty0 TBA, 2013.

\bibitem[Chakraborty and Moodie(2013)]{bibhasBook}
Bibhas Chakraborty and Erica~EM Moodie.
\newblock \emph{Statistical Methods for Dynamic Treatment Regimes}.
\newblock Springer, 2013.

\bibitem[Chakraborty and Murphy(2014)]{bibhasChapter}
Bibhas Chakraborty and Susan~A. Murphy.
\newblock Dynamic treatment regimes.
\newblock \emph{Annual Review of Statistics and Its Application}, 1\penalty0
  (1):\penalty0 null, 2014.
\newblock \doi{10.1146/annurev-statistics-022513-115553}.
\newblock URL
  \url{http://www.annualreviews.org/doi/abs/10.1146/annurev-statistics-022513-%
115553}.

\bibitem[Chen(2004)]{chen2004}
Jeesen Chen.
\newblock Notes on the bias-variance trade-off phenomenon.
\newblock \emph{A Festschrift for Herman Rubin: Institute of Mathematical
  Statistics}, 45:\penalty0 207--217, 2004.

\bibitem[Cheng(2008)]{cheng}
Xu~Cheng.
\newblock Robust confidence intervals in nonlinear regression under weak
  identification.
\newblock \emph{Job Market Paper}, 2008.

\bibitem[Cs\"{o}rg\H{o} and Rosalsky(2003)]{csorgo}
S\'{a}ndor Cs\"{o}rg\H{o} and Andrew Rosalsky.
\newblock A survey of limit laws for bootstrapped sums.
\newblock \emph{International Journal of Mathematics and Mathematical
  Statistics}, 45:\penalty0 2835--2861, 2003.

\bibitem[Davison and Hinkley(1997)]{davison1997}
Anthony~Christopher Davison and David~Victor Hinkley.
\newblock \emph{Bootstrap methods and their application}, volume~1.
\newblock Cambridge university press, 1997.

\bibitem[Doss and Sethuraman(1989)]{sethuramanDoss}
Hani Doss and Jayaram Sethuraman.
\newblock The price of bias reduction when there is no unbiased estimate.
\newblock \emph{Annals of Statistics}, 17\penalty0 (1):\penalty0 440--442,
  1989.

\bibitem[Dusseldorp and Van~Mechelen(2013)]{dusseldorp2013qualitative}
Elise Dusseldorp and Iven Van~Mechelen.
\newblock Qualitative interaction trees: a tool to identify qualitative
  treatment--subgroup interactions.
\newblock \emph{Statistics in medicine}, 2013.

\bibitem[Foster et~al.(2011)Foster, Taylor, and Ruberg]{foster2011}
Jared~C Foster, Jeremy~MG Taylor, and Stephen~J Ruberg.
\newblock Subgroup identification from randomized clinical trial data.
\newblock \emph{Statistics in medicine}, 30\penalty0 (24):\penalty0 2867--2880,
  2011.

\bibitem[Goldberg et~al.(2012)Goldberg, Song, and
  Kosorok]{goldberg2012adaptive}
Yair Goldberg, Rui Song, and Michael~R Kosorok.
\newblock Adaptive q-learning.
\newblock \emph{From Probability to Statistics and Back: High-Dimensional
  Models and Processes}, page 150, 2012.

\bibitem[Gunter et~al.(2011)Gunter, Zhu, and Murphy]{Gunter}
L~Gunter, J~Zhu, and SA~Murphy.
\newblock Variable selection for qualititative interactions.
\newblock \emph{Statistical Methodology}, 8\penalty0 (1):\penalty0 42--55,
  2011.

\bibitem[Henderson et~al.(2009)Henderson, Ansell, and
  Alshibani]{henderson2009regret}
R.~Henderson, P.~Ansell, and D.~Alshibani.
\newblock {Regret-Regression for Optimal Dynamic Treatment Regimes}.
\newblock \emph{Biometrics}, 66\penalty0 (4), 2009.

\bibitem[Hern{\'a}n et~al.(2006)Hern{\'a}n, Lanoy, Costagliola, and
  Robins]{hernan2006}
Miguel~A Hern{\'a}n, Emilie Lanoy, Dominique Costagliola, and James~M Robins.
\newblock Comparison of dynamic treatment regimes via inverse probability
  weighting.
\newblock \emph{Basic \& clinical pharmacology \& toxicology}, 98\penalty0
  (3):\penalty0 237--242, 2006.

\bibitem[Hern{\'a}n et~al.(2000)Hern{\'a}n, Brumback, and Robins]{hernan2000}
Miguel~{\'A}ngel Hern{\'a}n, Babette Brumback, and James~M Robins.
\newblock Marginal structural models to estimate the causal effect of
  zidovudine on the survival of hiv-positive men.
\newblock \emph{Epidemiology}, 11\penalty0 (5):\penalty0 561--570, 2000.

\bibitem[Hirano and Porter(2009)]{hirano}
Keisuke Hirano and Jack Porter.
\newblock Impossibility results for nondifferentiable functionals.
\newblock Mpra paper, University Library of Munich, Germany, 2009.
\newblock URL \url{http://econpapers.repec.org/RePEc:pra:mprapa:15990}.

\bibitem[Hirano and Porter(2012)]{hirano2012}
Keisuke Hirano and Jack~R Porter.
\newblock Impossibility results for nondifferentiable functionals.
\newblock \emph{Econometrica}, 80\penalty0 (4):\penalty0 1769--1790, 2012.

\bibitem[Janes et~al.(2013)Janes, Brown, Pepe, and Huang]{holly}
Holly Janes, Marshall~D Brown, Margaret Pepe, and Ying Huang.
\newblock Statistical methods for evaluating and comparing biomarkers for
  patient treatment selection.
\newblock 2013.

\bibitem[Kelly et~al.(2012)Kelly, Gooding, Pratt, Ainsworth, Welford, and
  Tarrier]{kelly2012intelligent}
James Kelly, Patricia Gooding, Daniel Pratt, John Ainsworth, Mary Welford, and
  Nicholas Tarrier.
\newblock Intelligent real-time therapy: Harnessing the power of machine
  learning to optimise the delivery of momentary cognitive-behavioural
  interventions.
\newblock \emph{Journal of Mental Health}, 21\penalty0 (4):\penalty0 404--414,
  2012.

\bibitem[Konda and Tsitsiklis(2003)]{konda2003}
Vijay~R Konda and John~N Tsitsiklis.
\newblock Onactor-critic algorithms.
\newblock \emph{SIAM journal on Control and Optimization}, 42\penalty0
  (4):\penalty0 1143--1166, 2003.

\bibitem[Kosorok(2008)]{kosorok}
Michael~R. Kosorok.
\newblock \emph{Introduction to empirical processes and semiparametric
  inference}.
\newblock Springer, 2008.

\bibitem[Laber et~al.(2010)Laber, Qian, Lizotte, and
  Murphy]{laber2010statistical}
Eric Laber, Min Qian, Dan~J Lizotte, and Susan~A Murphy.
\newblock Statistical inference in dynamic treatment regimes.
\newblock \emph{arXiv preprint arXiv:1006.5831}, 2010.

\bibitem[Laber and Murphy(2011)]{laber2011adaptive}
Eric~B Laber and Susan~A Murphy.
\newblock Adaptive confidence intervals for the test error in classification.
\newblock \emph{Journal of the American Statistical Association}, 106\penalty0
  (495):\penalty0 904--913, 2011.

\bibitem[Lavori and Dawson(2000)]{lavori2000design}
P.W. Lavori and R.~Dawson.
\newblock A design for testing clinical strategies: biased adaptive
  within-subject randomization.
\newblock \emph{Journal of the Royal Statistical Society: Series A (Statistics
  in Society)}, 163\penalty0 (1):\penalty0 29--38, 2000.

\bibitem[Leeb and Poetscher(2003)]{leeb2003finite}
H.~Leeb and B.M. Poetscher.
\newblock The finite-sample distribution of post-model-selection estimators and
  uniform versus nonuniform approximations.
\newblock \emph{Econometric Theory}, 19\penalty0 (1):\penalty0 100--142, 2003.

\bibitem[Lei et~al.(2012)Lei, Nahum-Shani, Lynch, Oslin, and
  Murphy]{lei2012smart}
H~Lei, I~Nahum-Shani, K~Lynch, D~Oslin, and SA~Murphy.
\newblock A “smart” design for building individualized treatment sequences.
\newblock \emph{Annual Review of Clinical Psychology}, 8:\penalty0 21--48,
  2012.

\bibitem[Liu and Brown(1993)]{brownLiu}
Richard~C. Liu and Lawrence~D. Brown.
\newblock Nonexistence of informative unbiased estimators in singular problems.
\newblock \emph{Annals of Statistics}, 21\penalty0 (1):\penalty0 1--13, 1993.

\bibitem[Marchand and Strawderman(2004)]{marchand2004}
Eric Marchand and William~E Strawderman.
\newblock Estimation in restricted parameter spaces: A review.
\newblock \emph{Lecture Notes-Monograph Series}, pages 21--44, 2004.

\bibitem[Moodie et~al.(2007)Moodie, Richardson, and Stephens]{moodie}
E.E.M. Moodie, T.S. Richardson, and D.A. Stephens.
\newblock Demystifying optimal dynamic treatment regimes.
\newblock \emph{Biometrics}, 63\penalty0 (2):\penalty0 447--455, 2007.

\bibitem[Moodie et~al.(2010)Moodie, Richardson, and Stephens]{moodieT}
E.E.M. Moodie, T.S. Richardson, and D.A. Stephens.
\newblock Estimating optimal dynamic regimes: Correcting bias under the null.
\newblock \emph{Biometrics}, 63\penalty0 (2):\penalty0 447--455, 2010.

\bibitem[Murphy(2005{\natexlab{a}})]{murphy2005experimental}
S.A. Murphy.
\newblock An experimental design for the development of adaptive treatment
  strategies.
\newblock \emph{Statistics in medicine}, 24\penalty0 (10):\penalty0 1455--1481,
  2005{\natexlab{a}}.

\bibitem[Murphy(2003)]{murphyZThree}
Susan~A. Murphy.
\newblock Optimal dynamic treatment regimes.
\newblock \emph{Journal of the Royal Statistical Society, Series B},
  65\penalty0 (2):\penalty0 331--366, 2003.

\bibitem[Murphy(2005{\natexlab{b}})]{murphyZfive}
Susan~A. Murphy.
\newblock A generalization error for {Q}-learning.
\newblock \emph{Journal of Machine Learning Research}, 6:\penalty0 1073--1097,
  Jul 2005{\natexlab{b}}.

\bibitem[Nahum-Shani et~al.(2012{\natexlab{a}})Nahum-Shani, Qian, Almirall,
  Pelham, Gnagy, Fabiano, Waxmonsky, Yu, and Murphy]{inbalOne}
Inbal Nahum-Shani, Min Qian, Daniel Almirall, William~E Pelham, Beth Gnagy,
  Gregory~A Fabiano, James~G Waxmonsky, Jihnhee Yu, and Susan~A Murphy.
\newblock Experimental design and primary data analysis methods for comparing
  adaptive interventions.
\newblock \emph{Psychological methods}, 17\penalty0 (4):\penalty0 457,
  2012{\natexlab{a}}.

\bibitem[Nahum-Shani et~al.(2012{\natexlab{b}})Nahum-Shani, Qian, Almirall,
  Pelham, Gnagy, Fabiano, Waxmonsky, Yu, and Murphy]{inbalTwo}
Inbal Nahum-Shani, Min Qian, Daniel Almirall, William~E Pelham, Beth Gnagy,
  Gregory~A Fabiano, James~G Waxmonsky, Jihnhee Yu, and Susan~A Murphy.
\newblock Q-learning: A data analysis method for constructing adaptive
  interventions.
\newblock \emph{Psychological methods}, 17\penalty0 (4):\penalty0 478,
  2012{\natexlab{b}}.

\bibitem[Olshen(1973)]{olshen1973conditional}
R.A. Olshen.
\newblock {The conditional level of the F-test}.
\newblock \emph{Journal of the American Statistical Association}, 68\penalty0
  (343):\penalty0 692--698, 1973.

\bibitem[Orellana et~al.(2010)Orellana, Rotnitzky, and Robins]{Orellana10}
L.~Orellana, A.~Rotnitzky, and J.~Robins.
\newblock Dynamic regime marginal structural mean models for estimation of
  optimal dynamic treatment regimes, part i: Main content.
\newblock \emph{Int. Jrn. of Biostatistics}, 6\penalty0 (2), 2010.

\bibitem[PSU Methodology~Center(2012)]{methCenterURL}
The PSU Methodology~Center.
\newblock Smart studies, January 2012.
\newblock URL \url{http://methodology.psu.edu/ra/adap-inter/projects}.

\bibitem[Putterman(1994)]{putterman1994markov}
Martin~L Putterman.
\newblock Markov decision processes.
\newblock \emph{John Wiely and Sons, New York}, 1994.

\bibitem[Qian et~al.(2013)Qian, Nahum-Shani, and Murphy]{minChapter}
Min Qian, Inbal Nahum-Shani, and Susan~A Murphy.
\newblock Dynamic treatment regimes.
\newblock In \emph{Modern Clinical Trial Analysis}, pages 127--148. Springer,
  2013.

\bibitem[Robins(1986)]{robins1986}
J.~Robins.
\newblock A new approach to causal inference in mortality studies with a
  sustained exposure period—application to control of the healthy worker
  survivor effect.
\newblock \emph{Mathematical Modelling}, 7\penalty0 (9):\penalty0 1393--1512,
  1986.

\bibitem[Robins(2004)]{robins2004optimal}
J.M. Robins.
\newblock {Optimal structural nested models for optimal sequential decisions}.
\newblock In \emph{Proceedings of the Second Seattle Symposium in
  Biostatistics: Analysis of Correlated Data}, 2004.

\bibitem[Robins et~al.(2008)Robins, Orellana, and Rotnitzky]{robinsetal2008}
J.M. Robins, L.~Orellana, and A.~Rotnitzky.
\newblock Estimation and extrapolation of optimal treatment and testing
  strategies.
\newblock \emph{Statistics in Medicine}, pages 4678--4721, 2008.

\bibitem[Rubin(1978)]{rubin}
D.B. Rubin.
\newblock Bayesian inference for causal effects: The role of randomization.
\newblock \emph{The Annals of Statistics}, pages 34--58, 1978.

\bibitem[Schulte et~al.(2013)Schulte, Tsiatis, Laber, , and Davidian]{schulte}
P.J. Schulte, A.A. Tsiatis, E.B. Laber, , and M.~Davidian.
\newblock Q- and a-learning methods for estimating optimal dynamic treatment
  regimes.
\newblock Technical Report arXiv:1202.4177v2, arXiv.org, 2013.

\bibitem[Si et~al.(2004)Si, Barto, Powell, Wunsch, et~al.]{si2004handbook}
Jennie Si, Andrew~G Barto, Warren~B Powell, Donald~C Wunsch, et~al.
\newblock \emph{Handbook of learning and approximate dynamic programming}.
\newblock IEEE Press Los Alamitos, 2004.

\bibitem[Song et~al.(2011)Song, Wang, Zeng, and Kosorok]{song}
R~Song, W.. Wang, D.~Zeng, and M.~Kosorok.
\newblock Penalized q-learning for dynamic treatment regimes.
\newblock Technical Report arXiv:1108.5338v1, arxiv.org, 2011.

\bibitem[Sutton et~al.(1999)Sutton, McAllester, Singh, and Mansour]{sutton1999}
Richard~S Sutton, David~A McAllester, Satinder~P Singh, and Yishay Mansour.
\newblock Policy gradient methods for reinforcement learning with function
  approximation.
\newblock In \emph{NIPS}, volume~99, pages 1057--1063, 1999.

\bibitem[Sutton and Barto(1998)]{sutton}
R.S. Sutton and A.G. Barto.
\newblock \emph{Reinforcment Learning: An Introduction}.
\newblock The MIT Press, 1998.

\bibitem[Szepesv{\'a}ri(2010)]{csaba2010}
Csaba Szepesv{\'a}ri.
\newblock Algorithms for reinforcement learning.
\newblock \emph{Synthesis Lectures on Artificial Intelligence and Machine
  Learning}, 4\penalty0 (1):\penalty0 1--103, 2010.

\bibitem[Van~der Vaart(1991)]{van1991differentiable}
A.~Van~der Vaart.
\newblock On differentiable functionals.
\newblock \emph{The Annals of Statistics}, pages 178--204, 1991.

\bibitem[Van~der Vaart and Wellner(1996)]{van1996weak}
Aad Van~der Vaart and Jon Wellner.
\newblock \emph{Weak convergence and empirical processes: with applications to
  statistics}.
\newblock Springer, 1996.

\bibitem[Watkins and Dayan(1992)]{watkins1992q}
C.J.C.H. Watkins and P.~Dayan.
\newblock {Q-learning}.
\newblock \emph{Machine learning}, 8\penalty0 (3):\penalty0 279--292, 1992.

\bibitem[Wiering and van Otterlo(2012)]{wiering2012}
Marco Wiering and Martijn van Otterlo.
\newblock \emph{Reinforcement Learning: State-of-the-art}, volume~12.
\newblock Springer, 2012.

\bibitem[Zhang et~al.(2012)Zhang, Tsiatis, Laber, and Davidian]{baqun}
B.~Zhang, A.A Tsiatis, E.B. Laber, and M~Davidian.
\newblock A robust method for estimating optimal treatment regimes.
\newblock \emph{Biometrics}, To appear, 2012.

\bibitem[Zhang et~al.(2013)Zhang, Tsiatis, Laber, and Davidian]{baqun2}
B.~Zhang, A.A Tsiatis, E.B. Laber, and M~Davidian.
\newblock Robust estimation of optimal dynamic treatment regimes for sequential
  treatment decisions.
\newblock \emph{Biometrika}, To appear, 2013.

\bibitem[Zhao et~al.(2012)Zhao, Zeng, Rush, and Kosorok]{yingqi}
Yingqi Zhao, Donglin Zeng, A~John Rush, and Michael~R Kosorok.
\newblock Estimating individualized treatment rules using outcome weighted
  learning.
\newblock \emph{Journal of the American Statistical Association}, 107\penalty0
  (499):\penalty0 1106--1118, 2012.

\bibitem[Zhao et~al.(2013)Zhao, Zeng, Laber, and Kosorok]{yingqi2}
Yingqi Zhao, Donglin Zeng, Eric~B Laber, and Michael~R Kosorok.
\newblock New statistical learning methods for estimating optimal dynamic
  treatment regimes.
\newblock \emph{Under review}, 107\penalty0 (499):\penalty0 1106--1118, 2013.

\end{thebibliography}
\nocite{hirano}
\nocite{robins2004optimal}
\nocite{watkins1992q}
\nocite{van1991differentiable}
\nocite{schulte}
\nocite{murphyZThree}
\nocite{robins1986}
\nocite{rubin}

\appendix

\section{Appendix: Outcome Weighted Learning}\label{ap:outcomewt}

Recall that the value of a DTR $\pi$, $\mathbb{E}^{\pi}Y$, is the
expected outcome of $Y$ under the restriction that $A_t = \pi_t(H_t)$.
For expositional simplicity, assume $P(A_t=1|H_t) = 1/2$ and that $Y$
is coded so that $Y\ge 0$ in this section.  Then a change of measure
implies that the value $\mathbb{E}^{\pi}Y = 4P \left(Y1_{A_1 =
    \pi_1(H_1)}1_{A_2 = \pi_2(H_2)}\right)$; the empirical analog is
$4 \pn\left( Y1_{A_1 = \pi_1(H_1)}1_{A_2 = \pi_2(H_2)}\right)$.
Note the resemblance to the classification rate. As in classification,
directly maximizing the empirical value over a class of DTRs is a
discrete optimization problem and is usually computationally
burdensome.  \cite{yingqi2} solve a concave relaxation of this problem
by replacing the nonsmooth indicator functions with concave
surrogates.  Consider decision rules of the form $\pi_t(h_t) =
1_{h_{t,1}^{\T}\psi_{t,1} \ge 0}$ where $h_{t,1}$ is a known feature
of $h_t$.  Note that $1_{ A_t=\pi_t(H_t) } =
1_{(2A_t-1)H_{t,1}^{\T}\psi_{t,1} \ge 0}$.  Let $\phi :\mathbb{R}
\rightarrow \mathbb{R}$ be a concave function that satisfies $\phi(z)
\le k+ 1_{z \ge 0}$ for all $z$ where $k$ is a constant.  A version of
the algorithm is as follows.
\begin{enumerate}
  \item Stage 2 optimization: $\widehat{\psi}_{2,1} =
    \arg\max_{\psi_{2,1}} \pn Y\phi
    \left((2A_2-1)H_{2,1}^{\T}\psi_{2,1}\right)$.
  \item Stage 1 optimization: $\widehat{\psi}_{1,1} =
    \arg\max_{\psi_{1,1}} \pn
    Y1_{(2A_2-1)H_{2,1}^{\T}\widehat{\psi}_{2,1} \ge 0} \phi\left(
      (2A_1 -1)H_{1,1}^{\T}\psi_{1,1}
      \right)$.
\end{enumerate}
The estimator of the optimal DTR is thus
$\widehat{\pi}_{t}(h_t)
=1_{h_{t,1}^{\T}\widehat{\psi}_{t,1}\ge 0}$.   For
illustration we use $\phi(z) =1 -(1-z)^2$.  Define the population
parameters:
\begin{eqnarray*}
  \psi_{2,1}^* &\triangleq & \arg\min_{\psi_{2,1}}
  P\left[Y(1-(2A_2-1)H_{2,1}^{\T}\psi_{2,1})^2\right],  \\
  \psi_{1,1}^* &\triangleq & \arg\min_{\psi_{1,1}}
  P\left[Y1_{(2A_2-1)H_{2,1}^{\T}\psi_{2,1}^* \ge
    0}(1-(2A_1-1)H_{1,1}^{\T}\psi_{1,1})^2\right].
\end{eqnarray*}
In addition, define $\Psi_1 \triangleq
PYH_{1,1}H_{1,1}^{\T}1_{(2A_2-1)H_{2,1}^{\T}\psi_{2,1}^* \ge 0}$ and
the corresponding plugin estimator $\widehat{\Psi}_{1} \triangleq \pn
YH_{1,1}H_{1,1}^{\T}1_{(2A_2-1)H_{2,1}^{\T}\widehat{\psi}_{2,1} \ge
  0}$, which we assume is invertible.  Then $\rtn
(\widehat{\psi}_{1,1} - \psi_{1,1}^*) = \rtn\widehat{\Psi}_{1}^{-1}\pn
Y H_{1,1}(2A_1-1)1_{(2A_1-1)H_{2,1}^{\T}\widehat{\psi}_{2,1} \ge
  0}(1-(2A_1-1)H_{1,1}^{\T}\psi_{1,1}^*)$ which can be decomposed as
\begin{equation*}
  \mathbb{T}_{n} + \rtn\widehat{\Psi}_{1}^{-1} \pn
  YH_{1,1}(2A_1-1)(1-(2A_1-1)H_{1,1}^{\T}\psi_{1,1}^*) \mathbb{L}_{n},
\end{equation*}
where
\begin{eqnarray*}
  \mathbb{T}_{n} &=& \widehat{\Psi}_{1}^{-1}\rtn
  (\pn-P)\left[YH_{1,1}(2A_1-1)(1-(2A_1-1)H_{1,1}^{\T}\psi_{1,1}^*)
  1_{(2A_2-1)H_{2,1}^{\T}\psi_{2,1}^* \ge 0}\right], \\
\mathbb{L}_{n} &=& 1_{(2A_2-1)H_{2,1}^{\T}\widehat{\psi}_{2,1} \ge 0}
- 1_{(2A_2-1)H_{2,1}^{\T}\psi_{2,1}^* \ge 0}.
\end{eqnarray*}
The term $\mathbb{T}_{n}$ is smooth and asymptotically normal under
mild conditions whereas $\mathbb{L}_{n}$ is nonsmooth.  If $h_{2,1}$
satisfies $h_{2,1}^{\T}\psi_{2,1}^* = 0$ then
$\mathbb{L}_{n}\big|_{H_{2,1} = h_{2,1}}$ converges in distribution to
a Bernoulli random variable with probability of success equal to
$1/2$.  On the other hand, if $h_{2,1}^{\T}\psi_{2,1}^* \ne 0$ then
$\mathbb{L}_{n}\big|_{H_{2,1} = h_{2,1}}$ converges in probability to
zero.  Thus, in parallel with the $Q$-learning case, the limiting
distribution of $\rtn(\widehat{\psi}_{1,1} - \psi_{1,1}^*)$ depends
abruptly on both the value of $\psi_{2,1}^*$ and the distribution of
$H_{2,1}$.  Therefore the same theoretical challenges as in $Q$-learning occur in
outcome-weighted learning.

\section{Appendix: Proofs}\label{ap:proofs}
\subsection{Proof of theorems in Section 3}

\begin{lem}
If $\omega \sim \mathrm{Normal}(0, \nu^2)$ then
$\mathbb{E}\left[\omega\right]_+ = \nu/\sqrt{2\pi}$.
\end{lem}
\begin{proof}
Let $\phi$ denote the density of a standard normal random variable.  Then
\begin{equation*}
\mathbb{E}\left[\omega\right]_+  =
\int_{\mathbb{R}}\left[\omega\right]_+\phi(\omega/\nu)/\nu d\omega =
\int_{\mathbb{R}_+}\omega\phi(\omega/\nu)/\nu d\omega =
\nu/\sqrt{2\pi}.
\end{equation*}
\end{proof}
\begin{proof}[Proof of Theorem 3.1]
Using Theorem 4.2, part I, it follows that
$\mathrm{Bias}(\bHat_1, c)$ is equal to
\begin{equation*}
\mathbb{E}\left(
c^{\T}\Sigma_{1,\infty}^{-1}PB_1\left[H_{2,1}^{\T}\mathbb{V}_{\infty}\right]_+1_{H_{2,1}^{\T}
\beta_{2,1}^* = 0}
\right).
\end{equation*}
Exchanging expectations and applying Lemma 7.1 gives the result.
\end{proof}

\begin{lem}
  If $z \sim \mathrm{Normal}(0, 1)$ and $\sigma > 0$ then
\begin{equation*}
  \mathbb{E}\left[z\right]_+\left(1-\sigma/z^2\right)_+ = \left\lbrace
\exp\lbrace -\sigma/2\rbrace - \sigma\int_{\sqrt{\sigma}}^{\infty}\exp{-z^2/2}/zdz
\right\rbrace/\sqrt{2\pi}.
\end{equation*}
\end{lem}
\begin{proof}
Let $\phi$ denote the density of a standard normal random variable,
then
\begin{equation*}
  \mathbb{E}\left[z\right]_+\left(1-\frac{\sigma}{z^2}\right)_+ =
\int_{\sqrt{\sigma}}^{\infty}z\left(1-\frac{\sigma}{z^2}\right)\phi(z)dz
 = \left\lbrace
\exp\lbrace -\sigma/2\rbrace - \sigma\int_{\sqrt{\sigma}}^{\infty}\frac{1}{z}\exp(-z^2/2)dz
\right\rbrace/\sqrt{2\pi}.
\end{equation*}
\end{proof}
\begin{proof}[Proof of Theorem 3.2]
Notice that $\rtn(\bHat_1^{\sigma} - \beta_1^*) =
\widehat{\Sigma}_1^{-1}\rtn\pn B_1(\tilde{Y}^{\sigma} -
B_1^{\T}\beta_1^*)$ which can be decomposed as
\begin{equation*}
\widehat{\Sigma}_{1}^{-1}\rtn(\pn-P)B_1(\tilde{Y}^* -
B_1^{\T}\beta_1^*)
+ \widehat{\Sigma}_{1}^{-1}\rtn\pn B_1(\tilde{Y}^{\sigma} - \tilde{Y}^*),
\end{equation*}
where we have used $PB_1(\tilde{Y}^*-B_1^{\T}\beta_1^*) = 0$.  The
first term in the above display is asymptotically normal with mean
zero and thus does not contribute to the asymptotic bias.  The second
term in the above display is equal to
\begin{multline*}
\widehat{\Sigma}_{1}^{-1}\pn B_1H_{2,0}^{\T}\rtn(\bHat_{2,0} -
\beta_{2,0}^*)
\\ + \widehat{\Sigma}_{1}^{-1}\rtn\pn B_1\left(
\left[H_{2,1}^{\T}\bHat_{2,1}\right]_+\left(
1-\frac{\sigma H_{2,1}^{\T}\widehat{\Sigma}_{21,21}H_{2,1}}{n(\bHat_{2,1}^{\T}H_{2,1})^2}
\right)_+ - \left[H_{2,1}^{\T}\beta_{2,1}^*\right]_+
\right)1_{H_{2,1}^{\T}\beta_{2,1}^*\ne 0} \\
+
\widehat{\Sigma}_{1}^{-1}\pn B_1\left[H_{2,1}^{\T}\rtn(\bHat_{2,1}-\beta_{2,1}^*)\right]_+
\left(
1-\frac{\sigma
  H_{2,1}^{\T}\widehat{\Sigma}_{21,21}H_{2,1}}
{(H_{2,1}^{\T}\rtn(\bHat_{2,1} - \beta_{2,1}^*))^2}
\right)_+1_{H_{2,1}^{\T}\beta_{2,1}^* = 0}.
\end{multline*}
The first two terms can be shown to have asymptotic mean zero and thus
they do not contribute the asymptotic bias.  The last term converges
in distribution to
\begin{equation*}
  \Sigma_{1,\infty}^{-1} P
  \Big[B_1\left[\mathbb{Z}\right]_+\sqrt{H_{2,1}^{\T}\Sigma_{21,21} H_{2,1}}
\left(
1-\frac{\sigma}
{\mathbb{Z}^2}
\right)_+1_{H_{2,1}^{\T}\beta_{2,1}^* = 0}\Big],
\end{equation*}
where $\mathbb{Z}$ is a standard normal random variable.  Exchanging
expectations
and applying Lemma 7.2 gives the result.
\end{proof}
\begin{proof}[Proof of Theorem 3.3]
From Theorem 4.2 part 2 it follows that $\mathrm{Bias}(\bHat_1, c, s)$
is equal to
\begin{equation*}
  \mathbb{E}\left(
    c^{\T}\Sigma_{1,\infty}^{-1}PB_1\left(
      \left[H_{2,1}^{\T}(\mathbb{V}_{\infty}+s)\right]_+ - \left[H_{2,1}^{\T}s\right]_+
      \right)1_{H_{2,1}^{\T}\beta_{2,1}^* = 0}
    \right),
\end{equation*}
taking absolute values and applying the Cauchy-Schwarz and triangle
inequalities gives the first result of the theorem.

It can be shown that $c^{\T}\rtn(\bHat_{1}^{\sigma} - \beta_1^*)$
converges in distribution to
\begin{equation*}
  c^{\T}\Sigma_{1,\infty}^{-1}PB_1\left( \left[H_{2,1}^{\T}(\mathbb{V}_{\infty}+s)\right]_+\left(1
    - \frac{\sigma H_{2,1}^{\T}\Sigma_{21,21}
      H_{2,1}}{(H_{2,1}^{\T}(\mathbb{V}_{\infty}+s))^2}\right)_+
   - \left[H_{2,1}^{\T}s\right]_+\right)1_{H_{2,1}^{\T}\beta_{2,1}^*=0}.
\end{equation*}
Recall that $H_{2,1}$ is assumed to have an intercept. Let $e_1$
denote the first column of an
$\dim(\beta_{2,1}^*)\times\dim(\beta_{2,1}^*)$ identity matrix, and
choose
$s = -\mathbb{V}_{\infty} + e_1\log\,\sigma$ then as
$\sigma\rightarrow \infty$ the above term behaves as
\begin{equation*}
  c^{\T}\Sigma_{1,\infty}^{-1}PB_11_{H_{2,1}^{\T}\beta_{2,1}^*=0}\log(\sigma),
\end{equation*}
which tends to $\infty$ in magnitude.  Thus, the supremum over $s$,
of $|\mathrm{Bias}(\bHat_1^{\sigma}, c, s)|$ must be at least as
large.
\end{proof}

\subsection{Proof of theorems in Section 4}
In the main body we assumed a single terminal reward $Y$,
here, to cover a more general case we assume that an intermediate
reward, $Y_1$ may be observed at the end of the first stage as
well as a terminal reward $Y_2$.
Thus, one seeks to maximize $\mathbb{E}^{\pi}(Y_1+Y_2)$
where $\mathbb{E}^{\pi}$ denotes expectation with respect
to the joint distribution of the trajectory under the
restriction that $A_t = \pi_t(H_t), t=1,2.$
Throughout this section, let $K$ denote a sufficiently large positive constant that may vary from line to line. Let $D_{p}$ denote
the space of $p\times p$ symmetric positive definite matrices equipped
with the spectral norm, and for any $k\in (0,1)$, let $D_{p}^{k}$ denote the
subset of $D_{p}$ with members having eigenvalues in the range $[k, 1/k]$.
For any class of real-valued functions $\mathcal{F}$, let $\rho_P(f)\triangleq (P(f-Pf)^2)^{1/2}$ denote the centered $L_2$-norm on $\mathcal{F}$,
$l^{\infty}(\mathcal{F})$ denote the space of uniformly bounded
real-valued functions on $\mathcal{F}$ equipped with the $\sup$ norm, and
$C_{b}(\mathcal{F})$ denote the subspace of $l^{\infty}(\mathcal{F})$
of continuous and bounded functions from $\mathcal{F}$ into $\mathbb{R}$, respectively. Furthermore, let $\mathbb{G}_n \triangleq \sqrt n(\pn-P)$, $\mathbb{G}^{(b)}_n \triangleq \sqrt n(\mathbb{\hat P}_n^{(b)}-\pn)$, and
$P_{M}$ denote probability taken with respect to the bootstrap
weights defining the bootstrap empirical measure, respectively.

\subsubsection{Results for second stage parameters}
\label{sec:2stage2}
In this section we will characterize the limiting distributions of
the second stage parameters under fixed and local alternatives.  We will
also derive the limiting distribution of the bootstrap analog
of the second stage parameters. For convenience, let $p_{t0}\triangleq {\rm dim}(\beta^*_{t,0})$,
$p_{t1}\triangleq {\rm dim}(\beta^*_{t,1})$, and $p_{t}\triangleq {\rm dim}(\beta^*_{t}) = p_{t0} + p_{t1}$ for $t=1, 2$.

\begin{thm}
\label{stage2}
Assume (A1) and (A2) and fix $a\in\mathbb{R}^{p_2}$, then
\begin{enumerate}
\item $a^{\T}\rtn(\bHat_2 - \beta_{2}^*) \leadsto_{P}
a^{\T}\mathbb{Z}_\infty$,
\item $a^{\T}\rtn(\bBoot_2 -\bHat_{2}) \leadsto_{P_{M}}
a^{\T}\mathbb{Z}_\infty$
in $P$-probability; and
\item if in addition (A3) holds, $a^{\T}\rtn(\bHat_2
  -\beta_{2,n}^*) \leadsto_{P_n}
  a^{\T}\mathbb{Z}_\infty$,
\end{enumerate}
where $\mathbb{Z}_\infty$ is
a mean zero normal random vector with covariance matrix $\Sigma_{2,\infty}^{-1}P[B_2B_2^{\T}(Y_2-B_2^{\T}\beta_2^*)^2]\Sigma_{2,\infty}^{-1}$.
\end{thm}
\begin{proof}
Define the class of functions $\mathcal{F}_2$ as
\begin{equation}
\label{F2}
\mathcal{F}_2 \triangleq \lbrace
f(b_2,y_2; a,\beta_2)\triangleq a^{\T}b_2(y_2 - b_2^{\T}\beta_2)\,:\, a, \beta_2\in
\mathbb{R}^{p_2},\, ||a||\le K,\,||\beta_2|| \le K\rbrace,
\end{equation}
and the function $w_{2}:D_{p_2}\times l^{\infty}(\mathcal{F}_2)
\times \mathbb{R}^{p_2}\times \mathbb{R}^{p_2}\rightarrow \mathbb{R}$ as
\begin{equation}
w_2(\Sigma, \mu, \beta_2,a) \triangleq
\mu\left(a^{\T}\Sigma^{-1}B_{2}(Y_2 - B_{2}^{\T}\beta_{2})\right). \label{w2}
\end{equation}
Since the estimated covariance matrices $\hat\Sigma_2=\pn B_2B_2^{\T}$ and $\hat{\Sigma}_{2}^{(b)} = \pHat B_2B_2^{\T}$ are weakly consistent (by Lemma \ref{sigma2_conv}), we will avoid additional notation by assuming
they are nonsingular for all $n$ without loss of generality.
Thus
\begin{align*}
& a^{\T}\rtn(\bHat_2 - \beta_{2}^*) = w_{2}(\hat{\Sigma}_{2}, \mathbb{G}_{n}, \beta_{2}^*, a), \phantom{bb}
 a^{\T}\rtn(\bBoot_2 - \bHat_2) =w_{2}(\hat{\Sigma}_{2}^{(b)}, \mathbb{G}_{n}^{(b)}, \bHat_2, a),\\
 \mbox{and } & a^{\T}\rtn(\bHat_{2} - \beta_{2,n}^*) =
    w_{2}(\hat{\Sigma}_{2}, \rtn(\pn - P_n), \beta_{2,n}^*, a).
 \end{align*}
In addition, note that $a^{\T}\mathbb{Z}_\infty= w_{2}(\Sigma_{2,\infty}, \mathbb{G}_{\infty}, \beta_{2}^*, a)$ in distribution,
where $\mathbb{G}_{\infty}$ is a tight Gaussian process in $l^{\infty}(\mathcal{F}_2)$ with covariance function
$\Cov(\mathbb{G}_{\infty}f_1, \mathbb{G}_{\infty}f_2)=P(f_1-Pf_1)(f_2-Pf_2)$.
Results 1 and 3 follow from Lemmas \ref{w2CtsLem} - \ref{process_conv} and the continuous mapping theorem [Theorem 1.3.6 of \citealt{van1996weak}].
Result 2 follows from the bootstrap continuous mapping theorem [Theorem 10.8 of \citealt{kosorok}] together with Lemmas \ref{w2CtsLem} - \ref{closedset}.
\end{proof}

\begin{lem}\label{w2CtsLem}
Under (A1), the function $w_2$ defined in (\ref{w2}) is continuous
at points in $D_{p_2}\times C_{b}(\mathcal{F})\times \mathbb{R}^{p_2}
\times \mathbb{R}^{p_2}$.
\end{lem}
\begin{proof}
Let $\epsilon > 0$ be arbitrary and let $(\Sigma, \mu, \beta_2, a)$
be an element of $D_{p_2}\times C_{b}(\mathcal{F})\times \mathbb{R}^{p_2}
\times \mathbb{R}^{p_2}$.  In addition, let
$(\Sigma', \mu', \beta_2', a')$ be an element of
$D_{p_2} \times l^{\infty}(\mathcal{F})\times \mathbb{R}^{p_2}\times
\mathbb{R}^{p_2}$.  From the form of $\mathcal{F}$ and the
moment assumptions in (A1) we see that if
$\Sigma - \Sigma'$, $a-a'$, and $\beta_2-\beta_2'$ are small then
so must $\rho_{P}(f-f')$ be small, where
\begin{eqnarray*}
f(B_2, Y_2) &=& a^{\T}\Sigma^{-1} B_{2}(Y_2 - B_2^{\T}\beta_2), \\
f'(B_2, Y_2) &=& a^{'\T}\Sigma'^{-1} B_{2}(Y_2 - B_{2}^{\T}\beta_{2}').
\end{eqnarray*}
In particular, we can choose $\delta > 0$ sufficiently small so that
$||\Sigma-\Sigma'|| + ||a-a'|| + ||\beta_2 - \beta_2'|| < \delta$
implies that $\rho_{P}(f-f')$ is small enough to guarantee, by appeal
to the continuity of $\mu$, that $|\mu(f) - \mu(f')| \le \epsilon/2$.
Finally, note that
\begin{equation*}
\big|w_{2}(\Sigma, \mu, \beta_{2}, a) - w_{2}(\Sigma', \mu', \beta_{2}',
a')\big| \le |\mu(f) - \mu(f')| + ||\mu - \mu'||_{\mathcal{F}_2}.
\end{equation*}
Let $\delta' = \min(\delta, \epsilon/2)$, then
$||\Sigma - \Sigma'|| + ||\mu-\mu'||_{\mathcal{F}_2} + ||\beta_2-\beta_2'|| +
||a-a'|| < \delta'$ implies that
$|w_2(\Sigma, \mu, \beta_2, a) - w_{2}(\Sigma', \mu', \beta_2', a')| \le
\epsilon$.  Thus, the desired result is proved.
\end{proof}

Having established the continuity of $w_2$ the next step will
be to characterize the limiting behavior of
$\beta_{2,n}^*$, $\hat\beta_2$, $\hat\Sigma_{2}$,
$\hat{\Sigma}_{2}^{(b)}$, and the limiting
distributions of $\mathbb{G}_n$, $\rtn(\pn-P_n)$, and
$\rtn(\pHat -\pn)$.  These limits are established in a series
of lemmas.
Once this has been accomplished we will be able
to apply the continuous mapping theorem to obtain the limiting
distributions of $\rtn(\bHat_2 - \beta_2^*)$,
$\rtn(\bHat - \beta_{2,n}^*)$, and $\rtn(\bHat_{2}^{(b)} -
\bHat_{2})$.
\begin{lem}
\label{sigma2_conv}
Assume (A1)-(A2), then $\hat{\Sigma}_{2}\to_P\Sigma_{2,\infty}$ and $\hat{\Sigma}_{2}^{(b)}\to_{P_M}\Sigma_{2,\infty}$ in $P$-probability as $n\to\infty$.
Furthermore, if (A3) holds, then 
$\hat{\Sigma}_{2}\to_{P_n}\Sigma_{2,\infty}$ as $n\to\infty$.
\end{lem}
\begin{proof}
The first two claims follow from weak law of large numbers [\citealt{bickel1981some, csorgo}]. For the third claim, note that
$\hat{\Sigma}_{2} - \Sigma_{2,\infty} = (\hat{\Sigma}_{2} - \Sigma_{2,n}) + (\Sigma_{2,n}  - \Sigma_{2,\infty})$ and $\hat{\Sigma}_{2} - \Sigma_{2,n}\to_{P_n} 0$ by law of large numbers.
Below we show that $\Sigma_{2,n}\to\Sigma_{2,\infty}$. This will complete the proof.

let $c\in\mathbb{R}^{p_2}$ be arbitrary and define
$\nu \triangleq c^{\T}B_2B_2^{\T}c$.  We will show that
$\int \nu(dP_{n} - dP) = o(1)$.  First, note that
\begin{equation*}
\int \nu(dP_{n}-dP) = \int \nu (dP_{n}^{1/2}+dP^{1/2})(dP_{n}^{1/2}-dP^{1/2}).
\end{equation*}
Furthermore, the absolute value of the foregoing expression is bounded
above by
\begin{equation*}
\int |\nu||(dP_{n}^{1/2} + dP^{1/2})|(dP_{n}^{1/2} - dP^{1/2}) \le
\sqrt{\int \nu^2(dP_{n}^{1/2}+dP^{1/2})^2}\sqrt{\int (dP_{n}^{1/2}
- dP^{1/2})^2},
\end{equation*}
where the last inequality is simply H\"older's inequality.  Next,
note that owing to the inequality $(\sqrt{a} + \sqrt{b})^2 \le
2a + 2b$ it follows that
\begin{equation*}
\int \nu^2(dP_{n}^{1/2}+dP^{1/2})^2 \le 2\int \nu^2dP_{n} + 2\int \nu^2 dP
= O(1),
\end{equation*}
by appeal to (A3).  Now write
\begin{multline*}
\int (dP_{n}^{1/2}-dP^{1/2})^2 = n^{-1}\bigg\lbrace
\int \left(\rtn(dP_{n}^{1/2}-dP^{1/2}) - \frac{1}{2}vdP^{1/2}\right)^2
\\ - \frac{1}{4}\int v^2dP + \rtn\int vdP^{1/2}(dP_{n}^{1/2}-dP^{1/2})
\bigg\rbrace.
\end{multline*}
The right hand side of the preceding display is equal to
\begin{equation*}
O(1/n) + n^{-1/2}\int vdP^{1/2}(dP_{n}^{1/2} - dP^{1/2}) \le
O(1/n) + n^{-1/2}\sqrt{\int v^2dP}\sqrt{\int(dP_{n}^{1/2}-dP^{1/2})^2},
\end{equation*}
which is $o(1)$. Thus $\Sigma_{2,n}\to\Sigma_{2,\infty}$.
\end{proof}

\begin{lem}
\label{beta2_conv}
Under (A1) and (A2), $\hat\beta_2 \to_P \beta_2^*$ as $n\to\infty$. If, in addition (A3) holds, then
$\lim_{n\rightarrow \infty}\rtn(\beta_{2,n}^* -
\beta_{2}^*) = \Sigma_{2}^{-1}PvB_2(Y_2-B_2^{\T}\beta_2^*)$.
\end{lem}
\begin{proof}
$\hat\beta_2 \to_P \beta_2^*$ follows from weak law of large numbers and Slutsky's lemma.

Recall that $0 = P_nB_2(Y_2 - B_2^{\T}\beta_{2,n}^*)$ which we can
write as
\begin{equation*}
\rtn(P_n - P)B_2(Y_2 -B_2^{\T}\beta_2^*) - \Sigma_{2,n}\rtn(\beta_{2}^*
- \beta_{2,n}^*),
\end{equation*}
so that for sufficiently large $n$ it follows that
$\rtn(\beta_{2,n}^* - \beta_{2}^*) = \Sigma_{2,n}^{-1}\rtn(P_n - P)
B_2(Y_2 - B_2^{\T}\beta_2^*)$.  By appeal to (A3) it follows that
for any vector $a \in \mathbb{R}^{p_2}$ we have
$\sup_{n}P_n(a^{\T}B_2(Y_2 -B_2^{\T}\beta_2^*))^2 < \infty$.
Theorem 3.10.12 of \cite{van1996weak} ensures
that
\begin{equation*}
\rtn(P_n-P)B_{2}(Y_2-B_2^{\T}\beta_2^*) \rightarrow PvB_2(Y_2-B_2^{\T}
\beta_2^*)
\end{equation*}
as $n\to\infty$. This completes the proof.
\end{proof}


\begin{lem}
\label{process_conv}
Assume (A1)-(A2), then

1) $\mathbb{G}_{n} \leadsto_{P} \mathbb{G}_{\infty}$ in $l^{\infty}(\mathcal{F}_2)$,
where $\mathcal{F}_2$ is defined in (\ref{F2}), and $\mathbb{G}_{\infty}$ is a tight Gaussian process in $l^{\infty}(\mathcal{F}_2)$ with covariance function
$\Cov(\mathbb{G}_{\infty}f_1, \mathbb{G}_{\infty}f_2)=P(f_1-Pf_1)(f_2-Pf_2)$; and

2) $\sup_{\omega\in BL_1}\big|\mathbb{E}_{M}\omega(\rtn(\pHat-\pn)) -
\mathbb{E}\omega(\mathbb{G}_{\infty})\big| \rightarrow_{P^*} 0 $ in $l^{\infty}(\mathcal{F}_2)$.\\
If, in addition (A3) holds, then

3) $\rtn(\pn - P_n) \leadsto_{P_{n}} \mathbb{G}_{\infty}$ in
$l^{\infty}(\mathcal{F}_2)$.
\end{lem}
\begin{proof}
First note that $\mathcal{F}_2$ is a subset of the
pairwise product of the linear classes $\lbrace a^{\T}b_2 :
a\in\mathbb{R}^{p_2}\rbrace$ and $\lbrace y_2-b_2^{\T}\beta_2
\,:\,\beta\in\mathbb{R}^{p_2}\rbrace$ each of which is VC-subgraph of index no more than $p_2 + 1$ and $P$-measurable.
Under (A1), the envelope of $\mathcal{F}_2$, $F_2(B_2,Y_2) = K||B_2||(|Y_2| + K||B_2||)$, is square integrable.
This implies that $\mathcal{F}_2$ is P-Donsker, and 1) follows immediately. 2) follows from Theorem 3.6.1 of van der Vaart and Wellner (1996).
For 3), note that from (A3) it follows that $\sup_{f}|P_nf|$ is a bounded sequence.  The
result follows from theorem 3.10.12 of \cite{van1996weak}.
\end{proof}

%
%
%

\begin{lem}
\label{closedset}
The space $C_{b}(\mathcal{F}_2)$ is a closed subset of $l^{\infty}(\mathcal{F}_2)$ and ${\rm P}(\mathbb{G}_\infty\in C_{b}(\mathcal{F}_2))=1$.
\end{lem}
\begin{proof}
  Let $\lbrace \mu_{n} \rbrace_{n=1}^\infty$ be a convergent sequence of
  elements in $C_{b}(\mathcal{F}_2)$ and $\mu_{0}$ the limiting
  element. For the first claim, we only need to show that  $||\mu_0||_{\mathcal{F}_2}=\sup_{f\in\mathcal{F}_2}|\mu_0(f)|$ is bounded, and for any
 $f \in \mathcal{F}$ and $\epsilon>0$, there exists some positive $\delta$ depending on $f$ so that $|\mu_0(f') - \mu_0(f)| < \epsilon$
 for all $f'\in\mathcal{F}_2$ and $\rho_{P}(f', f) < \delta$. The boundedness argument follows by noticing that
 $||\mu_0||_{\mathcal{F}_2}\leq ||\mu_n||_{\mathcal{F}_2} + ||\mu_n -\mu_{0}||_{\mathcal{F}_2}$ for any $n$; in particular,
 for some fixed large enough $n$, $||\mu_n||_{\mathcal{F}_2}$ is bounded by the fact $\mu_n\in C_b(\mathcal{F}_2)$, and $||\mu_n -\mu_{0}||_{\mathcal{F}_2}$ is bounded above by a constant
due to the convergence of $\mu_{n}$ to $\mu_0$. For continuity,
 note that since
  $\mu_{n}$ converges to $\mu_0$, we can choose $n^*$ so that
  $||\mu_{n} - \mu_{0}|| < \epsilon/4$ for all $n \ge n^*$. In addition, by the
  continuity of $\mu_{n^*}$, there exists some $\delta > 0$ so that $|\mu_{n^*}(f') - \mu_{n^*}(f)| < \epsilon$
 for all $\rho_{P}(f', f) < \delta$.  Thus \begin{eqnarray*}
    |\mu_{0}(f') - \mu_{0}(f)| &\le& |\mu_{0}(f) - \mu_{n^*}(f)| +
|\mu_{n^*}(f') - \mu_{0}(f')|  + |\mu_{n^*}(f) - \mu_{n^*}(f')| \\
&\le & 2||\mu_{0} - \mu_{n^*}||_{\mathcal{F}_2} + |\mu_{n^*}(f) - \mu_{n^*}(f')| \\
&\le & 3\epsilon/4.
  \end{eqnarray*}
This implies that $C_{b}(\mathcal{F})$ is closed.

Next note that $\mathbb{G}_\infty$ is a tight Gaussian process in  $l^{\infty}(\mathcal{F}_2)$. By the argument in section 1.5 of van de \cite{van1996weak},
almost all sample paths $f\to \mathbb{G}_\infty(f,\omega)$ are uniformly $\rho_2$-continuous, where $\rho_2(f_1,f_2)= [P(\mathbb{G}_\infty f_1 - \mathbb{G}_\infty f_2)^2]^{1/2}$ is a semimetric on $\mathcal{F}$.
Since $\rho_2(f_1,f_2) = [Var(f_1 -f_2)]^{1/2} \leq \rho_P(f_1,f_2)$, the continuity of the sample paths of $\mathbb{G}_\infty$ follows immediately.

\end{proof}

%

\subsubsection{A characterization of the first stage coefficients
and the  upper bound $\mathcal{U}(c)$}
\label{sec:2stage1}
In this section we present the proofs for Theorems \ref{thm:stage1} and \ref{thm:ACI}.
We first derive an expansion for the first stage
coefficients and two useful expansions of the upper bound
$\mathcal{U}(c)$.  The terms in the forementioned expansions will be
treated individually in subsequent sections. 
We will make use of the following functions.
\begin{enumerate}
\item $w_{11}:D_{p_1}\times D_{p_1\times p_{20}}\times l^{\infty}(\mathcal{F}_{11})
\times l^{\infty}(\mathcal{F}_{11}) \times \mathbb{R}^{p_{2}}\times \mathbb{R}^{p_1+p_{2}} \rightarrow
  \mathbb{R}$ is defined as
\begin{eqnarray}\label{w11}
w_{11}(\Sigma_1, \Sigma_{12},\mu, \omega,\nu, \beta)
&\triangleq &\mu\left[c^{\T}\Sigma_{1}^{-1}
B_1\big(Y_1 + H_{2,0}^{\T}\beta_{2,0}+
\left[H_{2,1}^{\T}\beta_{2,1}\right]_+ - B_1^{\T}\beta_1\big)\right]\nonumber\\
& &+c^{\T}\Sigma_{1}^{-1}\Sigma_{12}\nu_0
+\omega\left(c^{\T}\Sigma_{1}^{-1}B_1H_{2,1}^{\T}\nu_11_{H_{2,1}^{\T}\beta_{2,1}^* > 0}\right),
\end{eqnarray}
where $D_{p_1\times p_{20}}$ is the space of $p_1\times p_{20}$ matrices equipped with the spectral norm, and
$\mathcal{F}_{11} = \Big\lbrace
f(b_1, y_1,h_{2,0},h_{2,1}) =
a_1^{\T}b_1\big(y_1 + h_{2,0}^{\T}\beta_{2,0}+
\left[h_{2,1}^{\T}\beta_{2,1}\right]_+ - b_1^{\T}\beta_1\big)
+a_2^{\T}b_1(h_{2,1}^{\T}\nu_1)1_{h_{2,1}^{\T}\beta_{2,1}^* > 0},:\,
\beta = (\beta_1^{\T}, \beta_{2,0}^{\T},\beta_{2,1}^{\T})^{\T}\in\mathbb{R}^{p_1+p_2},
\nu = (\nu_0^{\T},\nu_1^{\T})^{\T}\in\mathbb{R}^{p_2},
a_1,a_2 \in\mathbb{R}^{p_{1}},
\max\{||a_1||, ||a_2||,||\beta||,||\nu||\}\leq K\Big\rbrace$.

\item $w_{12}:D_{p_1}\times l^{\infty}(\mathcal{F}_{12}) \times
  \mathbb{R}^{p_{21}} \times \mathbb{R}^{p_{21}} \rightarrow
  \mathbb{R}$ is defined as
\begin{equation}\label{w12}
w_{12}(\Sigma_1, \mu, \nu, \gamma) \triangleq \mu\left[c^{\T}\Sigma_{1}^{-1}
B_1\left(
\left[H_{2,1}^{\T}\nu + H_{2,1}^{\T}\gamma\right]_+
- \left[H_{2,1}^{\T}\gamma\right]_+
\right)1_{H_{2,1}^{\T}\beta_{2,1}^* = 0}
\right],
\end{equation}
where $\mathcal{F}_{12} = \Big\lbrace
f(b_1, h_{2,1}) = a^{\T}b_1\left([h_{2,1}^{T}\nu + h_{2,1}^{\T}\gamma]_+ -
[h_{2,1}^{\T}\gamma]_+\right)1_{h_{2,1}^{\T}\beta_{2,1}^*=0}
\,:a\in\mathbb{R}^{p_1}, \gamma,\,\nu\in\mathbb{R}^{p_{21}}, \max\{||a||, ||\nu||\}\le K
\Big\rbrace$.

\item $\rho_{11}:D_{p_1}\times D_{p_{21}}^{k}\times
  l^{\infty}(\mathcal{\widetilde F}_{11})\times \mathbb{R}^{p_{21}}\times
  \mathbb{R}^{p_{21}}\times \mathbb{R}^{p_{21}} \times \mathbb{R}
  \rightarrow \mathbb{R}$, is defined as
\begin{multline}\label{p11}
\rho_{11}(\Sigma_1, \Sigma_{21,21}, \mu, \nu, \eta, \gamma, \lambda)
\triangleq \mu\bigg[c^{\T}\Sigma_{1}^{-1}
B_1\left(\left[H_{2,1}^{\T}\nu + H_{2,1}^{\T}\gamma\right]_+
- [H_{2,1}^{\T}\gamma]_+\right)\\ \times \left(
1_{\frac{(H_{2,1}^{\T}\nu + H_{2,1}^{\T}\eta)^2}{H_{2,1}^{\T}\Sigma_{21,21}
H_{2,1}} \le \lambda} - 1_{H_{2,1}^{\T}\beta_{2,1}^* = 0}
\right)
\bigg],
\end{multline}
where $\mathcal{\widetilde F}_{11} = \Big\lbrace
f(b_1, h_{2,1}) = a^{\T}b_1\left([h_{2,1}^{\T}\nu - h_{2,1}^{\T}\gamma]_+ -
[h_{2,1}^{\T}\gamma]_+\right)(1_{\frac{(h_{2,1}^{\T}\nu + h_{2,1}^{\T}\eta)^2}
{h_{2,1}^{\T}\Sigma_{21,21}h_{2,1}} \le \lambda} - 1_{h_{2,1}^{\T}\beta_{2,1}^* = 0}
),:\, a\in\mathbb{R}^{p_1}, \nu, \eta, \gamma \in\mathbb{R}^{p_{21}}, \max\{||a||, ||\nu||\}\leq K, \lambda\in\mathbb{R},
\Sigma_{21,21}\in D_{p_{21}}^{k}\Big\rbrace$.

\item $\rho_{12}: D_{p_1}\times l^{\infty}(\mathcal{\widetilde F}_{12})\times \mathbb{R}^{p_{21}}
\times \mathbb{R}^{p_{21}} \rightarrow \mathbb{R}$, defined as
\begin{multline}\label{p12}
\rho_{12}(\Sigma_1, \mu, \nu, \eta) \triangleq
\mu\bigg[c^{\T}\Sigma_{1}^{-1}
B_1\left(\left[H_{2,1}^{\T}\nu + H_{2,1}^{\T}\eta\right]_+
-[H_{2,1}^{\T}\eta]_+ - H_{2,1}^{\T}\nu
\right)1_{H_{2,1}^{\T}\beta_{2,1}^* > 0} \\
+ c^{\T}\Sigma_{1}^{-1}B_1\left(\left[H_{2,1}^{\T}\nu + H_{2,1}^{\T}\eta\right]_+ - [H_{2,1}^{\T}\eta]_+\right)
1_{H_{2,1}^{\T}\beta_{2,1}^* < 0}\bigg],
\end{multline}
where $\mathcal{\widetilde F}_{12} = \Big\lbrace
a^{\T}b_1\big(\left[h_{2,1}^{\T}\nu + h_{2,1}^{\T}\eta\right]_+
  -[h_{2,1}^{\T}\eta]_+ - h_{2,1}^{\T}\nu
\big)1_{h_{2,1}^{\T}\beta_{2,1}^* > 0} - a^{\T}b_1\big(\left[h_{2,1}^{\T}\nu +
  h_{2,1}^{\T}\eta\right]_+ - [h_{2,1}^{\T}\eta]_+\big)
1_{h_{2,1}^{\T}\beta_{2,1}^* < 0}\,:\, a\in\mathbb{R}^{p_1}, \nu\in\mathbb{R}^{p_{21}}, \max\{||a||, ||\nu||\}\leq K, \eta\in
\mathbb{R}^{p_{21}} \Big\rbrace$.
\end{enumerate}
Using the foregoing functions, we have the following expressions for
the first stage parameters:
\begin{eqnarray}\label{firstStageP}
c^{\T}\rtn(\bHat_1 - \beta_1^*) &=&
w_{11}(\hat{\Sigma}_1, \hat{\Sigma}_{12}, \mathbb{G}_n,
\pn, \rtn(\bHat_{2}-\beta_{2}^*), (\beta_{1}^{*\T},\beta_{2}^{*\T})^{\T})  \nonumber \\
&& +\,\, w_{12}(\hat{\Sigma}_1, \pn, \rtn(\bHat_{2,1} - \beta_{2,1}^*),
\rtn\beta_{2,1}^*)  \nonumber \\
&& +\,\, \rho_{12}(\hat{\Sigma}_{1}, \pn, \rtn(\bHat_{2,1} - \beta_{2,1}^*),
\rtn\beta_{2,1}^*); \\
\rtn(\bHat_1 - \beta_{1,n}^*) &=& w_{11}(\hat{\Sigma}_1, \hat{\Sigma}_{12},
 \rtn(\pn-P_n), \pn, \rtn(\bHat_{2} - \beta_{2,n}^*),
(\beta_{1,n}^{*\T},\beta_{2,n}^{*\T})^{\T}) \nonumber \\
&& +\,\, w_{12}(\hat\Sigma_{1}, \pn, \rtn(\bHat_{2,1}-\beta_{2,1,n}^*),
\rtn\beta_{2,1,n}^*) \nonumber \\
&& + \,\rho_{12}(\hat{\Sigma}_1, \pn, \rtn(\bHat_{2,1}-\beta_{2,1,n}^*),
\rtn\beta_{2,1,n}^*),
\end{eqnarray}
where $\hat\Sigma_{12} = \mathbb{P}_n B_1H_{2,0}^{\T}$.
Similarly, we can express the upper bound $\mathcal{U}(c)$ in terms
of the above functions:
\begin{eqnarray}
\mathcal{U}(c) &=& w_{11}(\hat{\Sigma}_1, \hat{\Sigma}_{12}, \mathbb{G}_n,
\pn, \rtn(\bHat_{2}-\beta_{2}^*), (\beta_{1}^{*\T},\beta_{2}^{*\T})^{\T})
\nonumber \\
&&+\,\, \rho_{12}(\hat{\Sigma}_{1}, \pn, \rtn(\bHat_{2,1}-\beta_{2,1}^*),
\rtn\beta_{2,1}^*) \nonumber \\
&&-\,\, \rho_{11}(\hat{\Sigma}_{1}, \hat{\Sigma}_{21,21}, \pn,
\rtn(\bHat_{2,1}-\beta_{2,1}^*), \rtn\beta_{2,1}^*, \rtn\beta_{2,1}^*, \lambda_{n}) \nonumber \\
&&+\,\,
\sup_{\gamma\in\mathbb{R}^{p_{2,1}}}\bigg\lbrace w_{12}(\hat{\Sigma}_{1}, \pn,
\rtn(\bHat_{2,1}- \beta_{2,1}^*), \gamma)\nonumber \\ &&\quad +\,\,
\rho_{11}(\hat{\Sigma}_{1}, \hat{\Sigma}_{21,21}, \pn,
\rtn(\bHat_{2,1}-\beta_{2,1}^*), \rtn\beta_{2,1}^*, \gamma, \lambda_{n})
\bigg\rbrace. \label{eqn:upperbound}
\end{eqnarray}
We will also make use of the following alternative expression for the
upper bound $\mathcal{U}(c)$ under $P_n$:
\begin{eqnarray}
\mathcal{U}(c) &=& w_{11}(\hat{\Sigma}_1, \hat{\Sigma}_{12},
 \rtn(\pn-P_n), \pn, \rtn(\bHat_{2} - \beta_{2,n}^*),
(\beta_{1,n}^{*\T},\beta_{2,n}^{*\T})^{\T}) \nonumber \\ &&+\,\,
\rho_{12}(\hat{\Sigma}_{1}, \pn, \rtn(\bHat_{2,1} - \beta_{2,1,n}^*),
\rtn\beta_{2,1,n}^*)\nonumber
\\ &&-\,\,
\rho_{11}(\hat{\Sigma}_1, \hat{\Sigma}_{21,21},
\pn,\rtn(\bHat_{2,1}-\beta_{2,1,n}^*),
\rtn\beta_{2,1,n}^*, \rtn\beta_{2,1,n}^*, \lambda_{n})\nonumber
\\ &&+\,\,
\sup_{\gamma\in\mathbb{R}^{p_{21}}}\bigg\lbrace
w_{12}(\hat{\Sigma}_{1}, \pn,
\rtn(\bHat_{2,1}- \beta_{2,1,n}^*), \gamma)\nonumber \\ &&\quad +\,\,
\rho_{11}(\hat{\Sigma}_{1}, \hat{\Sigma}_{21,21}, \pn,
\rtn(\bHat_{2,1}-\beta_{2,1,n}^*), \rtn\beta_{2,1,n}^*, \gamma, \lambda_{n})
\bigg\rbrace.
\end{eqnarray}
Similarly, we will make use of following expression for the bootstrap
analog of the upper bound:
\begin{eqnarray}
\hat{\mathcal{U}}^{(b)}(c) &=& w_{11}(\hat{\Sigma}_1^{(b)}, \hat{\Sigma}_{12}^{(b)},
 \rtn(\mathbb{P}_n^{(b)}-\mathbb{P}_n), \mathbb{P}_n^{(b)}, \rtn(\bBoot_{2} - \bHat_{2}),
(\bHat_{1}^{\T},\bHat_{2}^{\T})^{\T})  \nonumber \\ &&+\,\,
\rho_{12}(\hat{\Sigma}_{1}^{(b)}, \pHat, \rtn(\bBoot_{2,1} - \bHat_{2,1}),
\rtn\bHat_{2,1})\nonumber
\\ &&-\,\,
\rho_{11}(\hat{\Sigma}_1^{(b)}, \hat{\Sigma}_{21,21}^{(b)},
\pHat,\rtn(\bBoot_{2,1}-\bHat_{2,1}),
\rtn\bHat_{2,1}, \rtn\bHat_{2,1}, \lambda_{n})\nonumber
\\ &&+\,\,
\sup_{\gamma\in\mathbb{R}^{p_{21}}}\bigg\lbrace
w_{12}(\hat{\Sigma}_{1}^{(b)}, \pHat,
\rtn(\bBoot_{2,1}- \bHat_{2,1}), \gamma)\nonumber \\ &&\quad +\,\,
\rho_{11}(\hat{\Sigma}_{1}^{(b)}, \hat{\Sigma}_{21,21}^{(b)}, \pHat,
\rtn(\bBoot_{2,1}-\bHat_{2,1}), \rtn\bHat_{2,1}, \gamma, \lambda_{n})
\bigg\rbrace. \label{eqn:bsbound}
\end{eqnarray}
The lower bound $\mathcal{L}(c)$ and its bootstrap analog
$\hat{\mathcal{L}}^{(b)}(c)$ can be expressed in a similar fashion by
replacing the $\sup$ with an $\inf$ in the expression of $\mathcal{U}(c)$
and $\hat{\mathcal{U}}^{(b)}(c)$, respectively.

By Lemmas \ref{lem:rho11neg} and \ref{lem:stage1cont} below, $\rho_{11}$ is negligible, and $w_{11}$ and $w_{12}$ are continuous at desired points. The negligibility of $\rho_{12}$ can be obtained in a similar fashion.
Note that the convergence of $\hat\Sigma_1$ and $\hat\Sigma_1^{(b)}$ to $\Sigma_1$ and the convergence of $\hat\Sigma_{12}$ and $\hat\Sigma_{12}^{(b)}$ to $PB_1H_{2,0}^{\T}$ can be obtained using similar proof techniques as in Lemma \ref{sigma2_conv}.
This together with Theorem \ref{stage2}, Lemmas \ref{sigma2_conv} - \ref{closedset}, and the continuous mapping theorems as presented in the previous section, implies that the conclusions of Theorems \ref{thm:stage1} and \ref{thm:ACI} hold with
\begin{align*}
\mathbb{S}_\infty = &\, \Sigma_{1,\infty}^{-1}\big[\mathbb{G}_\infty\big(B_1(Y_1+H_{2,0}^{\T}\beta_{2,0}^*+[H_{2,1}^{\T}\beta_{2,1}^*]_+-B_1^{\T}\beta_1^*)\big) + PB_1H_{2,0}^{\T}\mathbb{Z}_{\infty,0}\big]\\
\mbox{and }\mathbb{V}_\infty =  &\, \mathbb{Z}_{\infty,1},
\end{align*}
where $\mathbb{Z}^{\T}_{\infty,0}\in\mathbb{R}^{p_{20}}$, $\mathbb{Z}^{\T}_{\infty,1}\in\mathbb{R}^{p_{21}}$, and $\mathbb{Z}_\infty = (\mathbb{Z}^{\T}_{\infty,0}, \mathbb{Z}^{\T}_{\infty,1})^{\T}=\Sigma_{2,\infty}^{-1}\mathbb{G}_\infty[B_2(Y_2-B_2^{\T}\beta_2^*)]$.


\begin{lem}
\label{lem:rho11neg}
Assume (A1), (A2) and (A4). Then
\begin{enumerate}
  \item $\sup_{\gamma\in\mathbb{R}^{p_{21}}}\big|\rho_{11}(\hat{\Sigma}_1, \hat{\Sigma}_{21,21}, \pn,
\rtn(\bHat_{2,1} - \beta_{2,1}^*), \rtn\beta_{2,1}^*, \gamma, \lambda_{n})| \rightarrow_{P} 0$, and
\item
  $\sup_{\gamma\in\mathbb{R}^{p_{21}}}\big|\rho_{11}(\hat{\Sigma}_{1}^{(b)},
  \hat{\Sigma}_{21,21}^{(b)},\pHat,
  \rtn(\bHat_{2,1}^{(b)}-\bHat_{2,1}), \rtn\bHat_{2,1}, \gamma,
  \lambda_{n})\big| \rightarrow_{P_M} 0$ almost surely $P$.
\end{enumerate}
If, in addition, we assume (A3), then
\begin{enumerate}
\item[3.] $\sup_{\gamma\in\mathbb{R}^{p_{21}}}\big|\rho_{11}(\hat{\Sigma}_{1}, \hat{\Sigma}_{21,21}, \pn,
\rtn(\bHat_{2,1}-\beta_{2,1,n}^*), \rtn\beta_{2,1,n}^*, \gamma,
\lambda_{n})\big| \rightarrow_{P_{n}} 0$.
\end{enumerate}
\end{lem}
\begin{proof}
First it is easy to verify that $|[H_{2,1}^{\T}\nu - H_{2,1}^{\T}\gamma]_+ -
[H_{2,1}^{\T}\gamma]_+|\leq |h_{2,1}^{\T}\nu|$. Thus for any probability measure $\mu$ in
$l^{\infty}(\mathcal{\widetilde F}_{11})$,
\begin{multline*}
|\rho_{11}(\Sigma_1, \Sigma_{21,21}, \mu, \nu,\eta,\gamma,\lambda)| \le
K\bigg\lbrace \mu\left(||B_1||\,||H_{2,1}||\,
1_{H_{2,1}^{\T}\beta_{2,1}^* = 0, \frac{H_{2,1}^{\T}\eta}{||H_{2,1}||} > \sqrt{\lambda k} - K}\right)
 \\
+ \mu\left(||B_1||\,||H_{2,1}||\,
1_{H_{2,1}^{\T}\beta_{2,1}^* = 0, \frac{H_{2,1}^{\T}\eta}{||H_{2,1}||} < -\sqrt{\lambda k} - K}\right) \\
+ \mu\left(||B_1||\,||H_{2,1}||\,
1_{H_{2,1}^{\T}\beta_{2,1}^* \ne 0, -\sqrt{\lambda/k} - K \le \frac{H_{2,1}^{\T}\eta}{||H_{2,1}||} \le \sqrt{\lambda/k} + K}\right)
\bigg\rbrace
\end{multline*}
for a sufficiently large constant $K>0$ and a sufficiently small constant $k\in(0,1)$.  Since $k$ is held constant there is no loss
in generality taking $k=1$.  Define $\rho_{11}': l^{\infty}(\mathcal{F}_{11}') \times \mathbb{R}^{p_{21}} \times
\mathbb{R}\times \mathbb{R}\to\mathbb{R}$
as
\begin{multline}
\rho_{11}'(\mu, \eta, \delta, \delta') =
 \mu\left(||B_1||\,||H_{2,1}||\,
1_{H_{2,1}^{\T}\beta_{2,1}^* = 0, \frac{H_{2,1}^{\T}\eta}{||H_{2,1}||} >  \delta}\right)\\
+ \mu\left(||B_1||\,||H_{2,1}||\,
1_{H_{2,1}^{\T}\beta_{2,1}^* = 0, \frac{H_{2,1}^{\T}\eta}{||H_{2,1}||} < \delta'}\right) \\
+ \mu\left(||B_1||\,||H_{2,1}||\,
1_{H_{2,1}^{\T}\beta_{2,1}^* \ne 0, \delta' \le  \frac{H_{2,1}^{\T}\eta}{||H_{2,1}||} \le -\delta'}\right), \label{rho11'}
\end{multline}
where
$\mathcal{F}_{11}' = \bigg\lbrace f(b_1, h_{2,1}) =
||b_1||\,||h_{2,1}|| 1_{h_{2,1}^{\T}\beta_{2,1}^* = 0,
  \frac{h_{2,1}^{\T}\eta}{||h_{2,1}||} > \delta} +
||b_1||\,|||h_{2,1}||1_{h_{2,1}^{\T}\beta_{2,1}^* = 0,
  \frac{H_{2,1}^{\T}\eta}{||h_{2,1}||} < \delta'} +
||b_1||\,||h_{2,1}||1_{h_{2,1}^{\T}\beta_{2,1}^* \ne 0, \delta' \le
  \frac{h_{2,1}^{\T}\eta}{||h_{2,1}||} \le -\delta'}, \eta\in\mathbb{R}^{p_{21}},\max\{||\eta||,||\delta||, ||\delta'||\}\leq K\bigg\rbrace.$ Then
$$|\rho_{11}(\Sigma_1, \Sigma_{21,21}, \mu, \nu,\eta,\gamma,\lambda)|\leq K\rho_{11}'\left(\mu, \eta/\sqrt n, (\sqrt\lambda - K)/\sqrt n, -(\sqrt\lambda +K)/\sqrt n\right)$$
for $\mu\in l^\infty(\mathcal{\widetilde F}_{11})$. In particular for $n$ sufficiently large,
\begin{multline*}
\big|
\rho_{11}(\hat{\Sigma}_{1}^{(b)}, \hat{\Sigma}_{21,21}^{(b)}, \pHat,
\rtn(\bBoot_{2,1} - \bHat_{2,1}), \rtn\bHat_{2,1}, \gamma, \lambda_{n})
\big| \le \\ \,
K\,\rho_{11}'\left(\pHat, \bHat_{2,1},
(\sqrt{\lambda_{n}} - K)/\rtn,
-(\sqrt{\lambda_{n}} -K)/\rtn\right)\\ \, +
||c||\,||\hat{\Sigma}_{1}^{(b)}||\,||\rtn(\bBoot_{2,1}-\bHat_{2,1})||\,
\pHat \left(||B_1||\,||H_{2,1}||\right) 1_{||\rtn(\bBoot_{2,1} - \bHat_{2,1})||> K},
\end{multline*}
where we have assumed, without loss of generality, that
$\hat{\Sigma}_{21,21}^{(b)}$ is the identity matrix.
By part 2 of Lemma \ref{rho11_conv} below, we see that the first term on the right
hand side of the above display is $o_{P_M}(1)$
almost surely $P$. To deal with the second term, for any $\epsilon, \delta > 0$, let $K$ sufficiently large so that
$P_{M}\left(\big|\big|\rtn(\bBoot_{2,1} - \bHat_{2,1})\big|\big| > K\right)
< \delta$ for sufficiently large $n$ for almost all sequences $P$.
Then
\begin{multline*}
P_{M}\left(
||c||\,||\hat{\Sigma}_{1}^{(b)}||\,||\rtn(\bBoot_{2,1}-\bHat_{2,1})||
\pHat ||B_1||\,||H_{2,1}|| 1_{|\rtn(\bBoot_{2,1} - \bHat_{2,1})| > K} > \epsilon
\right) \\ \, \le P_{M}\left(\big|\big|\rtn(\bBoot_{2,1}-\bHat_{2,1})\big|\big|
> K\right) \le \delta,
\end{multline*}
almost surely $P$. This completes the proof of result 2. Similar arguments can be used to prove results 1 and 3, and are omitted.
\end{proof}

\begin{lem}
\label{rho11_conv}
Let $\rho_{11}'$ be defined in (\ref{rho11'}). Assume (A1), (A2) and (A4), then
\begin{enumerate}
  \item $\rho_{11}'(\pn, \beta_{2,1}^*, (\sqrt{\lambda_n} - K)/\rtn,
(-\sqrt{\lambda_{n}}-K)/\rtn) \rightarrow_{P} 0$, and
\item $\rho_{11}'(\pHat, \bHat_{2,1}, (\sqrt{\lambda_{n}}-K)/\rtn,
(-\sqrt{\lambda_{n}}-K)/\rtn) \rightarrow_{P_M} 0$, $P$-almost surely.
\end{enumerate}
If, in addition, we assume (A3), then
\begin{enumerate}
\item[3.] $\rho_{11}'(P_n, \beta_{2,1,n}^*, (\sqrt{\lambda_{n}}-K)/\rtn,
(-\sqrt{\lambda_{n}}-K)/\rtn) \rightarrow_{P_n} 0$.
\end{enumerate}
\end{lem}
\begin{proof}
 The class
$\mathcal{F}_{11}'$ is $P$-Donsker and measurable by Theorem 8.14 in
\cite{ab-nnltf-99} and Donkser preservation results (for example,
see Theorem 2.10.6 in \citealt{van1996weak}).  Note
that by (A1) and (A3) $\sup_{f\in\mathcal{F}_{11}'}|Pf^2| < \infty$
and $\sup_{f\in\mathcal{F}_{11}'}|P_nf^2|$ is a bounded sequence.  Thus,
it follows that (i) $||\pn -P|| \rightarrow 0$ almost surely under $P$
in $l^{\infty}(\mathcal{F}_{11}')$,
(ii) $||\pHat - P|| \rightarrow 0$ almost surely
$P_M$ for almost all sequences $P$ [Lemma 3.6.16 in \citealt{van1996weak}], and (iii) $||\pn -P_{n}|| \rightarrow 0$ almost surely under $P_n$ in
$l^{\infty}(\mathcal{F}_{11}')$ [Theorem 3.10.12 in \citealt{van1996weak}].  Additionally, the argument in the proof of Lemma (\ref{sigma2_conv})
shows that $\hat{\Sigma}_1$ is convergent to $\Sigma_{1}$ under
$P_n$, and the weak law of large numbers establishes convergence
under $P$.  The bootstrap strong law shows that
$\hat{\Sigma}_{1}^{(b)}$ converges to $\Sigma_{1}$ in $P_M$ probability
for almost all sequences $P$.

Next we show that $\rho_{11}'$ is continuous at the point $(P,
\beta_{2,1}^*, 0, 0)$. Let $\mu_{n} \rightarrow P$ in
$l^{\infty}(\mathcal{F}_{11}')$, $\eta_{n} \rightarrow \beta_{2,1}^*$,
$\delta_{n} \rightarrow 0$, and $\delta_{n}' \rightarrow 0$.
We have
\begin{equation*}
  \big|
  \rho_{11}'(\mu_{n}, \eta_{n}, \delta_{n}, \delta_{n}') -
  \rho_{11}'(P, \beta_{2,1}^*, 0, 0)
  \big| \le \big|\rho_{11}'(P, \eta_{n}, \delta_{n}, \delta_{n}') -
  \rho_{11}'(P, \beta_{2,1}^*, 0, 0)\big|
  + ||\mu_{n} - P||,
\end{equation*}
which converges to zero by the dominated convergence theorem.
The results follow from the continuous mapping theorems and
 the fact that
$\rho_{11}'(P, \beta_{2,1}^*, 0, 0) = 0$.
\end{proof}

\begin{lem}
\label{lem:stage1cont}
Assume (A1) and (A2). Then
\begin{enumerate}
\item $w_{11}$ is continuous at
points in $(\Sigma_{1,\infty}, \Sigma_{12,\infty}, C_{b}(\mathcal{F}_{11}),
P, \mathbb{R}^{p_{2}}, (\beta_1^{*\T},\beta_2^{*\T})^{\T})$;
\item $w_{12}(\cdot,\cdot,\cdot,\sqrt n\beta_{2,1}^*)$ and
$w_{12}(\cdot,\cdot,\cdot,\sqrt n\beta_{2,1,n}^*)$ are continuous at points in
$(\Sigma_{1,\infty}, P, \mathbb{R}^{p_{21}})$; and
\item $w_{12}'(\Sigma_1, \mu, \nu) \triangleq \sup_{\gamma\in\mathbb{R}^{p_{21}}}w_{12}(\Sigma_1, \mu, \nu, \gamma)$ is continuous at points in
$(\Sigma_{1,\infty}, P, \mathbb{R}^{p_{21}})$.
\end{enumerate}
\end{lem}
\begin{proof}
To prove the desired continuity of $w_{12}$ and $w_{12}'$, we will establish the
stronger result that $w_{12}$ is continuous at points
$(\Sigma_{1,\infty}, P, \mathbb{R}^{p_{21}},\gamma)$ uniformly in $\gamma$.  That is,
for any $\Sigma_{n}\rightarrow \Sigma_{1,\infty}$, probability measures $\mu_{n} \rightarrow P$ and
$\nu_{n} \rightarrow \nu$, we have
\begin{equation*}
\sup_{\gamma}\bigg|
w_{12}(\Sigma_{n}, \mu_{n}, \nu_{n}, \gamma) -
w_{12}(\Sigma_{1}, P, \nu, \gamma)
\bigg| \rightarrow 0.
\end{equation*}
Note that
\begin{multline*}
\big|
w_{12}(\Sigma_{n}, \mu_{n}, \nu_{n}, \gamma) -
w_{12}(\Sigma_{1}, P, \nu, \gamma)
\big| \\
\leq \big|
w_{12}(\Sigma_{n}, \mu_{n}, \nu_{n}, \gamma) - w_{12}(\Sigma_{n}, \mu_n, \nu, \gamma)\big|
+ \big|w_{12}(\Sigma_{n}, P, \nu, \gamma)-w_{12}(\Sigma_{1}, P, \nu, \gamma)\big|\\
+ \big| w_{12}(\Sigma_{n}, \mu_n, \nu, \gamma) -  w_{12}(\Sigma_{n}, P, \nu, \gamma)\big|\\
\leq \mu_n\left(\left|c^{\T}\Sigma_n^{-1}B_1 |H_{2,1}^{\T}(\nu_n-\nu)|\right|\right)
+P\left(|c^{\T}(\Sigma_n^{-1}-\Sigma_{1,\infty}^{-1})B_1|\,|H_{2,1}^{\T}\nu|\right)\\
+\left|(\mu_n-P)\left(c^{\T}\Sigma_n^{-1}B_1([H_{2,1}^{\T}\nu + H_{2,1}^{\T}\gamma]_+ - [H_{2,1}^{\T}\gamma]_+)1_{H_{2,1}^{\T}\beta_{2,1}^*=0}\right)\right|
\end{multline*}
By (A2), we have that $||\Sigma_{n}^{-1}||$ is bounded above for sufficiently large $n$, where $||\cdot||$ of a matrix denotes the spectral norm of the matrix.
Thus the first term in the above display is bounded by $||c||\,||\Sigma_{n}^{-1}||\mu_n(||B_1||\,||H_{2,1}||)\,||\nu_{n}-\nu|| = o(1)$, and the second term in the above display is bounded by
$||c||\,||\Sigma_{1}^{-1} - \Sigma_{n}^{-1}||\,P(||B_1||\,||H_{2,1}||)||\nu||=o(1)$.
For the third term,
note that if $||\nu||=0$, then it is zero. Otherwise,
\begin{multline*}
\left|(\mu_n-P)\left(c^{\T}\Sigma_n^{-1}B_1([H_{2,1}^{\T}\nu + H_{2,1}^{\T}\gamma]_+ - [H_{2,1}^{\T}\gamma]_+)1_{H_{2,1}^{\T}\beta_{2,1}^*=0}\right)\right|\\
\leq\left|(\mu_n-P)\left(c^{\T}\Sigma_n^{-1}B_1([H_{2,1}^{\T}\nu/||\nu|| + H_{2,1}^{\T}\gamma/||\nu||]_+ - [H_{2,1}^{\T}\gamma/||\nu||]_+)1_{H_{2,1}^{\T}\beta_{2,1}^*=0}\right)\right|||\nu||\\
\leq ||\mu_n-P||_{\mathcal{F}_{12}}||\nu||=o(1).
\end{multline*}
This established the continuity of $w_{12}$ and hence $w_{12}'$. The
continuity of $w_{11}$ can be established through similar arguments and is
therefore omitted.
\end{proof}

\section{Appendix: Definitions of Three-Treatment Models}\label{ap:threetxt}
\newcommand{\mcoeff}{\ensuremath{\xi}}
Here, we present a suite of example models similar to those
of Chakraborty et al. (2009), but that have three possible treatments at the second
stage. These models are defined as follows:
\begin{itemize}
\item $X_i \in \{-1,1\}$ for $i \in \{1,2\}$,  $A_1 \in \{-1,1\}$, and
  $A_2 \in \{(0,-0.5)^\T,(-1,0.5) ^\T,(1,0.5)^\T\}$
\item $P(A_1= 1) = P(A_1 = -1) = 1/2$,\\ $P(A_2=(0,-1)^\T) = P(A_2
  =(-1,0.5)^\T) = P(A_2 = (1,0.5)^\T) = 1/3$
\item $P(X_1=1) = P(X_1=-1) = 1/2$, $P(X_2=1|X_1,A_1) = \mathrm{expit}({\delta_1 X_1 + \delta_2 A_1})$
\item $Y_1 \triangleq 0$, \\$Y_2 = \mcoeff_1 + \mcoeff_2 X_1 + \mcoeff_3
  A_1 + \mcoeff_4 X_1 A_1 + (\mcoeff_5,\mcoeff_6)A_2 + X_2(\mcoeff_7,\mcoeff_8)A_2 + A_1 (\mcoeff_9,\mcoeff_{10})A_2 + \epsilon$,
$\epsilon \sim N(0,1)$
\end{itemize}
where $\mathrm{expit}(x) = \mathrm{e}^{x} / (1 + \mathrm{e}^x)$. This
class is parameterized by twelve values
$\mcoeff_1,\mcoeff_2,...,\mcoeff_{10},\delta_1,\delta_2$. The analysis model
uses histories defined by:
\begin{eqnarray}
  H_{2,0} & =& (1, X_1, A_1,X_1 A_1, X_2)^\T\\
  H_{2,1} & =& (1, X_2, A_1)^\T\\
  H_{1,0} & =& (1, X_1)^\T\\
  H_{1,1} & =& (1, X_1)^\T.
\end{eqnarray}
Our working models are given by $Q_2(H_2, A_2; \beta_2) \triangleq
H_{2,0}^{\T}\beta_{2,0} + H_{2,1}^{\T}\beta_{2,1,1} A_{2,1} + H_{2,1}^{\T}\beta_{2,1,2} A_{2,2}$ and
$Q_1(H_1, A_1; \beta_1) \triangleq
H_{1,0}^{\T}\beta_{1,0} + H_{1,1}^{\T}\beta_{1,1}A_1$.
\begin{table}[h!]
\begin{tabular}{c|c|c|c}\label{tb:threetreatments}
Example & $\mcoeff$ & $\delta$ & Regularity\\
\hline
1 & ${(0, 0, 0, 0, 0, 0, 0, 0, 0, 0)}^\T$ & $(0.5, 0.5)^\T$ &
$p = 1, \phi = 0/0$\\
2 & ${(0, 0, 0, 0, 0.01, 0.01, 0, 0, 0, 0)}^\T$ & $(0.5, 0.5)^\T$  &
$p = 0, \phi = \infty$\\
3 & ${(0, 0, -0.5, 0, 0.5, 0.5, 0, 0, 0.5, 0.5)}^\T$ & $(0.5, 0.5)^\T$ &
$p = 1/2, \phi = 1.0$\\
4 & ${(0, 0, -0.5, 0, 0.5, 0.5, 0, 0, 0.49, 0.49)}^\T$ & $(0.5, 0.5)^\T$ &
$p = 0, \phi = 1.0204$\\
5 & ${(0, 0, -0.5, 0, 1.00, 1.00, 0.5, 0.5, 0.5, 0.5)}^\T$ & $(1.0, 0.0)^\T$ &
$p = 1/4, \phi = 1.4142$\\
6 & ${(0, 0, -0.5, 0, 0.25, 0.25, 0.5, 0.5, 0.5, 0.5)}^\T$ & $(0.1, 0.1)^\T$ &
$p = 0, \phi = 0.3451$\\
\hline
A & ${(0, 0, -0.25, 0, 0.75, 0.75, 0.5, 0.5, 0.5, 0.5)}^\T$ &
$(0.1, 0.1)^\T$ & $p = 0, \phi = 1.035$ \\
B & ${(0, 0, 0, 0, 0.25, 0.25, 0, 0, 0.25, 0.25)}^\T$ & $(0, 0)^\T$ &
$p = 1/2, \phi = 1.00$ \\
C & ${(0, 0, 0, 0, 0.25, 0.25, 0, 0, 0.24, 0.24)}^\T$ & $(0, 0)^\T$ &
$p = 1/2, \phi = 1.00$
\end{tabular}
\caption{\label{tb:modelparams_2act_2stage}  Parameters indexing
the example models.}
\end{table}
In Table~\ref{tb:threetreatments}, for each of these models we give
the probability $p$ of generating a history where each of the three
possible treatments at the second stage have exactly the same
effect. This is analogous to having the second stage treatment show no
effect in a binary model. Furthermore, because of the
Helmert encoding we have used in our analysis models, and because of
the structure of $\mcoeff$, it happens that the
standardized effect size of treatment 1 versus treatment 2, treatment 1
versus treatment 3, and treatment 2 versus treatment 3 are all
exactly equal in our examples. We report this as $\phi$ in
Table~\ref{tb:threetreatments}.

\section{Appendix: Additional Empirical Results}

Here we present additional empirical results. Tables
(\ref{tb:lambdan_sensitivity_coverage}) and
(\ref{tb:lambdan_sensitivity_width}) show the estimated coverage and
interval diameter of the ACI across the nine generative models with
two stages and two treatments per stage.  The results appear stable
across choices of $\lambda_n$ for which the ACI is consistent.
However, the ACI becomes quite conservative when $\lambda_n$ is
allowed to grow faster than $\sqrt{\log\,\log\,n}$.

\begin{table}
\begin{small}
\centering
\begin{tabular}{cllllll|lll}
\parbox{5em}{\centering{$\beta_{1,1,1}$\\$\lambda_n =$}} &
\parbox{3em}{\centering{Ex. 1\\NR}} &
\parbox{3em}{\centering{Ex. 2\\NNR}} &
\parbox{3em}{\centering{Ex. 3\\NR}} &
\parbox{3em}{\centering{Ex. 4\\NNR}} &
\parbox{3em}{\centering{Ex. 5\\NR}} &
\parbox{3em}{\centering{Ex. 6\\R}} &
\parbox{3em}{\centering{Ex. A\\R}} &
\parbox{3em}{\centering{Ex. B\\R}} &
\parbox{3em}{\centering{Ex. C\\R}}
\vspace{0.15em} \\
\hline\noalign{\smallskip}
        $\sqrt{\log\log n}$ &  0.989  &  0.987  &  0.967  &  0.969  &  0.954  &  0.952  &  0.950  &  0.962  &  0.962  \\
               $\log\log n$ &  0.992  &  0.992  &  0.968  &  0.972  &  0.957  &  0.955  &  0.950  &  0.964  &  0.965  \\
                   $\log n$ &  0.993  &  0.994  &  0.975  &  0.976  &  0.962  &  0.966  &  0.959  &  0.969  &  0.972  \\
                  $\sqrt n$ &  0.994  &  0.995  &  0.975  &  0.976  &  0.967  &  0.972  &  0.968  &  0.973  &  0.975  \\
                        $n$ &  0.994  &  0.995  &  0.975  &  0.976  &  0.969  &  0.972  &  0.968  &  0.975  &  0.976
\end{tabular}\\
\vskip1em
\begin{tabular}{cllllll|lll}
\parbox{5em}{\centering{$\beta_{1,0,1}$\\$\lambda_n =$}} &
\parbox{3em}{\centering{Ex. 1\\NR}} &
\parbox{3em}{\centering{Ex. 2\\NNR}} &
\parbox{3em}{\centering{Ex. 3\\NR}} &
\parbox{3em}{\centering{Ex. 4\\NNR}} &
\parbox{3em}{\centering{Ex. 5\\NR}} &
\parbox{3em}{\centering{Ex. 6\\R}} &
\parbox{3em}{\centering{Ex. A\\R}} &
\parbox{3em}{\centering{Ex. B\\R}} &
\parbox{3em}{\centering{Ex. C\\R}}
\vspace{0.15em} \\
\hline\noalign{\smallskip}
        $\sqrt{\log\log n}$ &  0.952  &  0.962  &  0.952  &  0.954  &  0.950  &  0.953  &  0.947  &  0.952  &  0.954  \\
               $\log\log n$ &  0.956  &  0.964  &  0.954  &  0.955  &  0.950  &  0.957  &  0.948  &  0.956  &  0.957  \\
                   $\log n$ &  0.970  &  0.974  &  0.961  &  0.964  &  0.950  &  0.966  &  0.959  &  0.965  &  0.968  \\
                 $\sqrt{n}$ &  0.971  &  0.975  &  0.963  &  0.968  &  0.954  &  0.973  &  0.965  &  0.974  &  0.978  \\
                        $n$ &  0.971  &  0.975  &  0.987  &  0.987  &  0.979  &  0.980  &  0.975  &  0.983  &  0.984

\end{tabular}\\
\caption{\label{tb:lambdan_sensitivity_coverage} Monte Carlo estimates
  of coverage probabilities for the ACI method at the $95\%$ nominal
  level for different choices of $\lambda_n$. Here, $\beta_{1,1,1}$
  denotes the main effect of treatment and $\beta_{1,0,1}$ denotes the
  intercept.  Estimates are constructed using 1000 datasets of size
  150 drawn from each model, and 1000 bootstraps drawn from each
  dataset. No coverage estimates are significantly below $0.95$ at the
  $0.05$ level. Models have two treatments at each of two
  stages. Examples are designated NR = nonregular, NNR =
  near-nonregular, R = regular.  }
\end{small}
\end{table}
\begin{table}
\begin{small}
\centering
\begin{tabular}{cllllll|lll}
\parbox{5em}{\centering{$\beta_{1,1,1}$\\$\lambda_n =$}} &
\parbox{3em}{\centering{Ex. 1\\NR}} &
\parbox{3em}{\centering{Ex. 2\\NNR}} &
\parbox{3em}{\centering{Ex. 3\\NR}} &
\parbox{3em}{\centering{Ex. 4\\NNR}} &
\parbox{3em}{\centering{Ex. 5\\NR}} &
\parbox{3em}{\centering{Ex. 6\\R}} &
\parbox{3em}{\centering{Ex. A\\R}} &
\parbox{3em}{\centering{Ex. B\\R}} &
\parbox{3em}{\centering{Ex. C\\R}}
\vspace{0.15em} \\
\hline\noalign{\smallskip}
        $\sqrt{\log\log n}$ &  0.490  &  0.490  &  0.481  &  0.481  &  0.483  &  0.471  &  0.474  &  0.484  &  0.484  \\
               $\log\log n$ &  0.502  &  0.502  &  0.488  &  0.488  &  0.487  &  0.475  &  0.477  &  0.491  &  0.491  \\
                   $\log n$ &  0.557  &  0.557  &  0.518  &  0.518  &  0.503  &  0.495  &  0.492  &  0.523  &  0.523  \\
                 $\sqrt{n}$ &  0.583  &  0.582  &  0.533  &  0.533  &  0.513  &  0.514  &  0.511  &  0.540  &  0.540  \\
                        $n$ &  0.586  &  0.586  &  0.538  &  0.538  &  0.525  &  0.521  &  0.519  &  0.543  &  0.543
\end{tabular}\\
\vskip1em
\begin{tabular}{cllllll|lll}
\parbox{5em}{\centering{$\beta_{1,0,1}$\\$\lambda_n =$}} &
\parbox{3em}{\centering{Ex. 1\\NR}} &
\parbox{3em}{\centering{Ex. 2\\NNR}} &
\parbox{3em}{\centering{Ex. 3\\NR}} &
\parbox{3em}{\centering{Ex. 4\\NNR}} &
\parbox{3em}{\centering{Ex. 5\\NR}} &
\parbox{3em}{\centering{Ex. 6\\R}} &
\parbox{3em}{\centering{Ex. A\\R}} &
\parbox{3em}{\centering{Ex. B\\R}} &
\parbox{3em}{\centering{Ex. C\\R}}
\vspace{0.15em} \\
\hline\noalign{\smallskip}
        $\sqrt{\log\log n}$ &  0.506 &  0.506 &  0.481 &  0.481 &  0.483 &  0.490 &  0.474 &  0.490 &  0.490 \\
               $\log\log n$ &  0.518 &  0.518 &  0.487 &  0.487 &  0.486 &  0.494 &  0.476 &  0.497 &  0.498 \\
                   $\log n$ &  0.574 &  0.574 &  0.517 &  0.517 &  0.502 &  0.517 &  0.493 &  0.540 &  0.541 \\
                 $\sqrt{n}$ &  0.596 &  0.596 &  0.536 &  0.536 &  0.515 &  0.543 &  0.519 &  0.571 &  0.572 \\
                        $n$ &  0.598 &  0.598 &  0.576 &  0.576 &  0.565 &  0.586 &  0.565 &  0.579 &  0.579

\end{tabular}\\
\caption{\label{tb:lambdan_sensitivity_width}Monte Carlo estimates of
  mean width of the ACI method at the $95\%$ nominal level for
  different choices of $\lambda_n$.  Here, $\beta_{1,1,1}$ denotes the
  main effect of treatment and $\beta_{1,0,1}$ denotes the intercept.
  Estimates are constructed using 1000 datasets of size 150 drawn from
  each model, and 1000 bootstraps drawn from each dataset. No
  corresponding estimated coverages are significantly below $0.95$ at
  the $0.05$ level. Models have two treatments at each of two
  stages. Examples are designated NR = nonregular, NNR =
  near-nonregular, R = regular.}
 \end{small}
\end{table}

\section{Appendix: The double bootstrap algorithm for selecting $\lambda$}
Our algorithmic approach to choosing $\lambda_n$ is similar to that used
by \cite{mofn} to choose $m$ for their $m$-out-of-$n$ bootstrap
method. To select $\lambda_n$, we first draw $r$ bootstrapped datasets
$\calD^{(1)},...,\calD^{(r)}$ from the original dataset $\calD$. We
take each of these in turn and compute an ACI bootstrap confidence
interval at level $1 - \alpha$ with parameter $\lambda_n = \tau
\sqrt{\log\log n}$ for $\tau \in \{0.125,0.25,0.5,1,2,4\}$. (Because the ACI
uses the bootstrap itself, it actually uses double-bootstraps of
$\calD$ to compute each interval.) Using the parameters estimated by
Q-learning on the original $\calD$ as ground truth, we compute for
each value of $\tau$ the number of bootstrapped datasets $\kappa(\tau)$
for which the ACI covers. We then select $\tau^*$ to be the smallest
$\tau$ that satisfies $\kappa(\tau)/r > 1 - \alpha$, and apply the ACI
to the original dataset $\calD$ using $\lambda = \tau^* \sqrt{\log\log n}$. In our
experiments we used $r = 100$.

\end{document}